\documentclass[12pt]{article}
\usepackage{graphicx, amsmath,amsfonts,amsthm, euscript}
\usepackage{verbatim}
\usepackage{color}
\usepackage{ulem}
\usepackage{amsmath}
\usepackage{amscd}
\usepackage{amsfonts}
\usepackage{braket,amsfonts}
\usepackage{array}
\usepackage[caption=false]{subfig}
\usepackage{pgfplots}
\usepackage{graphicx,epstopdf}

\usepackage{graphicx}
\usepackage {tikz}
\usetikzlibrary {positioning}
\usetikzlibrary{shapes,arrows}

\newtheorem{theorem}{Theorem}[section]

\newtheorem{definition}{Definition}[section]
\newtheorem{lemma}{Lemma}[section]
\newtheorem{corollary}{Corollary}[section]
\newtheorem{remark}{Remark}[section]

 \numberwithin{equation}{section}
  \numberwithin{figure}{section}
\usepackage{graphicx, amsmath}
\usepackage{color}
\definecolor{Blue}{rgb}{0.3,0.3,0.9}

\usepackage{hyperref}
\usepackage{verbatim}
\usepackage{color}
\usepackage{ulem}
\usepackage{amsmath}
\usepackage{amscd}
\usepackage{amsfonts}
\usepackage{braket,amsfonts}
\usepackage{array}
\usepackage[caption=false]{subfig}
\usepackage{pgfplots}
\usepackage{graphicx,epstopdf}
\usepackage{graphicx}
\usepackage {tikz}
\usetikzlibrary {positioning}
\usetikzlibrary{shapes,arrows}
\usepackage{appendix}
\usepackage[super,sort,compress]{cite}
\begin{document}

\title{THE LOCALIZATION OF QUANTUM RANDOM WALKS\\
ON SIERPINSKI GASKETS}



\author{Kai Zhao and Wei-Shih Yang}
\maketitle


\vskip 3pt

Department of Mathematics

Temple University, Philadelphia, PA 19122

\vskip 6pt
Email: kai.zhao@temple.edu, yang@temple.edu

\vskip 12pt


\begin{abstract}
We consider the discrete time quantum random walks on a Sierpinski gasket. We study the hitting probability as the level of fractal goes to infinity in terms of their 
localization exponents $\beta_w$, total variation exponents $\delta_w$ and relative entropy exponents $\eta_w$. 
We define and solve the amplitude Green  functions recursively  when the level of the fractal graph goes to infinity. We obtain exact recursive formulas for the amplitude Green functions, based on which the hitting probabilities and expectation of the first-passage time are calculated. Using the recursive formula with the aid of Monte Carlo integration, we evaluate their numerical values. 
We also show  that when the level of the  fractal graph goes to infinity, with probability 1, the quantum random walks will return to origin, i.e., the quantum walks on Sierpinski gasket are recurrent. 
\end{abstract}

\markboth{Wei-Shih Yang and Kai Zhao}
{The localization of quantum random walks on Sierpinski gaskets}

\section{Introduction}

\setcounter{equation}{0}

\subsection{Classical theory}

\subsubsection{Classical walk, spectral and  geometric dimensions}

For classical random walks $(X_n)_{n=0}^{\infty}$ on a space, the walk dimension (or fractal dimension of the walk, diffusion exponent), $d_w$, is defined by the reciprocal of the exponent of the expected distance at time $n$ from the initial point $X_0$; namely, $E|X_n-X_0| \sim n^{1/d_{w}}$, as $n \to \infty$. For simple random walk on $d$-dimensional lattice $Z^d$, it is well known that $d_w= 2$, for all $d$. A random walk is called a diffusion  if $d_w =2$. When $d_w >2$, the walk is called an anomalous diffusion. This type of anomalous diffusion were first observed by random walks in random media such as critical percolation clusters. Critical percolation clusters behave like a fractal in a large scale and simple random walks on the clusters are expected to be   anomalous diffusions \cite{AlexanderOrbach 1982, RammalToulouse 1983}.

Kesten \cite{Kesten 1986} showed that the anomalous diffusion occurs on critical percolation on trees and $Z^2$. For high-dimensional $Z^d$, anomalous diffusion have also obtained for random walks on critical percolation clusters \cite{Kumagai 2006, BarlowJaraiKumagaiSlade 2006, KozmaNachmias 2009}. For simplicity, instead of on random clusters, random walks on deterministic regular fractal structures have also been studied. Discrete random walks on Sierpinski gaskets were studied in \cite{AlexanderOrbach 1982, 2005 diffusion,  RammalToulouse 1983} and Brownian motions on Sierpinski gaskets were constructed \cite{BarlowPerkins1988, Goldstein1987, Kusuoka1987}. 

It is generally believed that the fractal dimensions of the walk is related to the spectral dimension, not just the geometric dimension (Hausdorff dimension), of the space. It is conjectured by Alexander and Orbach \cite{AlexanderOrbach 1982} that $d_w=2d_f/d_s$, where $d_f$ is the Hausdorff dimension of the space and $d_s$ is the spectral dimension of the space. $d_f$ is a geometric constant, while $d_s$ is an analytic constant defined by properties of harmonic functions (density of the states) on the space. 

\subsubsection{Classical walk dimension on Sierpinski gasket}

It is proved that the fractal dimension of random walks on the Sierpinski gasket is  $d_w=\ln 5 /\ln 2$ agreeing with the conjecture, see Ref.~\citen{2005 diffusion}.

Brownian motions, heat kernels, Sobolev inequalities, and Harnak inequality have been studied on Sierpinski carpets \cite{BarlowBass1999} and shown that $d_s$ is more significant than $d_f$; for example, unlike the Euclidean space, $d_f$ does not enter into the Sobolev inequality in this case.

\subsubsection{Classical first passage time}

An alternative approach to computing $d_w$ is to consider first passage time $T_L$, the first time that the random walker has reached a distance $L$ from its initial position, and set $E(T_L) \sim L^{c_w}$, as $ L \to \infty$. It is well known that for simple random walk on $Z^d$, $c_w=d_w=2$. When the random walk is strongly localized in the sense that $P(T_L = \infty)>0$,  we consider the first returning  time $\tau_L$, the first time that the random walker has either reached a distance $L$ from the initial point or returned to its  initial point and the recurrent exponent $r_w$ that measures the degree of recurrence is defined by $E(\tau_L) \sim L^{r_w}$, as $ L \to \infty$.  

\subsubsection{Classical localization}
If the walk $\{X_t, t \ge 0\}$ is localized in the sense that 
$\lim_{L  \to \infty}P(X_{\tau_L} =x|X_0=x)=1$, we also want to find the localization exponents defined
by 

$$\beta_w=\lim_{L \to \infty}  \frac{-\ln |P(X_{\tau_L} =x|X_0=x)-1|}{\ln L}.$$

In this case, we would  like to study the asymptotic distributions of $p^{(L)}(y)=P(X_{\tau_L} =y|X_0=x)$. Let  $p^{(\infty)} (y)=\lim_{L \to \infty} p^{(L)} (y)$. Let $d^{(L)}=|| p^{(L)}-p^{(\infty)}||_{TV}$ be the total variation distance of $p^{(L)}$ and $p^{(\infty)}$. Then we define 
$$\delta_w=\lim_{L \to \infty}  \frac{-\ln d^{(L)}}{\ln L}.$$
We would also like to study the asymptotic behavior of their relative entropy. Let $H^{(L)}=-\sum_{y}p^{(\infty)}(y) \ln p^{(L)}(y)$ be the relative entropy of $p^{(L)} $ with respect to $p^{(\infty)}$. Then we define 
$$\eta_w=\lim_{L \to \infty}  \frac{-\ln H^{(L)}}{\ln L}.$$

Since the simple random walk on $\mathbb{Z}^d$ is either null recurrent or transient, $\beta_w$,
 $\delta_w$ and $\eta_w$ are undefined. For quantum random walk on Sierpinski gasket, we will show that it is strongly localized and these exponents exist and we will obtain their numerical values.

\subsection{Quantum theory}

In this paper, we will study the  recurrence and localization exponents for quantum random walks on Sierpinski gaskets. In quantum computing, quantum algorithms have shown speed ups over their classical counterpart. One of the most celebrated examples is Shor's factorization algorithm that shows exponential speed up over the classical algorithm \cite{Sho97}. For unsorted database search, Grover's quantum search algorithm has shown $O(\sqrt {N})$ speed up over the classical  algorithm of $O(N)$ \cite{Grover96}. Classical random walks have been used for spatial search. For quantum spatial search, quantum random walks were introduced by Aharanov, Davidovich, and Zagury \cite{Aharonov 1993}  as the quantum version  of classical random walks. Since quantum random walks spread out ballistically, which is faster than the diffusive behavior of classical random walks, one generally expects quantum spatial search algorithm also has a speed up over the classical counterpart when quantum random walks are used for the search.

\subsubsection{Quantum Search on $Z^d$}

A spatial search using the discrete quantum random walk on $Z^d$ has been studied \cite{AmbainisKempeRivosh2008}. 

For $d=2$, the number of queries to the oracle needed in order to search for a marked site in a $d$-dimensional cubes with $N$ sites is $T=O(\sqrt {N \log N})$ with flip-flop walk and the moving walk is worse than the classical case.  For $d \ge 3 $, $T=O(\sqrt {N})$ with flip-flop walk and the moving walk is worse than the classical case. 
For $ d=1$ Hadamard quantum random walk, $T=O(N)$ with moving walk and flip-flop walk is worse than the classical case \cite{AmbainisKempeRivosh2008}.

\subsubsection{Quantum Search on Sierpinski gasket}

In Ref.~\citen{PatelRaghunathan2012}, it is conjectured  that for the spatial search in $d$-dimensional lattice, the lower bound obeys 
\begin{eqnarray}
t_2 \ge \max  \{d N ^{1/d}, \pi \sqrt{N}/4 \}, \label{eq:search lower bound}
\end{eqnarray}
where $t_2$ is the number of times that the oracle have been called.

The first bound $d N ^{1/d}$ in (\ref{eq:search lower bound}) is called the relativity bound; it is the minimal time that the walk must be able to cross the lattice. The bound $\pi \sqrt{N}/4$ in  (\ref{eq:search lower bound}) is called the unitarity bound; it is the Grove's bound when there is no constraint on the walk. (\ref{eq:search lower bound}) shows that there is a critical dimension $d_c=2$. If $d > d_c$, then unitarity bound dominates the search speed and if $d <d_c$, the relativity bound dominates the search speed. In Ref.~\citen{PatelRaghunathan2012}, numerical results support the conjecture (\ref{eq:search lower bound}) for non-integer dimensional structures, such as Sierpinski gaskets, but with $d$ in (\ref{eq:search lower bound}) replaced by the spectral dimension $d_s$ of the space. Note that for integer $d$, $d=d_s$.

\subsubsection{Behavior of quantum walks on $Z^d$}
 
With the motivation of quantum spatial search, one is interested in the behavior of the quantum walks on various spaces.  There are two forms of quantum walks, continuous-time quantum walks and discrete-time quantum walks have been widely studied. In this paper, we restrict our discussion to the discrete time.
Discrete time quantum walks for searching algorithms have been widely studied on finite graphs. However, the limiting distributions of a quantum random walks have only been studied mostly on the line $\mathbb{Z}$  using combinatorial methods and recursive methods \cite{2001 Ambainis}, or higher dimensional space $\mathbb{Z}^d$ using Fourier transforms \cite{Grimmett et al}.
It is known that $d_w=1$, for standard (moving) Hadamard quantum random walks on $Z^d$, for all $d$. Yang et al. \cite{2007 Z^d} used path integral to derive a recursive formula to solve the return probability in $\mathbb{Z}^d$.  For partially decoherent quantum random walks on $Z^1$ with both position- and coin-space decoherence, Zhang \cite{K Zhang} showed that the scaling limit is Gaussian. For partially, only coin-space decoherent quantum random walks on $Z^1$, Fan et al. \cite{2011 decoherence} showed that the distribution of the scaling limit is a continuous
convex combination of normal distributions (but not Gaussian) for quantum random walks in $\mathbb{Z}$.  For a comprehensive review of quantum walks, see Ref.~\citen{2012 review}. More recently,
Attal et al. \cite{2014 attal} showed the quantum version of central limit theorem for open quantum random walks. Subsequently, Xiong and Yang  \cite {2013 open} extended their results and showed the scaling limit of a partially open quantum random walk converges to a convex combination of Gaussian distributions. 

\subsubsection{Behavior of quantum walks on Sierpinski gasket.}  

  Flip-flop quantum random walk on Sierpinski gasket has been studied numerically \cite{2012 sierpinski}, where the simulation results show the diffusion exponent  $d_w $ lies between $1.92$ and $3.45$ (with average over some initial points).   The simulation shows that it is of a very small spreading rate as a sub-diffusive process.  Flip-flop quantum random walk has also been applied to spatial search algorithm on Sierpinski gasket \cite{PatelRaghunathan2012}, where simulation shows that the scaling behavior of the spatial search is determined by its spectral dimensions not the fractal dimensions.

\subsubsection{Quantum localization} 
 Another feature of quantum random walks is their localization.  This property is reminiscent of the Anderson localization in condensed matter physics \cite{Anderson1978}. Anderson showed that there is no quantum diffusion on disordered medium and the quantum states after a long time are centered around their initial states \cite{Anderson1978, WiersmaBLR1997}.  

For quantum random walks, Tregenna et al. \cite{TregennaFMK2003} showed numerically that the quantum random walk with Grover coin operator centered around its initial location with high probability. In Ref.~\citen{InuiKonishiKonno2004}, Inui, Konoshi and Konno showed that the 2-dimensional Grover quantum walks exhibits localization phenomenon and they also gave a criterion for occurrence of localization in terms of degrees of degeneracy of eigenstates: when the Hamiltonian  spectrum are degenerated in a rate proportional to the size of the system. It was shown that the localization may depend on the initial states, for some initial state there are localization and for other initial states localization disappears. It was pointed out by Tregenna et al. \cite{TregennaFMK2003} that this behavior can be used to control the Grover's search. 

 \subsection{Our results on the behavior quantum walks on Sierpinski gasket} 
 
 Despite all of the  interesting simulation results, quantum walks on fractal have not been  studied analytically  because the coin operator depends on the position and the method of Fourier transform on the fractals is not effective. Therefore computing the hitting probability and expectation of the passage time using combinatorics or Fourier transform used on $\mathbb{Z}^d$ do not work. In this paper, we analytically compute the exit distribution using path integrals and recursive Green function approaches. We will obtain an exact recursive relations for the amplitude functions and use them to compute the exit and recurrence probabilities and the expected exit and recurrence times. Using our approach, by Theorem \ref{th: g recursive} and  Theorem \ref{th:T formula} and with the aid of   Monte Carlo integration, we obtained the quantum analogues of $\beta_w$, $\gamma_w$ and $\delta_w$ and $\eta_w$. The detail definitions and values of which for the quantum random walk on Sierpinski gaskets will be summarized in Section \ref{sec:main results}.
 
 We also obtain some interesting behaviors of quantum random walks that are different from the corresponding classical random walks. For example, the quantum random walk starting with $|\mathbf{0}^0 \rangle$ will not exit at $|\mathbf{a_1}^5\rangle$ from the right bottom direction while the exiting probability for the classical random walks from the right bottom direction is positive. And with probability less than $1$, the particle will ever exit the boundary of $F^{(n)}$ and its reflection part while it is always 1 for the classical case. When the coin operator is uniform, the evolution operator is no longer unitary, then our approach applies to the classical random walk on Sierpinski gasket, and we obtain the limiting behavior of hitting probabilities for each direction of the boundary points for classical random walks.

\subsection{Open problems on the behavior of quantum walks on Sierpinski gasket}  

We have shown that solving the Dirichlet boundary value problem for amplitude functions on the space that has self-similar structures such as Sierpinski gasket is effective. It is not known however, that this general method is effective to other graphs. The problems of calculating 
$\beta_w$, ${\gamma_w}$ ,$\delta_w$, and $\eta_w$ for quantum random walks on general graphs without self-similarities remain open. 

\subsection{Organization of our paper} 

This paper is organized as follows. In Section 2, we give definitions and develop necessary tools - path integrals, hitting times, Green functions,  and solutions of quantum Poisson equations and quantum Dirichlet problem - for quantum random walks on Sierpinski gasket and give our main Theorem \ref{th: g recursive} and Theorem \ref{th:T formula} on the recursive relations of the amplitude Green functions. In Section 3, we give proofs to Lemma 1, and from Theorems \ref{th: g recursive} - \ref{th:T formula}. In Section 4, we apply the recursive formulas to obtain our recurrence of the quantum random walk Theorem \ref{th:recurrent} and obtain the numerical values of exponents listed on the Table \ref{tab:localization table}. Section 5 contains  the conclusion of this paper and discussions of some open problems in this direction.

\section{Quantum walks on Sierpinski gasket}

\subsection{Definitions, Notations and  Main Results}\label{sec:main results}

\subsubsection{Definitions and  notations}
The Sierpinski gasket, a fractal set with the overall shape of an triangle, generated recursively into more triangles with the same shape. Let $F^{(n)}$ be the  $n^{th}$ order Sierpinski gasket, see Fig.  \ref{fig Sierpinski Gasket} - \ref{fig Sierpinski Gasket order two} for $F^{(n)}$,  $n = 0,1,2.$ We define $F^{(\infty)} = \cup_{n=0}^\infty F^{(n)}$.
\begin{figure}[htb] 
\begin{minipage}{0.3\textwidth} 
\begin{tikzpicture}[scale = 0.35]
\begin{axis}[dashed, xmin=0,ymin=0,xtick = {0,1,2},ytick = {0,1},xmax=8.4,ymax=4.5]
\addplot[mark = none] coordinates{(1,0) (1,5)};
\addplot[mark = none] coordinates{(2,0) (2,5)};
\addplot[mark = none] coordinates{(3,0) (3,5)};
\addplot[mark = none] coordinates{(4,0) (4,5)};
\addplot[mark = none] coordinates{(5,0) (5,5)};
\addplot[mark = none] coordinates{(6,0) (6,5)};
\addplot[mark = none] coordinates{(7,0) (7,5)};
\addplot[mark = none] coordinates{(8,0) (8,5)};
\addplot[mark = none] coordinates{(0,1) (9,1)};
\addplot[mark = none] coordinates{(0,2) (9,2)};
\addplot[mark = none] coordinates{(0,3) (9,3)};
\addplot[mark = none] coordinates{(0,4) (9,4)};
\addplot[mark = *, solid, thick] coordinates {(0, 0) (2,0) (1,1) (0,0)};
\addplot[mark=*] coordinates {(2,0)} node[label={$b_0$}]{} ;
\addplot[mark=*] coordinates {(1,1)} node[label={$a_0$}]{} ;
\end{axis}        
\end{tikzpicture}
\caption{$0^{th}$ order
Sierpinski Gasket $F^{(0)}$}
\label{fig Sierpinski Gasket}
\end{minipage}
\begin{minipage}{0.3\textwidth}
\begin{tikzpicture}[scale = 0.35] 
\begin{axis}[dashed, xmin=0,ymin=0,xtick = {0,1,...,4},ytick = {0,1,2},xmax=8.4,ymax=4.5]
\addplot[mark = none] coordinates{(1,0) (1,5)};
\addplot[mark = none] coordinates{(2,0) (2,5)};
\addplot[mark = none] coordinates{(3,0) (3,5)};
\addplot[mark = none] coordinates{(4,0) (4,5)};
\addplot[mark = none] coordinates{(5,0) (5,5)};
\addplot[mark = none] coordinates{(6,0) (6,5)};
\addplot[mark = none] coordinates{(7,0) (7,5)};
\addplot[mark = none] coordinates{(8,0) (8,5)};
\addplot[mark = none] coordinates{(0,1) (9,1)};
\addplot[mark = none] coordinates{(0,2) (9,2)};
\addplot[mark = none] coordinates{(0,3) (9,3)};
\addplot[mark = none] coordinates{(0,4) (9,4)};
\addplot[mark = *, solid, thick] coordinates {(0, 0) (2,0) (4,0) (3,1) (2,2) (1,1) (0,0)};
\addplot[mark = *, solid ,thick] coordinates {(1,1) (3,1) (2,0) (1,1)};
\addplot[mark=*] coordinates {(2,0)} node[label={$b_0$}]{} ;
\addplot[mark=*] coordinates {(1,1)} node[label={$a_0$}]{} ;
\addplot[mark=*] coordinates {(4,0)} node[label={$b_1$}]{} ;
\addplot[mark=*] coordinates {(2,2)} node[label={$a_1$}]{} ;
\end{axis}        
\end{tikzpicture}
\caption{$1^{st}$ order Sierpinski Gasket $F^{(1)}$}
\label{fig Sierpinski Gasket order one}
\end{minipage}
\begin{minipage}{0.35\textwidth} 
\begin{tikzpicture}[scale = 0.35]
\begin{axis}[dashed, xmin=0,ymin=0,xtick = {0,1,...,9},xmax=8.4,ymax=4.5]
\addplot[mark = none] coordinates{(1,0) (1,5)};
\addplot[mark = none] coordinates{(2,0) (2,5)};
\addplot[mark = none] coordinates{(3,0) (3,5)};
\addplot[mark = none] coordinates{(4,0) (4,5)};
\addplot[mark = none] coordinates{(5,0) (5,5)};
\addplot[mark = none] coordinates{(6,0) (6,5)};
\addplot[mark = none] coordinates{(7,0) (7,5)};
\addplot[mark = none] coordinates{(8,0) (8,5)};
\addplot[mark = none] coordinates{(0,1) (9,1)};
\addplot[mark = none] coordinates{(0,2) (9,2)};
\addplot[mark = none] coordinates{(0,3) (9,3)};
\addplot[mark = none] coordinates{(0,4) (9,4)};
\addplot[mark = *, solid, thick] coordinates {(0, 0) (2,0) (4,0) (6,0) (8,0) (7,1) (6,2) (5,3) (4,4) (3,3) (2,2) (1,1) (0,0)};
\addplot[mark = *, solid ,thick] coordinates {(1,1) (3,1) (2,0) (1,1)};
\addplot[mark = *, solid, thick] coordinates {(2,2) (4,2) (6,2) (5,1) (4,0) (3,1) (2,2)};
\addplot[mark = *, solid ,thick] coordinates {(3,3) (5,3) (4,2) (3,3)};
\addplot[mark = *, solid, thick] coordinates {(5,1) (7,1) (6,0) (5,1)};
\addplot[mark=*] coordinates {(2,0)} node[label={$b_0$}]{} ;
\addplot[mark=*] coordinates {(1,1)} node[label={$a_0$}]{} ;
\addplot[mark=*] coordinates {(4,0)} node[label={$b_1$}]{} ;
\addplot[mark=*] coordinates {(2,2)} node[label={$a_1$}]{} ;
\addplot[mark=*] coordinates {(8,0)} node[label={$b_2$}]{} ;
\addplot[mark=*] coordinates {(4,4)} node[label={$a_2$}]{} ;

\end{axis}        
\end{tikzpicture}
\caption{$2^{nd}$ order Sierpinski Gasket $F^{(2)}$}
\label{fig Sierpinski Gasket order two}
\end{minipage}
\end{figure}

The $n^{th}$ order Sierpinski gasket $F^{(n)}$ is a degree-4 regular graph except for overall corners. Let $V_{n}$ be the number of vertices for $F^{(n)}$. Then $V_{0} = 3$, $V_{1} = 6$, $V_{2} = 15$ and it is easy to check that $V_{n} = 3V_{n-1}-3 = 3^{n+1}-3^n-3^{n-1}-\cdots -3$ = $\frac{3(3^n+1)}{2}$. Let $S_n$ be the number of vertices  in the $x$ direction, for example, $S_0 = 2$, $S_1 = 3$, $S_2 = 5$, so $S_n = 2S_{n-1}-1 = 2^n+1$. By the definition of Hausdorff dimension $d_f$, as $n $ goes to infinity, 
\begin{equation*}
V_n \sim S_n^{d_f},
\end{equation*}
we have $d_f = \lim_{n\rightarrow \infty} \frac{\ln V_n}{\ln S_n}$ = $\frac{\ln3}{\ln2}$, which is bigger than the dimension of the line and smaller than that of the plane. Classical random walks on the Sierpinski gasket have been studied in Ref.~\citen{2005 diffusion}. It has shown a recursive formula for the hitting probabilities, and that the expectation of the passage time on the $(n+1)^{th}$ generation is 5 five times the expectation of the passage time on the $n^{th}$ generation. It then follows that  the diffusion exponent $d_w = \frac{\ln 5}{\ln 2}.$  A self-avoiding random walk model on the Sierpinski gasket has also been studied by using recursive formulas in Ref.~\citen{2016 Kumiko}. Using recursive formulas, Chang et al. studied lattice trees on the Sierpinski gasket \cite{ChangChenYang2007}. Despite existences of these classical results, recursive method cannot be applied directly to quantum random walks on graphs other than $Z^1$ due to the differences between amplitude functions and probability distributions and, moreover, the solution of amplitude Dirichlet problem has not been known.

The $n^{th}$ order Sierpinski gasket $F^{(n)}$ is a degree-4 regular graph except overall corners. 
For the $n^{th}$ order Sierpinski gasket $F^{(n)}$ embedded in the two-dimensional plane $\mathbb{R}^2$, any point in $F^{(n)}$ can be represented by its coordinates $x=(x_1, x_2)$,  with $0 \leq x_1 \leq 2^{n+1}$, and $0 \leq x_2 \leq 2^n$ restricted to be on the gasket, see Fig.  \ref{fig Sierpinski Gasket} - \ref{fig Sierpinski Gasket order two} for $F^{(n)}$,  $n = 0,1,2.$

 For the directions coming out of  each point, we denote  $\mathbf{e}_0  = (2,0) $, $\mathbf{e}_1  = (1,1) $, $\mathbf{e}_2  = (-1,1) $, $\mathbf{e}_3 = (-2,0) $, $\mathbf{e}_4  = (-1,-1) $, $\mathbf{e}_5  = (1,-1)  $ as in Figure \ref{fig:computational basis}. Note that only 4 of these directions will be used at each point for the random walk, so  we define $out(\mathbf{x})$ to be the set of these four direction vectors. For $\mathbf{a_n}$ and $\mathbf{b_n}$, we treat $F^{(n)}$ as part of $F^{(n+1)}$, so $out(\mathbf{a_n}) = \{\mathbf{e}_0,\mathbf{e}_1,\mathbf{e}_4,\mathbf{e}_5\}$ and $out(\mathbf{b_n}) = \{\mathbf{e}_0,\mathbf{e}_1,\mathbf{e}_2,\mathbf{e}_3\}$. As for $\mathbf{0}$, we extend $F^{(n)}$ as follows. If we think $\mathbf{0},\mathbf{a_n},\mathbf{b_n}$ as boundary points of $F^{(n)}$, then we call reflection of $F^{(n)}$ denoted by $F^{(n)'}$ with boundary points $\mathbf{0},\mathbf{a'_n},\mathbf{b'_n}$, here $\mathbf{a'_n} = -\mathbf{a_n} =(-2^n,-2^n)$ and $\mathbf{b'_n}  = (2^n,-2^n) $, see Figure \ref{fig:F prime}.  Therefore, $out(\mathbf{0}) = \{\mathbf{e}_0,\mathbf{e}_1,\mathbf{e}_4,\mathbf{e}_5\}$.

\begin{figure} 
\begin{minipage}{0.45\textwidth} 
\begin{tikzpicture}[scale = 0.65]
\begin{axis}[dashed, xmin=-5,ymin=-5,xticklabels = {,,}, yticklabels = {},xmax=9,ymax=5,]
\addplot[mark = none] coordinates{(-5,0) (9,0)};
\addplot[mark = none] coordinates{(-5,1) (9,1)};
\addplot[mark = none] coordinates{(-5,2) (9,2)};
\addplot[mark = none] coordinates{(-5,3) (9,3)};
\addplot[mark = none] coordinates{(-5,4) (9,4)};
\addplot[mark = none] coordinates{(-5,-1) (9,-1)};
\addplot[mark = none] coordinates{(-5,-2) (9,-2)};
\addplot[mark = none] coordinates{(-5,-3) (9,-3)};
\addplot[mark = none] coordinates{(-5,-4) (9,-4)};
\addplot[mark = none] coordinates{(-5,-1) (5,-1)};
\addplot[mark = none] coordinates{(-5,-2) (5,-2)};
\addplot[mark = none] coordinates{(0,-5) (0,9)};
\addplot[mark = none] coordinates{(1,-5) (1,9)};
\addplot[mark = none] coordinates{(2,-5) (2,9)};
\addplot[mark = none] coordinates{(3,-5) (3,9)};
\addplot[mark = none] coordinates{(4,-5) (4,9)};
\addplot[mark = none] coordinates{(5,-5) (5,9)};
\addplot[mark = none] coordinates{(6,-5) (6,9)};
\addplot[mark = none] coordinates{(7,-5) (7,9)};
\addplot[mark = none] coordinates{(8,-5) (8,9)};
\addplot[mark = none] coordinates{(-1,-5) (-1,9)};
\addplot[mark = none] coordinates{(-2,-5) (-2,9)};
\addplot[mark = none] coordinates{(-3,-5) (-3,9)};
\addplot[mark = none] coordinates{(-4,-5) (-4,9)};
\addplot[mark = *, solid, thick] coordinates {(0, 0) (2,0) (4,0) (6,0) (8,0) (7,1) (6,2) (5,3) (4,4) (3,3) (2,2) (1,1) (0,0)};
\addplot[mark = *, solid ,thick] coordinates {(1,1) (3,1) (2,0) (1,1)};
\addplot[mark = *, solid, thick] coordinates {(2,2) (4,2) (6,2) (5,1) (4,0) (3,1) (2,2)};
\addplot[mark = *, solid ,thick] coordinates {(3,3) (5,3) (4,2) (3,3)};
\addplot[mark = *, solid, thick] coordinates {(5,1) (7,1) (6,0) (5,1)};
\addplot[mark=*] coordinates {(0,0)} node[label={$0$}]{} ;
\addplot[mark=*] coordinates {(2,0)} node[label={$b_0$}]{} ;
\addplot[mark=*] coordinates {(1,1)} node[label={$a_0$}]{} ;
\addplot[mark=*] coordinates {(4,0)} node[label={$b_1$}]{} ;
\addplot[mark=*] coordinates {(2,2)} node[label={$a_1$}]{} ;
\addplot[mark=*] coordinates {(8,0)} node[label={$b_2$}]{} ;
\addplot[mark=*] coordinates {(4,4)} node[label={$a_2$}]{} ;
\addplot[mark = *, solid, thick] coordinates {(0, 0) (-1,-1) (-2,-2) (-3,-3) (-4,-4) (-2,-4) (0,-4) (2,-4) (4,-4) (3,-3) (2,-2) (1,-1) (0,0)};
\addplot[mark = *, solid ,thick] coordinates {(-1,-1) (1,-1) (0,-2) (-1,-1)};
\addplot[mark = *, solid, thick] coordinates {(-2,-2) (0,-2) (2,-2) (1,-3) (0,-4) (-1,-3) (-2,-2)};
\addplot[mark = *, solid ,thick] coordinates {(-3,-3) (-1,-3) (-2,-4) (-3,-3)};
\addplot[mark = *, solid, thick] coordinates {(1,-3) (3,-3) (2,-4) (1,-3)};
\addplot[mark=*] coordinates {(1,-1)} node[xshift=0pt,yshift=-4pt][label={$b'_0$}]{} ;
\addplot[mark=*] coordinates {(-1,-1)} node[xshift=0pt,yshift=-4pt][label={$a'_0$}]{} ;
\addplot[mark=*] coordinates {(2,-2)} node[xshift=0pt,yshift=-4pt][label={$b'_1$}]{} ;
\addplot[mark=*] coordinates {(-2,-2)} node[xshift=0pt,yshift=-4pt][label={$a'_1$}]{} ;
\addplot[mark=*] coordinates {(4,-4)} node[xshift=0pt,yshift=-4pt][label={$b'_2$}]{} ;
\addplot[mark=*] coordinates {(-4,-4)} node[xshift=0pt,yshift=-4pt][label={$a'_2$}]{} ;
\end{axis}        
\end{tikzpicture}
\caption{$F^{(2)} \cup F^{(2)'}$}
\label{fig:F prime}
\end{minipage}
\hspace{1cm}
\begin{minipage}{0.45\textwidth} 
\begin{tikzpicture}[scale = 0.65]
\begin{axis}[dashed, xmin=0,ymin=0,xticklabels = {,,},yticklabels = {,,}, xmax=8,ymax=4]
\addplot[mark = none] coordinates{(1,0) (1,5)};
\addplot[mark = none] coordinates{(2,0) (2,5)};
\addplot[mark = none] coordinates{(3,0) (3,5)};
\addplot[mark = none] coordinates{(4,0) (4,5)};
\addplot[mark = none] coordinates{(5,0) (5,5)};
\addplot[mark = none] coordinates{(6,0) (6,5)};
\addplot[mark = none] coordinates{(7,0) (7,5)};
\addplot[mark = none] coordinates{(8,0) (8,5)};
\addplot[mark = none] coordinates{(0,1) (9,1)};
\addplot[mark = none] coordinates{(0,2) (9,2)};
\addplot[mark = none] coordinates{(0,3) (9,3)};
\addplot[mark = none] coordinates{(0,4) (9,4)};
\addplot[mark=*] coordinates {(4,2)} node(0)[]{} ;
\addplot[mark=*] coordinates {(6,2)} node(1)[]{} ;
\addplot[mark=*] coordinates {(5,3)} node(2)[]{} ;
\addplot[mark=*] coordinates {(3,3)} node(3)[]{} ;
\addplot[mark=*] coordinates {(2,2)} node(4)[]{} ;
\addplot[mark=*] coordinates {(3,1)} node(5)[]{} ;
\addplot[mark=*] coordinates {(5,1)} node(6)[]{} ;
\draw [->,thick ,solid] (0) -- (1)node[xshift=-10pt,yshift=5pt] {$e_0$};
\draw [->,thick ,solid] (0) -- (2)node[xshift=0pt,yshift=-10pt] {$e_1$};
\draw [->,thick ,solid] (0) -- (3)node[xshift=0pt,yshift=-10pt] {$e_2$};
\draw [->,thick ,solid] (0) -- (4)node[xshift=10pt,yshift=5pt] {$e_3$};
\draw [->,thick ,solid] (0) -- (5)node[xshift=0pt,yshift=10pt] {$e_4$};
\draw [->,thick ,solid] (0) -- (6)node[xshift=0pt,yshift=10pt] {$e_5$};
\end{axis}
\end{tikzpicture}
\caption{computational basis}
\label{fig:computational basis}
\end{minipage}
\end{figure}

For quantum random walk on the Sierpinski gasket, we let $H^{(n)}_p = Span\{ |\mathbf{x}\rangle ; \mathbf{x} \in  F^{(n)} \cup F^{(n)'} \}$ be the position Hilbert space associated with $F^{(n)} \cup F^{(n)'}$, where   $\{ |\mathbf{x}\rangle ; \mathbf{x} \in  F^{(n)} \cup F^{(n)'} \}$ is a set of   orthonormal basis. Let $H_c$ be the coin space with computational orthonormal basis $\{|\mathbf{e}_i\rangle , 0 \leq i \leq 5 \}$.
The state space is defined by $H^{(n)} =  Span \{|\bar{\mathbf{x}}\rangle  = |\mathbf{x}\rangle \otimes |\mathbf{e}_i\rangle ,\mathbf{x}\in F^{(n)} \cup F^{(n)'}, \mathbf{e}_i \in out(\mathbf{x})\}$. $H^{(n)}$ is a subspace of $H_p^{(n)} \otimes H_c$. For simplicity, throughout this paper, we write the state space as $H^{(n)} =  Span \{|\bar{\mathbf{x}}\rangle  = |\mathbf{x}^i\rangle , \mathbf{x}\in F^{(n)} \cup F^{(n)'}, \mathbf{e}_i \in out(\mathbf{x})\}$.

Let $F^{(\infty)} = \cup_{n=0}^\infty F^{(n)}$ and $F^{(\infty)'} = \cup_{n=0}^\infty F^{(n)'}$, then the state space can be extended as 
$H^{(\infty)} =  Span \{|\mathbf{x}^i\rangle , \mathbf{x}\in F^{(\infty)} \cup F^{(\infty)'}, \mathbf{e}_i \in out(\mathbf{x})\}$.

The shift operator $S: H^{(\infty)} \rightarrow H^{(\infty)}$ is defined by 
$ S |\mathbf{x}^k\rangle  = |\mathbf{y}^j\rangle , $ for $j = k+3\  (mod\ 6)$, where $\mathbf{y} = \mathbf{x}+\mathbf{e}_k$.

Let G be the $4 \times 4$ Grover matrix $(g_{ij})$, where 
$g_{ij} =  -1/2$   if $ i = j$ and $g_{ij} =  1/2$   if $ i \ne j$. So,
           $$G = r
  \left[
 \begin{array}{ c c c c }
-1 & 1 & 1 & 1 \\
1 & -1 & 1 & 1 \\
1 & 1 & -1 & 1 \\
1 & 1 & 1 & -1 
\end{array} \right], \text{where} \  \displaystyle r = \frac{1}{2}.$$\\
Let $G_{\mathbf{x}}: H^{(\infty)} \rightarrow H^{(\infty)}$ be a local operator defined by
$$G_{\mathbf{x}}(|\mathbf{y}^i\rangle ) = \sum_{j: \mathbf{e}_j \in out(\mathbf{y})} g_{ji} |\mathbf{x}^j\rangle  \delta_{{\mathbf{x}}{\mathbf{y}}}, \forall {\mathbf{y}} \in F^{(\infty)}\cup F^{(\infty)'}, \mathbf{e}_i \in out(\mathbf{y}). $$
Define $\tilde{G} = \sum_{\mathbf{x}} G_{\mathbf{x}}$.
The evolution operator for the quantum random walk is defined by $U = S\tilde{G}$. Then $U: H^{(\infty)} \rightarrow H^{(\infty)}$ is a unitary operator.

Let $\psi_0 \in H^{(\infty)}$ and $\psi_t = U^t\psi_0$. The sequence $\{\psi_t\}_0^\infty$ is called a quantum random walk with the initial state $\psi_0$. 
Let $\psi_t = \sum_{\mathbf{e}_i \in out(\mathbf{x})}\sum_{\mathbf{x}\in F^{(\infty)} \cup F^{(\infty)'}}\psi_t(\mathbf{x},i)|\mathbf{x}^i\rangle $ be the quantum random walk at time $t$, where $\psi_t(\mathbf{x},i)$ is the coefficient at $|\mathbf{x}^i\rangle $. Let $p_t(\mathbf{x},i) = |\psi_t(\mathbf{x},i)|^2$ be the probability that the particle is at state $|\mathbf{x}^i\rangle $ at time $t$, and $p_t(\mathbf{x}) = \sum_{\mathbf{e}_i \in out(\mathbf{x})}p_t(\mathbf{x},i)$ be the probability that the particle is found at state $|\mathbf{x}\rangle $ at time $t$.

 We will use the path integral 
 as formulated as follows. 
 For the state space 
$Span \{|\mathbf{x}^i\rangle, \mathbf{x}\in F^{(\infty)}\cup F^{(\infty)'}, \mathbf{e_i} \in out(\mathbf{x})\}$, a path $w$ is defined by $w = (w_0, \dots, w_m)$, where $w_t = \mathbf{x}_t ^{k_t} $, and $\Vert \mathbf{x}_{t+1} - \mathbf{x}_t \Vert \le  2$ where $\Vert \mathbf{x} -\mathbf{y}\Vert = |x_1 - y_1|+|x_2-y_2|.$ The length of $w$ is defined by $|w| = m$. Let $\Omega^m = \{ w;|w|=m \}.$

\begin{definition}(Amplitude function)  \label{def1}

Let $$\phi (\mathbf{x}^i, \mathbf{y}^j) = 
\left\{
        \begin{array}{ll}
            -r & if \ \Vert \mathbf{x}-\mathbf{y} \Vert = 2, j = i+3\ (mod\ 6) \\
            r & if \ \Vert \mathbf{x}-\mathbf{y} \Vert = 2, j \neq i+3\ (mod\ 6), \mathbf{y} = \mathbf{x} + \mathbf{e}_k, j = k+3\ (mod\ 6) \\
            0 & otherwise
        \end{array}
    \right.$$
The amplitude function for $w \in  
\Omega^{m}$ is defined as 
\begin{equation}
 \Psi (w) =\prod_{t=0}^{m-1} \phi (w_t , w_{t+1}). 
\end{equation}
\end{definition}

\begin{definition}  Let $\Gamma \subseteq \Omega^m$ .Then the amplitude function of a $\Gamma $ is defined by 
\begin{equation}
 \Psi (\Gamma) =\sum_{w \in \Gamma}\Psi(w).
\end{equation}
Let $\Omega = \cup_{m=0}^{\infty} \Omega^m$. For $\Gamma \in \Omega$ with $\Gamma^m = \Gamma \cap \Omega^m$, we also define
\begin{equation*}
\Psi (\Gamma) =\sum_{m=0}^{\infty} \Psi(\Gamma^m).
\end{equation*}
\end{definition}

For any $\psi \in H^{(\infty)}$, we shall write $\psi = \sum_{\mathbf{e}_i \in out(\mathbf{x})}\sum_{\mathbf{x}\in F^{(\infty)} \cup F^{(\infty)'}}\psi(\mathbf{x},i)|\mathbf{x}^i\rangle $.
Our {\it{path integral formula}}, which is a variant of the one obatined for the discrete quantum random walks on $Z^d$  \cite{2007 Z^d}, takes the following form:  suppose $\psi_0 = |\mathbf{x}^i\rangle $, and $\psi_t = U^t \psi_0$, then 
$$\psi_t(\mathbf{y},j) = \Psi(w_0 = \mathbf{x}^i , w_{t} = \mathbf{y}^j).$$

The first exit time from level $n$ is defined by 
\begin{eqnarray}
 {\tau^{(n)}} = \inf \{ t\geq 1; w_t=x_t^{k_t}, x_t  \in {\partial F^{(n)}} \cup {\partial F^{(n)'}} \},\label{eq:first exit time}
 \end{eqnarray}
  where ${\partial F^{(n)}} \cup {\partial F^{(n)'}} = \{\mathbf{0}, \mathbf{a_n}, \mathbf{b_n},\mathbf{a'_n}, \mathbf{b'_n}\}$. We define the set of path that start with $\mathbf{x}^i$ and exit from $\mathbf{y}^j$, where $\mathbf{x},\mathbf{y}  \in {\partial F^{(n)}} \cup {\partial F^{(n)'}}$ at time $t$ as $\{w = (w_0 , \dots, w_{\tau^{(n)}}), w_0 = \mathbf{x}^i, w_{t} = \mathbf{y}^j ,\tau^{(n)} = t\}.$
  
  We also consider the first passage time at level $n$ defined as follows. Let $\partial (F^{(n)}\cup F^{(n)'}) := \{\mathbf{a_n}, \mathbf{b_n},\mathbf{a'_n}, \mathbf{b'_n}\}$. 
 Let 
 \begin{eqnarray}
 T^{(n)} = \inf \{ t\geq 1; w_t \in \partial (F^{(n)}\cup F^{(n)'}) \} \label{eq:first passage time}
 \end{eqnarray}
  be the first-passage time taken to exit $F^{(n)}\cup F^{(n)'}$ at the four vertices.

The amplitude Green function for quantum random walk is defined by
\begin{equation}
g^{(n)}(z)(\mathbf{x},\mathbf{y})_{j}^{i} =\sum_{t=1}^{\infty}z^t \Psi (w_0 = \mathbf{x}^i , w_{t} = \mathbf{y}^j ,\tau^{(n)} = t) .
\end{equation}
The probability that a quantum random walk starts with $\mathbf{x}^i$ and exits from $\mathbf{y}^j$ is given by
\begin{equation}\label{eq:quantum transition probability direction}
P^{(n)}(\mathbf{x},\mathbf{y})_{j}^{i} = \sum_{t=1}^{\infty}|\Psi (w_0 = \mathbf{x}^i , w_{t} = \mathbf{y}^j ,\tau^{(n)} = t)|^2 .
\end{equation}
The total probability that a quantum random walk starts with $\mathbf{x}^i$ and exits from $\mathbf{y}^j$ for any $j$ is given by
\begin{equation}\label{eq:quantum transition probability}
P^{(n)}(\mathbf{x},\mathbf{y})_{}^{i}=\sum _{j} P^{(n)}(\mathbf{x},\mathbf{y})_{j}^{i} .
\end{equation}
By Parseval's identity, we have
\begin{equation}\label{eq:probability parseval}
P^{(n)}(\mathbf{x},\mathbf{y})_{j}^{i} = \frac{1}{2\pi}\int_0^{2\pi} |g^{(n)}(e^{i\theta})(\mathbf{x},\mathbf{y})_{j}^{i}|^2 d\theta. 
\end{equation}
On $F^{(n)}$, consider the exiting amplitude Green matrix $g^{(n)}(z)$ as
\begin{scriptsize}
\begin{equation*}
\bordermatrix{~& \mathbf{0}^0& \mathbf{0}^1 & \mathbf{0}^4 & \mathbf{0}^5& \mathbf{a_n}^0&\mathbf{a_n}^1&\mathbf{a_n}^4&\mathbf{a_n}^5 & \mathbf{b_n}^0&\mathbf{b_n}^1&\mathbf{b_n}^2&\mathbf{b_n}^3         \cr
                  \mathbf{0}^0 &{g^{(n)}(z)(\mathbf{0},\mathbf{0})}_{0}^{0}  &&&\cdots&&&&\cdots&&&&{g^{(n)}(z)(\mathbf{0},\mathbf{b_n})}_{3}^{0}                                     \cr 
                  \mathbf{0}^1                                                  \cr
                  \mathbf{0}^4                                                  \cr
                  \mathbf{0}^5     &\vdots             &&&\ddots      &&&&&&&&\vdots                       \cr
                  \mathbf{a_n}^0                            \cr   
                  \mathbf{a_n}^1                           \cr
                  \mathbf{a_n}^4                           \cr
                  \mathbf{a_n}^5  &\vdots      &&&&&&&\ddots     &&&&\vdots               \cr
                  \mathbf{b_n}^0            \cr
                  \mathbf{b_n}^1            \cr
                  \mathbf{b_n}^2            \cr
                  \mathbf{b_n}^3 &{g^{(n)}(z)(\mathbf{b_n},\mathbf{0})}_{0}^{3}        &&&\cdots&&&&\cdots&&&&  {g^{(n)}(z)(\mathbf{b_n},\mathbf{b_n})}_{3}^{3} \cr}.
\end{equation*}
\end{scriptsize}

To facilitate our recursive scheme, we shall use the following general amplitude functions based on general transition amplitudes. Instead of using $\phi$, we consider any complex-valued function 
$\rho (\mathbf{x}^i, \mathbf{y}^j) $, for $\mathbf{x}^i, \mathbf{y}^j$, $ \mathbf{x}, \mathbf{y} \in F^{(\infty)} \cup F^{(\infty)'}, \mathbf{e}_i \in out(\mathbf{x}), \mathbf{e}_j \in out(\mathbf{y})$, $||\mathbf {x}-\mathbf{y}|| =0$ or $2$.
The amplitude function for $w \in  
\Omega^{m}$ associated with $\rho$ is defined by
\begin{equation}
 \Psi _{\rho}(w) =\prod_{t=0}^{m-1} \rho (w_t , w_{t+1}). 
\end{equation}

Similarly,  for $\Gamma \subseteq \Omega^m$, the amplitude function of a $\Gamma $ associated with $\rho$ is defined by 
\begin{equation}
 \Psi _{\rho} (\Gamma) =\sum_{w \in \Gamma}\Psi_{\rho}(w).
\end{equation}
Let $\Omega = \cup_{m=0}^{\infty} \Omega^m$. For $\Gamma \subset \Omega$ with $\Gamma^m = \Gamma \cap \Omega^m$, we also define
\begin{equation*}
\Psi_\rho (\Gamma) =\sum_{m=0}^{\infty} \Psi_\rho(\Gamma^m).
\end{equation*}

The amplitude Green function for quantum random walk associated with $\rho$ is defined by
\begin{equation}\label{eq:Green rho}
g_{\rho}^{(n)}(\mathbf{x},\mathbf{y})_{j}^{i} =\sum_{t=1}^{\infty} \Psi_\rho (w_0 = \mathbf{x}^i , w_{t} = \mathbf{y}^j ,\tau^{(n)} = t) .
\end{equation}

Note that we have omitted the factor $z$ in the above definition, since $z$ can be absorbed in $\rho$.

With these definitions, we have 
\begin{eqnarray}\label{en:general Green function}
g^{(n)}(z)(\mathbf{x},\mathbf{y})_{j}^{i}=g_{\rho^{(0)}}^{(n)}(\mathbf{x},\mathbf{y})_{j}^{i},
\end{eqnarray}
where $\rho^{(0)}=z\phi$.

\subsubsection{Main results}

If the walk is localized in the sense that $\lim_{n \to \infty}P^{(n)} (0, 0) =1$, we want to find the localization exponents and, in particular, in each direction, defined as follows.
\begin{definition}\label{def:localization exponent}
$$\beta_w(0, 0)^i_j =\lim_{n \to \infty}  \frac{-\ln |P^{(n)} (0, 0)^i_j-P^{(\infty)} (0, 0)^i_j|}{n \ln 2},$$
$$\gamma_w (0, a_n)^i_ j=\lim_{n \to \infty}  \frac{-\ln P^{(n)} (0, a_n)^i_j}{n \ln 2},$$
where $P^{(\infty)} (0, 0)^i_j = \lim_{n\to  \infty} P^{(n)} (0, 0)^i_j$. $\beta_w$ and $\gamma_w$ are called 
 the localization exponents. They are $4 \times 4$ matrices. 
Here $P^{(\infty)}(0, 0)^i_j$ is the probability of the walk starting at the initial state $0^i$   and return to $0^j$  before reaching the boundary of level $n$ fractal.   $P^{(n)}(0, a_n) ^i_j$ is the probability of the walk reaching the boundary point $a{_n}^{j}$ of level $n$ fractal before coming back to the initial point; see (\ref{eq:quantum transition probability}) for their  definitions.  
\end{definition}

In the following definition, note that if $x \in F^{(1)} \cup F^{(1)'}$. Then $2^n x \in F^{(n)} \cup F^{(n)'}$. 
For $x \in F^{(1)} \cup F^{(1)'}$, 
let $p^{(n)} (0, x)^i_j=P^{(n)} (0, 2^nx)^i_j$ and $p^{(\infty)} (0, x)^i_j=\lim_{n \to \infty} p^{(n)} (0, x)^i_j$  be the re-scaled limiting distribution of the first passage distribution of level $n$ as $n \to \infty$. 
\begin{definition}\label{def:d exponent} Let $d^{(n)}=|| p^{(n)}-p^{(\infty)}||_{TV}=\frac{1}{2} \sum_{i, j}\sum_{x}|p^{(n)}(0, x)^i_j-p^{(\infty)}(0, x)^i_j|$ be the total variation distance of $p^{(n)}$ and $p^{(\infty)}$. Then we define 
$\delta_w=\lim_{n \to \infty}  \frac{-\ln d^{(n)}}{n \ln 2}$.
\end{definition}

\begin{definition}\label{def:entropy exponent} Let $H^{(n)}=-\sum_{i, j} \sum_{x}p^{(\infty)}(0, x)^i_j \ln p^{(n)}(0, x)^i_j$ be the relative entropy of $p^{(n)} $ with respect to $p^{(\infty)}$. Then we define 
$\eta_w=\lim_{n \to \infty}  \frac{-\ln H^{(n)}}{n \ln 2}$.
\end{definition}

Our results for quantum random walks on Sierpinski gaskets are summarized and compared to known results in the following table.

\begin{table}
\begin{center}
\begin{tabular}{ |p{2cm}||p{6.5cm}|p{6.5cm}|} 
\hline
\textbf{} & \textbf{Quantum} & \textbf{Classical} \\
 \hline   
 \hline

${P^{(\infty)}}(\mathbf{0},\mathbf{0})$ & \begin{scriptsize} $\bordermatrix{~& 0 & 1 & 4              & 5           \cr
                                                       0 & 0 & 0 & 0.25 & 0.25     \cr
                                                       1 & 0 & 0 & 0.25 & 0.25  \cr
                                                       4 & 0 & 0 & 0.25  & 0.25   \cr
                                                       5 & 0 & 0 & 0.25 & 0.25      \cr}$ \end{scriptsize}   & \begin{scriptsize}  $\bordermatrix{~& 0 & 1 & 4              & 5           \cr
                                                       0 & 0 & 0 & 0.25& 0.25     \cr
                                                       1 & 0 & 0 & 0.25 & 0.25  \cr
                                                       4 & 0 & 0 & 0.25 & 0.25    \cr
                                                       5 & 0 & 0 & 0.25 & 0.25      \cr}$  \end{scriptsize}  \\ 

\hline

${P^{(\infty)}}(\mathbf{0},\mathbf{a_n})$ & \begin{scriptsize} $\bordermatrix{~& 0 & 1 & 4              & 5           \cr
                                                       0 & 0 & 0 & 0 & 0    \cr
                                                       1 & 0 & 0 & 0 & 0  \cr
                                                       4 & 0 & 0 & 0  & 0    \cr
                                                       5 & 0 & 0 & 0 & 0     \cr}$ \end{scriptsize}   & \begin{scriptsize}  $\bordermatrix{~& 0 & 1 & 4              & 5           \cr
                                                       0 & 0 & 0 & 0 & 0     \cr
                                                       1 & 0 & 0 & 0 & 0  \cr
                                                       4 & 0 & 0 & 0 & 0    \cr
                                                       5 & 0 & 0 & 0 & 0      \cr}$  \end{scriptsize}  \\ 

\hline

${\gamma_w}(\mathbf{0},\mathbf{a_n})$ & \begin{scriptsize} $\bordermatrix{~& 0 & 1 & 4              & 5           \cr
                                                       0 & NA & NA & 2.3011 & 2.4526      \cr
                                                       1 & NA & NA & 2.3011 & 2.4526   \cr
                                                       4 & NA & NA & 0.8547  & 0.8668    \cr
                                                       5 & NA & NA & 0.8547 & 0.8668      \cr}$ \end{scriptsize}   & \begin{scriptsize}  $\bordermatrix{~& 0 & 1 & 4              & 5           \cr
                                                       0 & NA & NA & 1.1758 & 0.5667      \cr
                                                       1 & NA & NA & 1.1758 & 0.5667  \cr
                                                       4 & NA & NA & 1.1758 & 0.5667    \cr
                                                       5 & NA & NA & 1.1758 & 0.5667      \cr}$  \end{scriptsize}  \\ 

\hline
${\beta_w}(\mathbf{0},\mathbf{0})$ & \begin{scriptsize} $\bordermatrix{~& 0 & 1 & 4              & 5           \cr
                                                       0 & 1.0229 & 1.0229 & 0.8609 & 0.8609       \cr
                                                       1 & 1.0229 & 1.0229 & 0.8609 & 0.8609   \cr
                                                       4 & 0.8609 & 0.8609 & 1.0229 & 1.0229  \cr
                                                       5 & 0.8609 & 0.8609 & 1.0229 & 1.0229     \cr} $ \end{scriptsize}&  \begin{scriptsize} $\bordermatrix{~& 0 & 1 & 4              & 5           \cr
                                                       0 & \frac{\ln 5 - \ln 3}{\ln 2} & \frac{\ln 5 - \ln 3}{\ln 2} & \frac{\ln 5 - \ln 3}{\ln 2} & \frac{\ln 5 - \ln 3}{\ln 2}      \cr
                                                       1 & \frac{\ln 5 - \ln 3}{\ln 2} & \frac{\ln 5 - \ln 3}{\ln 2} & \frac{\ln 5 - \ln 3}{\ln 2} & \frac{\ln 5 - \ln 3}{\ln 2}   \cr
                                                       4 & \frac{\ln 5 - \ln 3}{\ln 2}& \frac{\ln 5 - \ln 3}{\ln 2} & \frac{\ln 5 - \ln 3}{\ln 2} & \frac{\ln 5 - \ln 3}{\ln 2}   \cr
                                                       5 & \frac{\ln 5 - \ln 3}{\ln 2} & \frac{\ln 5 - \ln 3}{\ln 2} & \frac{\ln 5 - \ln 3}{\ln 2} & \frac{\ln 5 - \ln 3}{\ln 2}         \cr}$ \end{scriptsize}  \\  
                                                   
\hline
$\delta_w$  &   $0.9240$                                    & $\frac{\ln 5 - \ln 3}{\ln 2} \approx 0.7370$  \\ 
\hline
$\eta_w$  &  $1.1455$    & $0.9412$   \\ 
\hline    
\end{tabular}
\caption{\label{tab:localization table}Localization exponents}
\end{center}
\end{table}

\newpage

\subsection{The dimension reduction} \label{subsec:dimension reduction}

The matrix $g^{(n)}(z)$ contains $12\times12$ entries, with each one defined as in (2.3). In this subsection, we will exploit the symmetries of the graph and $G$ to reduce its dimensions when the transition function $\rho^{(0)}=z \phi$. The following lemma shows that it  requires only 6 free variables. 

\begin{lemma} \label{le:g reduction sierpinski}
There exist functions $u_1^{(n)}, u_2^{(n)}, u_3^{(n)}, u_4^{(n)}, u_5^{(n)}, u_6^{(n)}$ such that
$$
{g^{(n)}(\mathbf{0},\mathbf{a_n})} = \bordermatrix{~& 0 & 1 & 4              & 5           \cr
                                                       0 & 0 & 0 & u^{(n)}_4 & u^{(n)}_3       \cr
                                                       1 & 0 & 0 & -u^{(n)}_4 & -u^{(n)}_3   \cr
                                                       4 & 0 & 0 & u^{(n)}_1 & u^{(n)}_2    \cr
                                                       5 & 0 & 0 & u^{(n)}_1 & u^{(n)}_2       \cr},
$$
$$
{g^{(n)}(\mathbf{a_n},\mathbf{0})} = \bordermatrix{~& 0 & 1 & 4              & 5           \cr
                                                       0& u^{(n)}_2 & u^{(n)}_1 & 0 & 0       \cr
                                                       1& u^{(n)}_2 & u^{(n)}_1 & 0 & 0   \cr
                                                       4&-u^{(n)}_3 &-u^{(n)}_4 & 0 & 0    \cr
                                                       5&u^{(n)}_3 & u^{(n)}_4 & 0 & 0       \cr},
$$
$$
{g^{(n)}(\mathbf{0},\mathbf{b_n})} = \bordermatrix{~& 0 & 1 & 2              & 3           \cr
                                                       0 & 0 & 0 & -u^{(n)}_3 & -u^{(n)}_4       \cr
                                                       1 & 0 & 0 & u^{(n)}_3 & u^{(n)}_4   \cr
                                                       4 & 0 & 0 & u^{(n)}_2 & u^{(n)}_1    \cr
                                                       5 & 0 & 0 & u^{(n)}_2 & u^{(n)}_1       \cr},
$$
$$
{g^{(n)}(\mathbf{b_n},\mathbf{0})} = \bordermatrix{~& 0 & 1 & 4              & 5           \cr
                                                       0&u^{(n)}_1 &u^{(n)}_2 & 0 & 0       \cr
                                                       1& u^{(n)}_1 & u^{(n)}_2 & 0 & 0   \cr
                                                       2&u^{(n)}_4 &u^{(n)}_3 & 0 & 0    \cr
                                                       3&-u^{(n)}_4 & -u^{(n)}_3 & 0 & 0       \cr},
$$
$$
{g^{(n)}(\mathbf{a_n},\mathbf{b_n})} = \bordermatrix{~& 0 & 1 & 2              & 3           \cr
                                                       0 & 0 & 0 & u^{(n)}_1 & u^{(n)}_2       \cr
                                                       1 & 0 & 0 & u^{(n)}_1 & u^{(n)}_2   \cr
                                                       4 & 0 & 0 & u^{(n)}_4 & u^{(n)}_3    \cr
                                                       5 & 0 & 0 & -u^{(n)}_4 & -u^{(n)}_3      \cr},
$$
$$
{g^{(n)}(\mathbf{b_n},\mathbf{a_n})} = \bordermatrix{~& 0 & 1 & 4              & 5           \cr
                                                       0 & 0 & 0 & u^{(n)}_2 & u^{(n)}_1       \cr
                                                       1 & 0 & 0 & u^{(n)}_2 & u^{(n)}_1   \cr
                                                       2 & 0 & 0 & -u^{(n)}_3 & -u^{(n)}_4    \cr
                                                       3 & 0 & 0 & u^{(n)}_3 & u^{(n)}_4      \cr},
$$
$$
{g^{(n)}(\mathbf{0},\mathbf{0})}= \bordermatrix{~& 0 & 1 & 4              & 5           \cr
                                                       0 & u^{(n)}_5 & -u^{(n)}_5 & u^{(n)}_6 & u^{(n)}_6       \cr
                                                       1 & -u^{(n)}_5 & u^{(n)}_5 & u^{(n)}_6 & u^{(n)}_6   \cr
                                                       4 & u^{(n)}_6 & u^{(n)}_6 & u^{(n)}_5 & -u^{(n)}_5    \cr
                                                       5 & u^{(n)}_6 & u^{(n)}_6 & -u^{(n)}_5 & u^{(n)}_5       \cr},
$$
$$
{g^{(n)}(\mathbf{a_n},\mathbf{a_n})}= \bordermatrix{~& 0 & 1 & 4              & 5           \cr
                                                       0 & u^{(n)}_5 & -u^{(n)}_5 & u^{(n)}_6 & u^{(n)}_6       \cr
                                                       1 & -u^{(n)}_5 & u^{(n)}_5 & u^{(n)}_6 & u^{(n)}_6   \cr
                                                       4 & u^{(n)}_6 & u^{(n)}_6 & u^{(n)}_5 & -u^{(n)}_5    \cr
                                                       5 & u^{(n)}_6 & u^{(n)}_6 & -u^{(n)}_5 & u^{(n)}_5       \cr}, $$
$$
{g^{(n)}(\mathbf{b_n},\mathbf{b_n})} = \bordermatrix{~& 0 & 1 & 2              & 3          \cr
                                                       0 & u^{(n)}_5 & -u^{(n)}_5 & u^{(n)}_6 & u^{(n)}_6       \cr
                                                       1 & -u^{(n)}_5 & u^{(n)}_5 & u^{(n)}_6 & u^{(n)}_6   \cr
                                                       2 & u^{(n)}_6 & u^{(n)}_6 & u^{(n)}_5 & -u^{(n)}_5    \cr
                                                       3 & u^{(n)}_6 & u^{(n)}_6 & -u^{(n)}_5 & u^{(n)}_5       \cr}. $$

\end{lemma}

\subsection{The recursive formula} 

 In this subsection, we will  find a recursive relation for $g^{(n)}$.

Given $g^{(n)} (z) (x, y)^i_j$, $x, y \in \{\mathbf{0}, \mathbf{a_n}, \mathbf{b_n}, \mathbf{a_n^{\prime}}, \mathbf{b_n^{\prime}} \}$, we let the $n$th level amplitudes 
\begin{eqnarray}\label{eq:rho n}
\rho^{(n)}(x^i, y^j)=g^{(n)}(z)(2^n x, 2^n y)^i_j,
\end{eqnarray}
for $x, y\in \{\mathbf{0}, \mathbf{a_0}, \mathbf{b_0}, \mathbf{a_0^{\prime}}, \mathbf{b_0^{\prime}} \}$, for all $n \ge 1$.   Then we extend  $\rho^{(n)}(x^i, y^j)$ by parallel translations to all $x, y \in F^{(\infty)} \cup F^{(\infty)'}$, with $||x-y||=0$ or $ 2$. So $g_{\rho ^{(n)}}^{(1)}(x, y)^i_j$ is well-defined for all $x, y \in F^{(1)} \cup F^{(1)'}$.

For convenience, we also write   $\rho^{(n)}(x^i, y^j)=\rho^{(n)}(x, y)^i_j$.
By the self-similar property of the Sierpinski gasket, we have the following recursive relations.
\begin{theorem} \label{th: g recursive}
(a) For $ n \ge 0$, we have 
\begin{eqnarray}
 g_{\rho^{(0)}}^{(n+1)}(x, y)^i_j = g^{(1)}_{\rho^{(n)}}(\frac{1}{2^n} x, \frac{1}{2^n} y)^i_j,
 \end{eqnarray}
  for all $x  \in 2^n(F^{(1)} \cup F^{(1)'})$, $y  \in 2^{n+1}(F^{(0)} \cup F^{(0)'})$.

(b) \begin{eqnarray}\label{eq:recursive relation rho} 
\rho ^{(n+1)} (x, y)=g^{(1)}_{\rho ^{(n)} }(2x, 2y),
\end{eqnarray}
for all $x, y  \in F^{(0)} \cup F^{(0)'}$.
\end{theorem}

\subsection{The quantum Dirichlet problem and the solution of $g$} 

In view of Theorem \ref{th: g recursive} and (\ref{eq:recursive relation rho}), the iteration map that takes from $\rho^{(n)}$ to $\rho^{(n+1)}$ only depends on $g^{(1)}$ and a scaling factor by $2$. In classical random walk, the first exit distribution is the solution of the classical Dirichlet problem.  Its associated Green function is the solution of the corresponding Poisson equations. In quantum random walk, before we calculate the exit probability distributions, we first calculate the exit  amplitude functions. The exit amplitude function is the solution of quantum analog of Dirichlet problem and its associated Green function is the solution of the quantum analog of Poisson equations. 
In this subsection we will solve  for $g_{\rho}^{(1)}$ which is the solution of a quantum mechanical version of Poisson equation (Theorem \ref{th:Dirichlet g} below). 

To find $g_{\rho}^{(1)}$, for convenience of notations we consider the following graph on $F^{(1)}$ with vertices relabeled as 
{$\mathbf{0}=\mathbf{0}$, 
$\mathbf{a_0}=\mathbf{1}$, 
$\mathbf{b_0}=\mathbf{2}$, 
$\mathbf{a_0+b_0}=\mathbf{3}$, 
$\mathbf{a_1}=\mathbf{5}$, 
$\mathbf{b_1}=\mathbf{4}$, see Fig \ref{fig:relabel on F^{(1)}}.

By definition, 
$g_{\rho}^{(1)}(z)(\mathbf{x},\mathbf{y})_{j}^{i} =\sum_{t=1}^{\infty}\Psi _ \rho (w_0 = \mathbf{x}^i , w_{t} = \mathbf{y}^j ,\tau^{(1)} = t)$, where ${\tau^{(1)}} = \inf \{ t\geq 1; w_t=x_t^{i_t}, x_t  \in \{\mathbf{0},\mathbf{4},\mathbf{5}\}\}.$
\begin{figure}[htb]  
\begin{center}
\begin{tikzpicture}[scale = 0.7]
\begin{axis}[dashed, xmin=-1,ymin=-1,xticklabels = {,,}, yticklabels = {},xmax=5,ymax=3,]
\addplot[mark = none] coordinates{(-1,0) (5,0)};
\addplot[mark = none] coordinates{(-1,1) (5,1)};
\addplot[mark = none] coordinates{(-1,2) (5,2)};
\addplot[mark = none] coordinates{(0,-1) (0,3)};
\addplot[mark = none] coordinates{(1,-1) (1,3)};
\addplot[mark = none] coordinates{(2,-1) (2,3)};
\addplot[mark = none] coordinates{(3,-1) (3,3)};
\addplot[mark = none] coordinates{(4,-1) (4,3)};
\addplot[mark = *, solid, thick] coordinates {(0, 0) (2,0) (4,0) (3,1) (2,2) (1,1) (0,0)};
\addplot[mark = *, solid ,thick] coordinates {(1,1) (3,1) (2,0) (1,1)};
\addplot[mark=*] coordinates {(0,0)} node(1)[label={$\mathbf{0}$}]{} ;
\addplot[mark=*] coordinates {(1,1)} node(2)[label={$\mathbf{1}$}]{} ;
\addplot[mark=*] coordinates {(2,0)} node(3)[label={$\mathbf{2}$}]{} ;
\addplot[mark=*] coordinates {(3,1)} node(4)[label={$\mathbf{3}$}]{} ;
\addplot[mark=*] coordinates {(2,2)} node(5)[label={$\mathbf{5}$}]{} ;
\addplot[mark=*] coordinates {(4,0)} node(6)[label={$\mathbf{4}$}]{} ;
\addplot[mark=none] coordinates {(-1,-1)} node(-1)[]{} ;
\addplot[mark=none] coordinates {(1,-1)} node(-2)[]{} ;
\draw [->] (1) -- (3)node[xshift=-10pt,yshift=5pt] {$e_0$};
\draw [->] (1) -- (2)node[xshift=0pt,yshift=-10pt] {$e_1$};
\draw [->] (1) -- (-1)node[xshift=5pt,yshift=15pt] {$e_4$};
\draw [->] (1) -- (-2)node[xshift=-5pt,yshift=15pt] {$e_5$};
\addplot[mark=none] coordinates {(3,2)} node(32)[]{} ;
\addplot[mark=none] coordinates {(3,3)} node(33)[]{} ;
\draw [->] (5) -- (32)node[xshift=-10pt,yshift=5pt] {$e_0$};
\draw [->] (5) -- (33)node[xshift=0pt,yshift=-10pt] {$e_1$};
\draw [->] (5) -- (2)node[xshift=12pt,yshift=5] {$e_4$};
\draw [->] (5) -- (4)node[xshift=-12pt,yshift=5pt] {$e_5$};
\addplot[mark=none] coordinates {(5,0)} node(50)[]{} ;
\addplot[mark=none] coordinates {(5,1)} node(51)[]{} ;
\draw [->] (6) -- (50)node[xshift=-10pt,yshift=5pt] {$e_0$};
\draw [->] (6) -- (51)node[xshift=-5pt,yshift=-15pt] {$e_1$};
\draw [->] (6) -- (4)node[xshift=10pt,yshift=-5] {$e_2$};
\draw [->] (6) -- (3)node[xshift=12pt,yshift=5pt] {$e_3$};
\end{axis}        
\end{tikzpicture}
\caption{relabel on $F^{(1)}$}\label{fig:relabel on F^{(1)}}
\end{center}
\end{figure}
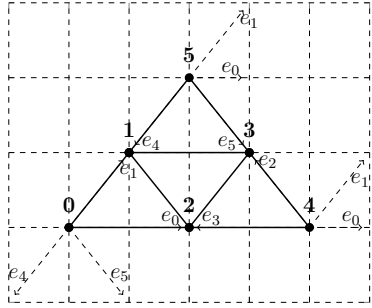

In the following theorem, Part (a)  is a quantum analog  of the classical Poisson equation and  Part (b) is its solution. 

\begin{theorem}\label{th:Dirichlet g}
(a) We have 
\begin{eqnarray}\label{eq:Poisson eq}
[I_{12}-\tilde {\rho}]g^{(1)}_{\rho}(\mathbf{x}, \mathbf{y})=\rho(\mathbf{x}, \mathbf{y}),
\end{eqnarray}
for all $\mathbf{x}=\mathbf{1}, \mathbf{2}, \mathbf{3}$, $\mathbf{y}=\mathbf{0}, \mathbf{4}, \mathbf{5}$, here $\tilde {\rho}$ is the $12 \times 12$ matrix obtained from $\rho $  with restriction  to components involving only $\mathbf{1}, \mathbf{2}, \mathbf{3}$ and all $k$ and $I_{12}$ is the $12 \times 12$ identity matrix.

(b) If $I_{12}-\tilde {\rho}$ is invertible, then for all $\mathbf{y}=\mathbf{0}, \mathbf{4}, \mathbf{5} $, we have 
\begin{equation}
 \begin{bmatrix}
g_{\rho}^{(1)}(\mathbf{1},\mathbf{y})  \\ g_{\rho}^{(1)}(\mathbf{2},\mathbf{y}) \\ g_{\rho}^{(1)}(\mathbf{3},\mathbf{y})
\end{bmatrix} = \frac{1}{I_{12}-\tilde {\rho}} \begin{bmatrix}
\rho(\mathbf{1},\mathbf{y})  \\ \rho(\mathbf{2},\mathbf{y})  \\ \rho(\mathbf{3},\mathbf{y})
\end{bmatrix}.
\end{equation}

\end{theorem}

\begin{remark} (Poisson equation)
If we let 
 $$\rho (\mathbf{x}^i, \mathbf{y}^j) = 
\left\{
        \begin{array}{ll}
            \frac{1}{4} & if \ \Vert \mathbf{x}-\mathbf{y} \Vert = 2, j = i+3\ (mod\ 6) \\
            \frac{1}{4} & if \ \Vert \mathbf{x}-\mathbf{y} \Vert = 2, j \neq i+3\ (mod\ 6), \mathbf{y} = \mathbf{x} + \mathbf{e}_k, j = k+3\ (mod\ 6) \\
            0 & otherwise
        \end{array}
    \right.$$
 then (\ref{eq:Poisson eq}) becomes the classical Poisson equation
\begin{eqnarray}\label{eq:classical Poisson eq}
\Delta g^{(1)}_{\rho}(\mathbf{x^i}, \mathbf{y^j})=-\rho(\mathbf{x^i}, \mathbf{y^j}),
\end{eqnarray}
for all $\mathbf{x}=\mathbf{1}, \mathbf{2}, \mathbf{3}$, $\mathbf{y}=\mathbf{0}, \mathbf{4}, \mathbf{5}$,  $i \in out(\mathbf{x}), j \in out(\mathbf{y})$, and  $\Delta $ is the discrete Laplace operator on  $\{\mathbf{x}^i; \mathbf{x} \in F^{(1)} \cup F^{(1)'}, i \in out(\mathbf{x})\}$  with $0$ boundary conditions on $\{\mathbf{x}^i; \mathbf{x} \in \partial F^{(1)} \cup \partial F^{(1)'}, i \in out(\mathbf{x})\}$.  

\end{remark}

The above Theorem gives the solution of $g_{\rho}^{(1)}(\mathbf{x},\mathbf{y})$ for interior points $\mathbf{x}$ of $F^{(1)} \cup F^{(1)'}$. For the boundary points $\mathbf{x}$, we use a similar argument as in the proof of Part (a) Theorem \ref{th:Dirichlet g} to obtain
\begin{eqnarray}\label{eq:g1 boundary 1}
g_\rho^{(1)}(\mathbf{0},\mathbf{5})
= \rho (\mathbf{0},\mathbf{1})g_\rho^{(1)}(\mathbf{1},\mathbf{5}) +\rho (\mathbf{0},\mathbf{2})g_\rho^{(1)}(\mathbf{2},\mathbf{5}).
\end{eqnarray}
For $i = 0,1,4,5$, $j = 0,1$,
\begin{eqnarray}\label{eq:g1 boundary 2}
g_\rho^{(1)}(\mathbf{0},\mathbf{0})^i_j 
 = \rho(\mathbf{0},\mathbf{0})^i_j + [\rho(\mathbf{0},\mathbf{1})g_\rho^{(1)}(\mathbf{1},\mathbf{0})]^i_j + [\rho(\mathbf{0},\mathbf{2})g_\rho^{(1)}(\mathbf{2},\mathbf{0})]^i_j.
\end{eqnarray}
Using the reflection $F^{(1)'}$ from $F^{(1)}$, we also have 
for $i = 0,1,4,5$, $j = 4,5$,
\begin{eqnarray}\label{eq:g1 boundary 3}
g_\rho^{(1)}(\mathbf{0},\mathbf{0})^i_j = \rho(\mathbf{5},\mathbf{5})^i_j + [\rho(\mathbf{5},\mathbf{1})g_\rho^{(1)}(\mathbf{1},\mathbf{5})]^i_j + [\rho(\mathbf{5},\mathbf{3})g_\rho^{(1)}(\mathbf{3},\mathbf{5})]^i_j.
\end{eqnarray} 
Other boundary points can be treated similarly. However, (\ref{eq:g1 boundary 1}) - (\ref{eq:g1 boundary 3}) will be  sufficient for our recursive procedure.

\subsection{The inversion of $I_{12}-\tilde {\rho}$} 

We now consider the initial transition function  $\rho^{(0)}=z\phi$. By Lemma \ref{le:g reduction sierpinski} and (\ref{eq:rho n}), $\rho^{(n)}=\rho^{(n)}(u^{(n)}_1,u^{(n)}_2,u^{(n)}_3,u^{(n)}_4,u^{(n)}_5,u^{(n)}_6)$ is a function of $u^{(n)}_1$ - $u^{(n)}_6$.
By Theorem \ref{th: g recursive} and (\ref{eq:recursive relation rho}), 
there is an iteration map   $T:\mathbb{C}^6 \rightarrow \mathbb{C}^6$, such that 
\begin{eqnarray}\label{eq:iteration map T}
(u^{(n+1)}_1,u^{(n+1)}_2,u^{(n+1)}_3,u^{(n+1)}_4,u^{(n+1)}_5,u^{(n+1)}_6) = T(u^{(n)}_1,u^{(n)}_2,u^{(n)}_3,u^{(n)}_4,u^{(n)}_5,u^{(n)}_6).
\end{eqnarray}
 By Theorem \ref{th: g recursive}, (\ref{eq:recursive relation rho}) and (\ref{eq:rho tilde n formula}) below, existence of $T$ and iteration of $\rho^{(n)}$ are equivalent.  Our goal is to find the map $T$. In view of Theorem \ref{th:Dirichlet g}, $I_{12}-\tilde{ \rho} $ is a $12\times 12$ matrix,  so we will need to find an efficient method  to compute its inverse for the first step  iteration. In this subsection we will  find an efficient method to compute the inverse of $I_{12}-\tilde{ \rho}^{(0)} $.

\begin{lemma}\label{le:fast inversion}
Suppose we reorder the matrix by switching $\mathbf{2}^0, \mathbf{2}^1$ with $\mathbf{2}^2, \mathbf{2}^3$ and $\mathbf{3}^2, \mathbf{3}^3$ with $\mathbf{3}^4,\mathbf{3}^5$ for both rows and columns. Then  

(a) $[\tilde {\rho}^{(n)}] $ has the following form $[\tilde {\rho}^{(n)}] = $
$$
\left[
\begin{array}{c|c|c}
A & B & C \\
\hline
C & A& B\\
\hline
B & C & A\\
\end{array}
\right].$$

(b) $-[I_{12} - \tilde {\rho}^{(n)}]^{-1}=$
$$
\left[
\begin{array}{c|c|c}
X & Y & Z\\
\hline
Z & X & Y\\
\hline
Y & Z & X\\
\end{array}
\right],$$
where 
\begin{eqnarray}
X&=&\bar{A}^{-1}(I_4-BZ-C Y)\\
Y&=&-(\bar{A}-B\bar{A}^{-1}C)^{-1}(C-B\bar{A}^{-1}B)Z-(\bar{A}-B\bar{A}^{-1}C)^{-1}B\bar{A}^{-1}\\
Z&=&H^{-1}D\\
H&=&\bar{A}-C\bar{A}^{-1}B-(B-C\bar{A}^{-1}C)(\bar{A}-B\bar{A}^{-1}C)^{-1}(C-B\bar{A}^{-1}B)\\
D&=&(B-C\bar{A}^{-1}C)(\bar{A}-B\bar{A}^{-1}C)^{-1}B\bar{A}^{-1}-C\bar{A}^{-1}\\
\bar{A}&=&A-I_4.
\end{eqnarray}

\end{lemma}

Note that in the above Lemma, Part (b), we only need to compute the inverse of a $4\times 4$ matrix instead of a $12\times 12 $ matrix.

\subsection{Beginning  iterations and local behavior}

 In this subsection, we will calculate the first two iterations and obtain exact solutions for the exit distributions for  n=1 and 2. These results show the local behavior of quantum random walks in a neighborhood of the origin.

 As an initial step, we find $u^{(0)}$ as follows. 
For $F^{(0)}$, 
$$
\rho^{(0)}(\mathbf{0},\mathbf{a_0})={g^{(0)}(\mathbf{0},\mathbf{a_0})} = \bordermatrix{~& 0 & 1 & 4              & 5           \cr
                                                       0 & 0 & 0 & rz & 0       \cr
                                                       1 & 0 & 0 & -rz & 0   \cr
                                                       4 & 0 & 0 & rz & 0    \cr
                                                       5 & 0 & 0 & rz & 0       \cr},
$$ 
$$
{\rho^{(0)}(\mathbf{0},\mathbf{0})}={g^{(0)}(\mathbf{0},\mathbf{0})} = \bordermatrix{~& 0 & 1 & 4              & 5           \cr
                                                       0 & 0 & 0 & 0 & 0       \cr
                                                       1 & 0 & 0 & 0 & 0   \cr
                                                       4 & 0 & 0 & 0& 0    \cr
                                                       5 & 0 & 0 & 0 & 0       \cr}.
$$
By Lemma \ref{le:g reduction sierpinski} and comparing  with ${g^{(n)}(\mathbf{0},\mathbf{a_0})} $ and ${g^{(n)}(\mathbf{0},\mathbf{0})}$ with $n=0$, we have

\begin{eqnarray}\label{eq:u0}
u^{(0)}_1 = rz,
 u^{(0)}_2 =0,
 u^{(0)}_3 = 0,
 u^{(0)}_4 = rz,
 u^{(0)}_5 = 0, 
 u^{(0)}_6 = 0.
\end{eqnarray}
The following theorem shows the results of the first iteration. 
\begin{theorem}\label{th:Sierpinski step 1}
Given initial $u^{(0)}_1 = rz,u^{(0)}_2 =0,u^{(0)}_3=0,u^{(0)}_4=rz,u^{(0)}_5=0,u^{(0)}_6=0$, we have
\begin{eqnarray}
u^{(1)}_1(z) &=& r^3z^3+r^2z^2-\frac{r^4z^4}{1+rz},\\
 u^{(1)}_2(z) &= & 2r^3z^3+\frac{2r^4z^4}{1+rz},\\
u^{(1)}_3(z) &=& 0,\\
u^{(1)}_4(z) &=& -r^3z^3+r^2z^2-\frac{3r^4z^4}{1+rz},\\
u^{(1)}_5(z) &=& r^3z^3+r^2z^2+\frac{3r^4z^4}{1+rz},\\
u^{(1)}_6(z) &=& r^3z^3-r^2z^2-\frac{r^4z^4}{1+rz}.
\end{eqnarray}

\end{theorem}

Now we compute the first hitting  probability distribution for quantum random walks at $\partial F^{(1)}\cup \partial F^{(1)'}$.
By Theorem \ref{th:Sierpinski step 1} and (\ref{eq:probability parseval}), we have the exact exit probability distributions:
\begin{align*}
P^{(1)}(\mathbf{0},\mathbf{a_1}) &=\frac{1}{2\pi}\int_0^{2\pi} \bordermatrix{~& 0 & 1 & 4              & 5           \cr
                                                       0  &0&0& |u^{(1)}_4(e^{i\theta})|^2 &|u^{(1)}_3(e^{i\theta})|^2    \cr
                                                       1  &0&0& |u^{(1)}_4(e^{i\theta})|^2 &|u^{(1)}_3(e^{i\theta})|^2  \cr
                                                       4  &0&0& |u^{(1)}_1(e^{i\theta})|^2 & |u^{(1)}_2(e^{i\theta})|^2\cr
                                                       5  &0&0& |u^{(1)}_1(e^{i\theta})|^2 & |u^{(1)}_2(e^{i\theta})|^2\cr} d\theta \\
&= \frac{1}{2\pi}\int_0^{2\pi} \bordermatrix{~& 0 & 1 & 4              & 5           \cr
                                                       0  &0&0& \frac{\sin^2\theta}{5+4\cos\theta} &0    \cr
                                                       1  &0&0& \frac{\sin^2\theta}{5+4\cos\theta} &0  \cr
                                                       4  &0&0& \frac{1}{2}\frac{1+\cos\theta}{5+4\cos\theta} & \frac{1}{2}\frac{1+\cos\theta}{5+4\cos\theta}\cr
                                                       5  &0&0& \frac{1}{2}\frac{1+\cos\theta}{5+4\cos\theta} & \frac{1}{2}\frac{1+\cos\theta}{5+4\cos\theta}\cr} d\theta \\
 &=\bordermatrix{~& 0 & 1 & 4              & 5           \cr
                                                       0  &0&0& \frac{1}{8} &0    \cr
                                                       1  &0&0& \frac{1}{8} &0  \cr
                                                       4  &0&0& \frac{1}{12} & \frac{1}{12}\cr
                                                       5  &0&0& \frac{1}{12} & \frac{1}{12}\cr},
\end{align*}
and
\begin{align*}
P^{(1)}(\mathbf{0},\mathbf{0}) &=\frac{1}{2\pi}\int_0^{2\pi} \bordermatrix{~& 0 & 1 & 4              & 5           \cr
                                                       0  &|u^{(1)}_5(e^{i\theta})|^2&|u^{(1)}_5(e^{i\theta})|^2& |u^{(1)}_6(e^{i\theta})|^2 &|u^{(1)}_6(e^{i\theta})|^2    \cr
                                                       1  &|u^{(1)}_5(e^{i\theta})|^2&|u^{(1)}_5(e^{i\theta})|^2& |u^{(1)}_6(e^{i\theta})|^2 &|u^{(1)}_6(e^{i\theta})|^2  \cr
                                                       4  &|u^{(1)}_6(e^{i\theta})|^2&|u^{(1)}_6(e^{i\theta})|^2& |u^{(1)}_5(e^{i\theta})|^2 & |u^{(1)}_5(e^{i\theta})|^2\cr
                                                       5  &|u^{(1)}_6(e^{i\theta})|^2&|u^{(1)}_6(e^{i\theta})|^2& |u^{(1)}_5(e^{i\theta})|^2 & |u^{(1)}_5(e^{i\theta})|^2\cr} d\theta \\
&= \frac{1}{2\pi}\int_0^{2\pi} \bordermatrix{~& 0 & 1 & 4              & 5           \cr
                                                       0  &\frac{1}{4}-\frac{\sin^2\theta}{5+4\cos\theta}&\frac{1}{4}-\frac{\sin^2\theta}{5+4\cos\theta}&\frac{1}{4}\frac{1}{5+4\cos\theta}&\frac{1}{4}\frac{1}{5+4\cos\theta}    \cr
                                                       1  &\frac{1}{4}-\frac{\sin^2\theta}{5+4\cos\theta}&\frac{1}{4}-\frac{\sin^2\theta}{5+4\cos\theta}& \frac{1}{4}\frac{1}{5+4\cos\theta} &\frac{1}{4}\frac{1}{5+4\cos\theta} \cr
                                                       4  &\frac{1}{4}\frac{1}{5+4\cos\theta}&\frac{1}{4}\frac{1}{5+4\cos\theta}& \frac{1}{4}-\frac{\sin^2\theta}{5+4\cos\theta} & \frac{1}{4}-\frac{\sin^2\theta}{5+4\cos\theta}\cr
                                                       5  &\frac{1}{4}\frac{1}{5+4\cos\theta}&\frac{1}{4}\frac{1}{5+4\cos\theta}& \frac{1}{4}-\frac{\sin^2\theta}{5+4\cos\theta} & \frac{1}{4}-\frac{\sin^2\theta}{5+4\cos\theta}\cr} d\theta \\
 &=\bordermatrix{~& 0 & 1 & 4              & 5           \cr
                                                       0  &\frac{1}{8}&\frac{1}{8}& \frac{1}{12}&\frac{1}{12}    \cr
                                                       1  &\frac{1}{8}&\frac{1}{8}& \frac{1}{12}&\frac{1}{12}  \cr
                                                       4  &\frac{1}{12}&\frac{1}{12}& \frac{1}{8} & \frac{1}{8}\cr
                                                       5  &\frac{1}{12}&\frac{1}{12}& \frac{1}{8} & \frac{1}{8}\cr}.
\end{align*}
Suppose $\psi_0 = |\mathbf{0}^0\rangle $, then distribution of the position at first exit from $F^{(1)}\cup F^{(1)'}$, $(\mathbf{x}^i)_{\tau^{(1)}}$ are  shown in Figure \ref{fig:exit amplitude distributions 1} and Figure \ref{fig:exit distributions 1}.   
\begin{figure}[htb]
\begin{minipage}{0.45\textwidth}
\begin{tikzpicture}[scale = 0.7]
\begin{axis}[dashed, xmin=-3,ymin=-3,xticklabels = {,,}, yticklabels = {},xmax=5,ymax=3,]
\addplot[mark = none] coordinates{(-5,0) (5,0)};
\addplot[mark = none] coordinates{(-5,1) (5,1)};
\addplot[mark = none] coordinates{(-5,2) (5,2)};
\addplot[mark = none] coordinates{(-5,-1) (5,-1)};
\addplot[mark = none] coordinates{(-5,-2) (5,-2)};
\addplot[mark = none] coordinates{(0,-3) (0,3)};
\addplot[mark = none] coordinates{(1,-3) (1,3)};
\addplot[mark = none] coordinates{(2,-3) (2,3)};
\addplot[mark = none] coordinates{(3,-3) (3,3)};
\addplot[mark = none] coordinates{(4,-3) (4,3)};
\addplot[mark = none] coordinates{(-1,-3) (-1,3)};
\addplot[mark = none] coordinates{(-2,-3) (-2,3)};
\addplot[mark = none] coordinates{(-3,-3) (-3,3)};
\addplot[mark = none] coordinates{(-4,-3) (-4,3)};
\addplot[mark = *, solid, thick] coordinates {(0, 0) (2,0) (4,0) (3,1) (2,2) (1,1) (0,0)};
\addplot[mark = *, solid ,thick] coordinates {(1,1) (3,1) (2,0) (1,1)};
\addplot[mark = *, solid, thick] coordinates {(0, 0) (-1,-1) (-2,-2) (0,-2) (2,-2) (1,-1) (0,0)};
\addplot[mark = *, solid ,thick] coordinates {(-1,-1) (0,-2) (1,-1) (-1,-1)};
\addplot[mark=*] coordinates {(0,0)} node(1)[]{} ;
\addplot[mark=] coordinates {(-0.5,0)}[xshift=0pt,yshift=-10pt] node[label={$\mathbf{0}$}]{} ;
\addplot[mark=*] coordinates {(1,1)} node(2)[]{} ;
\addplot[mark=*] coordinates {(2,0)} node(3)[xshift=0pt,yshift=-20pt][]{} ;
\addplot[mark=*] coordinates {(3,1)} node(4)[]{} ;
\addplot[mark=*] coordinates {(2,2)} node(5)[label={$\mathbf{a_1}$}]{} ;
\addplot[mark=*] coordinates {(4,0)} node(6)[]{} ;
\addplot[mark=] coordinates {(4.5,0)} node[xshift=0pt,yshift=-10pt][label={$\mathbf{b_1}$}]{} ;
\addplot[mark=*] coordinates {(-1,-1)} node(-2)[]{} ;
\addplot[mark=*] coordinates {(1,-1)} node(-3)[]{} ;
\addplot[mark=*] coordinates {(0,-2)} node(-4)[]{} ;
\addplot[mark=*] coordinates {(-2,-2)} node(-5)[]{} ;
\addplot[mark=] coordinates {(-2.5,-2)} node[xshift=2pt,yshift=-12pt][label={$\mathbf{a'_1}$}]{} ;
\addplot[mark=*] coordinates {(2,-2)} node(-6)[]{} ;
\addplot[mark=] coordinates {(2.5,-2)} node[xshift=0pt,yshift=-12pt][label={$\mathbf{b'_1}$}]{} ;
\addplot[mark=none] coordinates {(-1,-1)} node(-1)[]{} ;
\addplot[mark=none] coordinates {(1,-1)} node(-2)[]{} ;
\draw [->,thick ,solid] (1) -- (5)node[xshift=-13pt,yshift=-5pt] {$u^{(1)}_4$};
\draw [->,thick ,solid] (6) -- (5)node[xshift=17pt,yshift=-5pt]{$u^{(1)}_3$};
\draw [->,thick ,solid] (1) -- (6)node[xshift=-12pt,yshift=-9pt] {-$u^{(1)}_4$};
\draw [->,thick ,solid] (5) -- (6)node[xshift=1pt,yshift=15pt]{-$u^{(1)}_3$};
\draw [->,thick ,solid] (5) -- (1)node[xshift=2pt,yshift=16pt] {-$u^{(1)}_5$};
\draw [->,thick ,solid] (6) -- (1)node[xshift=20pt,yshift=8pt]{$u^{(1)}_5$};

\draw [->,thick ,solid] (1) -- (-5)node[xshift=6pt,yshift=18pt] {$u^{(1)}_1$};
\draw [->,thick ,solid] (-6) -- (-5)node[xshift=15pt,yshift=-8pt]{$u^{(1)}_2$};
\draw [->,thick ,solid] (1) -- (-6)node[xshift=-2pt,yshift=18pt] {$u^{(1)}_1$};
\draw [->,thick ,solid] (-5) -- (-6)node[xshift=-15pt,yshift=-8pt]{$u^{(1)}_2$};
\draw [->,thick ,solid] (-5) -- (1)node[xshift=-2pt,yshift=-13pt] {$u^{(1)}_6$};
\draw [->,thick ,solid] (-6) -- (1)node[xshift=18pt,yshift=-8pt]{$u^{(1)}_6$};
\end{axis}        
\end{tikzpicture}
\caption{First exit amplitude distribution from $F^{(1)} \cup F^{(1)'}$ starting $|\mathbf{0}^0\rangle $}\label{fig:exit amplitude distributions 1}
\end{minipage}
\hfill
\begin{minipage}{0.45\textwidth}
\begin{tikzpicture}[scale = 0.7]
\begin{axis}[dashed, xmin=-3,ymin=-3,xticklabels = {,,}, yticklabels = {},xmax=5,ymax=3,]
\addplot[mark = none] coordinates{(-5,0) (5,0)};
\addplot[mark = none] coordinates{(-5,1) (5,1)};
\addplot[mark = none] coordinates{(-5,2) (5,2)};
\addplot[mark = none] coordinates{(-5,-1) (5,-1)};
\addplot[mark = none] coordinates{(-5,-2) (5,-2)};
\addplot[mark = none] coordinates{(0,-3) (0,3)};
\addplot[mark = none] coordinates{(1,-3) (1,3)};
\addplot[mark = none] coordinates{(2,-3) (2,3)};
\addplot[mark = none] coordinates{(3,-3) (3,3)};
\addplot[mark = none] coordinates{(4,-3) (4,3)};
\addplot[mark = none] coordinates{(-1,-3) (-1,3)};
\addplot[mark = none] coordinates{(-2,-3) (-2,3)};
\addplot[mark = none] coordinates{(-3,-3) (-3,3)};
\addplot[mark = none] coordinates{(-4,-3) (-4,3)};
\addplot[mark = *, solid, thick] coordinates {(0, 0) (2,0) (4,0) (3,1) (2,2) (1,1) (0,0)};
\addplot[mark = *, solid ,thick] coordinates {(1,1) (3,1) (2,0) (1,1)};
\addplot[mark = *, solid, thick] coordinates {(0, 0) (-1,-1) (-2,-2) (0,-2) (2,-2) (1,-1) (0,0)};
\addplot[mark = *, solid ,thick] coordinates {(-1,-1) (0,-2) (1,-1) (-1,-1)};
\addplot[mark=*] coordinates {(0,0)} node(1)[]{} ;
\addplot[mark=] coordinates {(-0.5,0)}[xshift=0pt,yshift=-10pt] node[label={$\mathbf{0}$}]{} ;
\addplot[mark=*] coordinates {(1,1)} node(2)[]{} ;
\addplot[mark=*] coordinates {(2,0)} node(3)[xshift=0pt,yshift=-20pt][]{} ;
\addplot[mark=*] coordinates {(3,1)} node(4)[]{} ;
\addplot[mark=*] coordinates {(2,2)} node(5)[label={$\mathbf{a_1}$}]{} ;
\addplot[mark=*] coordinates {(4,0)} node(6)[]{} ;
\addplot[mark=] coordinates {(4.5,0)} node[xshift=0pt,yshift=-10pt][label={$\mathbf{b_1}$}]{} ;
\addplot[mark=*] coordinates {(-1,-1)} node(-2)[]{} ;
\addplot[mark=*] coordinates {(1,-1)} node(-3)[]{} ;
\addplot[mark=*] coordinates {(0,-2)} node(-4)[]{} ;
\addplot[mark=*] coordinates {(-2,-2)} node(-5)[]{} ;
\addplot[mark=] coordinates {(-2.5,-2)} node[xshift=2pt,yshift=-12pt][label={$\mathbf{a'_1}$}]{} ;
\addplot[mark=*] coordinates {(2,-2)} node(-6)[]{} ;
\addplot[mark=] coordinates {(2.5,-2)} node[xshift=0pt,yshift=-12pt][label={$\mathbf{b'_1}$}]{} ;
\addplot[mark=none] coordinates {(-1,-1)} node(-1)[]{} ;
\addplot[mark=none] coordinates {(1,-1)} node(-2)[]{} ;
\draw [->,thick ,solid] (1) -- (5)node[xshift=-13pt,yshift=-5pt] {$\frac{1}{8}$};
\draw [->,thick ,solid] (6) -- (5)node[xshift=14pt,yshift=-5pt]{$0$};
\draw [->,thick ,solid] (1) -- (6)node[xshift=-12pt,yshift=-9pt] {$\frac{1}{8}$};
\draw [->,thick ,solid] (5) -- (6)node[xshift=-3pt,yshift=15pt]{$0$};
\draw [->,thick ,solid] (5) -- (1)node[xshift=2pt,yshift=16pt] {$\frac{1}{8}$};
\draw [->,thick ,solid] (6) -- (1)node[xshift=20pt,yshift=8pt]{$\frac{1}{8}$};

\draw [->,thick ,solid] (1) -- (-5)node[xshift=6pt,yshift=18pt] {$\frac{1}{12}$};
\draw [->,thick ,solid] (-6) -- (-5)node[xshift=15pt,yshift=-8pt]{$\frac{1}{12}$};
\draw [->,thick ,solid] (1) -- (-6)node[xshift=-2pt,yshift=18pt] {$\frac{1}{12}$};
\draw [->,thick ,solid] (-5) -- (-6)node[xshift=-15pt,yshift=-8pt]{$\frac{1}{12}$};
\draw [->,thick ,solid] (-5) -- (1)node[xshift=-2pt,yshift=-13pt] {$\frac{1}{12}$};
\draw [->,thick ,solid] (-6) -- (1)node[xshift=18pt,yshift=-8pt]{$\frac{1}{12}$};
\end{axis}        
\end{tikzpicture}
\caption{First exit probability distribution  from  $F^{(1)} \cup F^{(1)'}$ starting $|\mathbf{0}^0\rangle $}\label{fig:exit distributions 1}
\end{minipage}
\end{figure}

\begin{figure}[htb]
\begin{minipage}{0.45\textwidth}
\begin{tikzpicture}[scale = 0.7]
\begin{axis}[dashed, xmin=-5,ymin=-5,xticklabels = {,,}, yticklabels = {},xmax=9,ymax=5,]
\addplot[mark = none] coordinates{(-5,0) (9,0)};
\addplot[mark = none] coordinates{(-5,1) (9,1)};
\addplot[mark = none] coordinates{(-5,2) (9,2)};
\addplot[mark = none] coordinates{(-5,3) (9,3)};
\addplot[mark = none] coordinates{(-5,4) (9,4)};
\addplot[mark = none] coordinates{(-5,-1) (9,-1)};
\addplot[mark = none] coordinates{(-5,-2) (9,-2)};
\addplot[mark = none] coordinates{(-5,-3) (9,-3)};
\addplot[mark = none] coordinates{(-5,-4) (9,-4)};
\addplot[mark = none] coordinates{(-5,-1) (5,-1)};
\addplot[mark = none] coordinates{(-5,-2) (5,-2)};
\addplot[mark = none] coordinates{(0,-5) (0,9)};
\addplot[mark = none] coordinates{(1,-5) (1,9)};
\addplot[mark = none] coordinates{(2,-5) (2,9)};
\addplot[mark = none] coordinates{(3,-5) (3,9)};
\addplot[mark = none] coordinates{(4,-5) (4,9)};
\addplot[mark = none] coordinates{(5,-5) (5,9)};
\addplot[mark = none] coordinates{(6,-5) (6,9)};
\addplot[mark = none] coordinates{(7,-5) (7,9)};
\addplot[mark = none] coordinates{(8,-5) (8,9)};
\addplot[mark = none] coordinates{(-1,-5) (-1,9)};
\addplot[mark = none] coordinates{(-2,-5) (-2,9)};
\addplot[mark = none] coordinates{(-3,-5) (-3,9)};
\addplot[mark = none] coordinates{(-4,-5) (-4,9)};
\addplot[mark = *, solid, thick] coordinates {(0, 0) (2,0) (4,0) (6,0) (8,0) (7,1) (6,2) (5,3) (4,4) (3,3) (2,2) (1,1) (0,0)};
\addplot[mark = *, solid ,thick] coordinates {(1,1) (3,1) (2,0) (1,1)};
\addplot[mark = *, solid, thick] coordinates {(2,2) (4,2) (6,2) (5,1) (4,0) (3,1) (2,2)};
\addplot[mark = *, solid ,thick] coordinates {(3,3) (5,3) (4,2) (3,3)};
\addplot[mark = *, solid, thick] coordinates {(5,1) (7,1) (6,0) (5,1)};

\addplot[mark=*] coordinates {(4,0)} node[label={$b_1$}]{} ;
\addplot[mark=*] coordinates {(2,2)} node[label={$a_1$}]{} ;

\addplot[mark=] coordinates {(-0.5,0)}[xshift=0pt,yshift=-10pt] node[label={$\mathbf{0}$}]{} ;
\addplot[mark=*] coordinates {(4,4)} node[xshift =0pt, yshift = -2pt][label={$\mathbf{a_2}$}]{} ;

\addplot[mark=] coordinates {(8.5,0)} node[xshift=2pt,yshift=-14pt][label={$\mathbf{b_2}$}]{} ;

\addplot[mark=] coordinates {(-4.5,-4)} node[xshift=-2pt,yshift=-14pt][label={$\mathbf{a'_2}$}]{} ;
\addplot[mark=] coordinates {(4.5,-4)} node[xshift=4pt,yshift=-14pt][label={$\mathbf{b'_2}$}]{} ;
\addplot[mark=*] coordinates {(0,0)} node(1)[]{} ;
\addplot[mark=*] coordinates {(4,4)} node(5)[]{} ;
\addplot[mark=*] coordinates {(8,0)} node(6)[]{} ;
\addplot[mark=*] coordinates {(-4,-4)} node(-5)[]{} ;
\addplot[mark=*] coordinates {(4,-4)} node(-6)[]{} ;
\draw [->,thick ,solid] (1) -- (5)node[xshift=-13pt,yshift=-5pt] {$u^{(2)}_4$};
\draw [->,thick ,solid] (6) -- (5)node[xshift=17pt,yshift=-5pt]{$u^{(2)}_3$};
\draw [->,thick ,solid] (1) -- (6)node[xshift=-12pt,yshift=-9pt] {-$u^{(2)}_4$};
\draw [->,thick ,solid] (5) -- (6)node[xshift=1pt,yshift=15pt]{-$u^{(2)}_3$};
\draw [->,thick ,solid] (5) -- (1)node[xshift=2pt,yshift=16pt] {-$u^{(2)}_5$};
\draw [->,thick ,solid] (6) -- (1)node[xshift=20pt,yshift=8pt]{$u^{(2)}_5$};
\draw [->,thick ,solid] (1) -- (-5)node[xshift=6pt,yshift=18pt] {$u^{(2)}_1$};
\draw [->,thick ,solid] (-6) -- (-5)node[xshift=15pt,yshift=-8pt]{$u^{(2)}_2$};
\draw [->,thick ,solid] (1) -- (-6)node[xshift=-2pt,yshift=18pt] {$u^{(2)}_1$};
\draw [->,thick ,solid] (-5) -- (-6)node[xshift=-15pt,yshift=-8pt]{$u^{(2)}_2$};
\draw [->,thick ,solid] (-5) -- (1)node[xshift=-2pt,yshift=-13pt] {$u^{(2)}_6$};
\draw [->,thick ,solid] (-6) -- (1)node[xshift=18pt,yshift=-8pt]{$u^{(2)}_6$};
\addplot[mark = *, solid, thick] coordinates {(0, 0) (-1,-1) (-2,-2) (-3,-3) (-4,-4) (-2,-4) (0,-4) (2,-4) (4,-4) (3,-3) (2,-2) (1,-1) (0,0)};
\addplot[mark = *, solid ,thick] coordinates {(-1,-1) (1,-1) (0,-2) (-1,-1)};
\addplot[mark = *, solid, thick] coordinates {(-2,-2) (0,-2) (2,-2) (1,-3) (0,-4) (-1,-3) (-2,-2)};
\addplot[mark = *, solid ,thick] coordinates {(-3,-3) (-1,-3) (-2,-4) (-3,-3)};
\addplot[mark = *, solid, thick] coordinates {(1,-3) (3,-3) (2,-4) (1,-3)};
\addplot[mark=*] coordinates {(2,-2)} node[xshift=0pt,yshift=-4pt][label={$b'_1$}]{} ;
\addplot[mark=*] coordinates {(-2,-2)} node[xshift=0pt,yshift=-4pt][label={$a'_1$}]{} ;
\end{axis}        
\end{tikzpicture}
\caption{First exit amplitude distribution  from  $F^{(2)} \cup F^{(2)'}$ starting $|\mathbf{0}^0\rangle $}\label{fig:exit amplitude distributions 2}
\end{minipage}
\hfill
\begin{minipage}{0.45\textwidth}
\begin{tikzpicture}[scale = 0.7]
\begin{axis}[dashed, xmin=-5,ymin=-5,xticklabels = {,,}, yticklabels = {},xmax=9,ymax=5,]
\addplot[mark = none] coordinates{(-5,0) (9,0)};
\addplot[mark = none] coordinates{(-5,1) (9,1)};
\addplot[mark = none] coordinates{(-5,2) (9,2)};
\addplot[mark = none] coordinates{(-5,3) (9,3)};
\addplot[mark = none] coordinates{(-5,4) (9,4)};
\addplot[mark = none] coordinates{(-5,-1) (9,-1)};
\addplot[mark = none] coordinates{(-5,-2) (9,-2)};
\addplot[mark = none] coordinates{(-5,-3) (9,-3)};
\addplot[mark = none] coordinates{(-5,-4) (9,-4)};
\addplot[mark = none] coordinates{(-5,-1) (5,-1)};
\addplot[mark = none] coordinates{(-5,-2) (5,-2)};
\addplot[mark = none] coordinates{(0,-5) (0,9)};
\addplot[mark = none] coordinates{(1,-5) (1,9)};
\addplot[mark = none] coordinates{(2,-5) (2,9)};
\addplot[mark = none] coordinates{(3,-5) (3,9)};
\addplot[mark = none] coordinates{(4,-5) (4,9)};
\addplot[mark = none] coordinates{(5,-5) (5,9)};
\addplot[mark = none] coordinates{(6,-5) (6,9)};
\addplot[mark = none] coordinates{(7,-5) (7,9)};
\addplot[mark = none] coordinates{(8,-5) (8,9)};
\addplot[mark = none] coordinates{(-1,-5) (-1,9)};
\addplot[mark = none] coordinates{(-2,-5) (-2,9)};
\addplot[mark = none] coordinates{(-3,-5) (-3,9)};
\addplot[mark = none] coordinates{(-4,-5) (-4,9)};
\addplot[mark = *, solid, thick] coordinates {(0, 0) (2,0) (4,0) (6,0) (8,0) (7,1) (6,2) (5,3) (4,4) (3,3) (2,2) (1,1) (0,0)};
\addplot[mark = *, solid ,thick] coordinates {(1,1) (3,1) (2,0) (1,1)};
\addplot[mark = *, solid, thick] coordinates {(2,2) (4,2) (6,2) (5,1) (4,0) (3,1) (2,2)};
\addplot[mark = *, solid ,thick] coordinates {(3,3) (5,3) (4,2) (3,3)};
\addplot[mark = *, solid, thick] coordinates {(5,1) (7,1) (6,0) (5,1)};

\addplot[mark=*] coordinates {(4,0)} node[label={$b_1$}]{} ;
\addplot[mark=*] coordinates {(2,2)} node[label={$a_1$}]{} ;

\addplot[mark=] coordinates {(-0.5,0)}[xshift=0pt,yshift=-10pt] node[label={$\mathbf{0}$}]{} ;
\addplot[mark=*] coordinates {(4,4)} node[xshift =0pt, yshift = -2pt][label={$\mathbf{a_2}$}]{} ;

\addplot[mark=] coordinates {(8.5,0)} node[xshift=2pt,yshift=-14pt][label={$\mathbf{b_2}$}]{} ;

\addplot[mark=] coordinates {(-4.5,-4)} node[xshift=-2pt,yshift=-14pt][label={$\mathbf{a'_2}$}]{} ;
\addplot[mark=] coordinates {(4.5,-4)} node[xshift=4pt,yshift=-14pt][label={$\mathbf{b'_2}$}]{} ;
\addplot[mark=*] coordinates {(0,0)} node(1)[]{} ;
\addplot[mark=*] coordinates {(4,4)} node(5)[]{} ;
\addplot[mark=*] coordinates {(8,0)} node(6)[]{} ;
\addplot[mark=*] coordinates {(-4,-4)} node(-5)[]{} ;
\addplot[mark=*] coordinates {(4,-4)} node(-6)[]{} ;
\draw [->,thick ,solid] (1) -- (5)node[xshift=-23pt,yshift=-5pt] {$0.0183$};
\draw [->,thick ,solid] (6) -- (5)node[xshift=23pt,yshift=-5pt]{$0.0486$};
\draw [->,thick ,solid] (1) -- (6)node[xshift=-12pt,yshift=-9pt] {$0.0183$};
\draw [->,thick ,solid] (5) -- (6)node[xshift=-2pt,yshift=20pt]{$0.0486$};
\draw [->,thick ,solid] (5) -- (1)node[xshift=-10pt,yshift=10pt] {$0.1831$};
\draw [->,thick ,solid] (6) -- (1)node[xshift=25pt,yshift=6pt]{$0.1831$};
\draw [->,thick ,solid] (1) -- (-5)node[xshift=3pt,yshift=28pt] {$0.0455$};
\draw [->,thick ,solid] (-6) -- (-5)node[xshift=15pt,yshift=-8pt]{$0.0503$};
\draw [->,thick ,solid] (1) -- (-6)node[xshift=10pt,yshift=15pt] {$0.0455$};
\draw [->,thick ,solid] (-5) -- (-6)node[xshift=-15pt,yshift=-8pt]{$0.0503$};
\draw [->,thick ,solid] (-5) -- (1)node[xshift=-25pt,yshift=-6pt] {$0.1542$};
\draw [->,thick ,solid] (-6) -- (1)node[xshift=25pt,yshift=-6pt]{$0.1542$};
\addplot[mark = *, solid, thick] coordinates {(0, 0) (-1,-1) (-2,-2) (-3,-3) (-4,-4) (-2,-4) (0,-4) (2,-4) (4,-4) (3,-3) (2,-2) (1,-1) (0,0)};
\addplot[mark = *, solid ,thick] coordinates {(-1,-1) (1,-1) (0,-2) (-1,-1)};
\addplot[mark = *, solid, thick] coordinates {(-2,-2) (0,-2) (2,-2) (1,-3) (0,-4) (-1,-3) (-2,-2)};
\addplot[mark = *, solid ,thick] coordinates {(-3,-3) (-1,-3) (-2,-4) (-3,-3)};
\addplot[mark = *, solid, thick] coordinates {(1,-3) (3,-3) (2,-4) (1,-3)};
\addplot[mark=*] coordinates {(2,-2)} node[xshift=0pt,yshift=-4pt][label={$b'_1$}]{} ;
\addplot[mark=*] coordinates {(-2,-2)} node[xshift=0pt,yshift=-4pt][label={$a'_1$}]{} ;
\end{axis}        
\end{tikzpicture}
\caption{First exit probability distribution  from  $F^{(2)} \cup F^{(2)'}$ starting $|\mathbf{0}^0\rangle $}\label{fig:exit distributions 2}
\end{minipage}
\end{figure}

Similar method applied to $u^{(2)}$, suppose $\psi_0 = |\mathbf{0}^0\rangle $, we obtain the exact first exit distributions of $(\mathbf{x}^i)_{\tau^{(2)}}$ for quantum random walks from $F^{(2)}\cup F^{(2)'}$. The results are summarized in Figure \ref{fig:exit amplitude distributions 2} and Figure \ref{fig:exit distributions 2}.

\subsection{The iteration map
} \label{subsec:iteration map}

The following theorem gives a formula for the iteration  map $T:\mathbb{C}^6 \rightarrow \mathbb{C}^6$ such that, 
$$(u^{(n+1)}_1,u^{(n+1)}_2,u^{(n+1)}_3,u^{(n+1)}_4,u^{(n+1)}_5,u^{(n+1)}_6) = T(u^{(n)}_1,u^{(n)}_2,u^{(n)}_3,u^{(n)}_4,u^{(n)}_5,u^{(n)}_6).$$

\begin{theorem} \label{th:T formula}
For all $n \ge 0$, we have 
\begin{eqnarray*}
u^{(n+1)}_1 &=& u^{(n)}_1(g_{\rho^{(n)}}^{(1)}(\mathbf{1'},\mathbf{5'})^0_1 + g_{\rho^{(n)}}^{(1)}(\mathbf{2'},\mathbf{5'})^2_1) + u^{(n)}_2 (g_{\rho^{(n)}}^{(1)}(\mathbf{1'},\mathbf{5'})^0_1 + g_{\rho^{(n)}}^{(1)}(\mathbf{2'},\mathbf{5'})^3_1)\\
u^{(n+1)}_2 &=& u^{(n)}_1(g_{\rho^{(n)}}^{(1)}(\mathbf{1'},\mathbf{5'})^0_0 + g_{\rho^{(n)}}^{(1)}(\mathbf{2'},\mathbf{5'})^2_0) + u^{(n)}_2 (g_{\rho^{(n)}}^{(1)}(\mathbf{1'},\mathbf{5'})^0_0 + g_{\rho^{(n)}}^{(1)}(\mathbf{2'},\mathbf{5'})^3_0)\\
u^{(n+1)}_3 &=& u^{(n)}_4(g_{\rho^{(n)}}^{(1)}(\mathbf{1},\mathbf{5})^4_5 - g_{\rho^{(n)}}^{(1)}(\mathbf{2},\mathbf{5})^3_5) + u^{(n)}_3 (g_{\rho^{(n)}}^{(1)}(\mathbf{1},\mathbf{5})^5_5 - g_{\rho^{(n)}}^{(1)}(\mathbf{2},\mathbf{5})^2_5)\\
u^{(n+1)}_4 & =& u^{(n)}_4(g_{\rho^{(n)}}^{(1)}(\mathbf{1},\mathbf{5})^4_4 - g_{\rho^{(n)}}^{(1)}(\mathbf{2},\mathbf{5})^3_4) + u^{(n)}_3 (g_{\rho^{(n)}}^{(1)}(\mathbf{1},\mathbf{5})^5_4 - g_{\rho^{(n)}}^{(1)}(\mathbf{2},\mathbf{5})^2_4)\\
u^{(n+1)}_5 &=& u^{(n)}_5 + u^{(n)}_4(g_{\rho^{(n)}}^{(1)}(\mathbf{1},\mathbf{0})^4_0 - g_{\rho^{(n)}}^{(1)}(\mathbf{2},\mathbf{0})^3_0) + u^{(n)}_3 (g_{\rho^{(n)}}^{(1)}(\mathbf{1},\mathbf{0})^5_0 - g_{\rho^{(n)}}^{(1)}(\mathbf{2},\mathbf{0})^2_0)\\
u^{(n+1)}_6 &=& u^{(n)}_6 + u^{(n)}_1(g_{\rho^{(n)}}^{(1)}(\mathbf{1'},\mathbf{0})^0_4 + g_{\rho^{(n)}}^{(1)}(\mathbf{2'},\mathbf{0})^2_4) + u^{(n)}_2 (g_{\rho^{(n)}}^{(1)}(\mathbf{1'},\mathbf{0})^0_4 + g_{\rho^{(n)}}^{(1)}(\mathbf{2'},\mathbf{0})^3_4),
\end{eqnarray*}
where $g_{\rho^{(n)}}^{(1)}(\mathbf{x},\mathbf{y})$ is given by Theorem \ref{th:Dirichlet g} for interior points $\mathbf{x}$ and (\ref{eq:g1 boundary 1})- (\ref{eq:g1 boundary 3}) for boundary points $\mathbf{x}$. 
\end{theorem}

\section{Proofs of Lemma \ref{le:g reduction sierpinski} - Theorem \ref{th:T formula}} \label{sec:proofs of tools}

In this section, we give the proofs of lemmas and theorems in Sections \ref{subsec:dimension reduction}-\ref{subsec:iteration map}. Proofs of the main results listed on the table in Section \ref{sec:main results} will be given in Section \ref{sec:proofs of main theorems}.

\begin{proof}[Proof of Lemma \ref{le:g reduction sierpinski}]
Let ${\Theta^{(n)}(z)(\mathbf{x},\mathbf{y})}_{j}^{i}$ be the same path amplitude Green function as ${g^{(n)}(z)(\mathbf{x},\mathbf{y})}_{j}^{i}$ without applying the first rotation $\tilde{G}$. 
\begin{figure}[htb]  
\begin{minipage}{0.25\textwidth}
\begin{tikzpicture}[scale = 0.7]
\begin{axis}[dashed, xmin=-1,ymin=-1,xticklabels = {,,}, yticklabels = {},xmax=5,ymax=3,]
\addplot[mark = none] coordinates{(-1,0) (5,0)};
\addplot[mark = none] coordinates{(-1,1) (5,1)};
\addplot[mark = none] coordinates{(-1,2) (5,2)};
\addplot[mark = none] coordinates{(0,-1) (0,3)};
\addplot[mark = none] coordinates{(1,-1) (1,3)};
\addplot[mark = none] coordinates{(2,-1) (2,3)};
\addplot[mark = none] coordinates{(3,-1) (3,3)};
\addplot[mark = none] coordinates{(4,-1) (4,3)};
\addplot[mark = *, solid, thick] coordinates {(0, 0) (2,0) (4,0) (3,1) (2,2) (1,1) (0,0)};
\addplot[mark = *, solid ,thick] coordinates {(1,1) (3,1) (2,0) (1,1)};
\addplot[mark=*] coordinates {(0,0)} node(1)[label={$\mathbf{0}$}]{} ;
\addplot[mark=*] coordinates {(1,1)} node(2)[]{} ;
\addplot[mark=*] coordinates {(2,0)} node(3)[]{} ;
\addplot[mark=*] coordinates {(3,1)} node(4)[]{} ;
\addplot[mark=*] coordinates {(2,2)} node(5)[label={$\mathbf{a_n}$}]{} ;
\addplot[mark=*] coordinates {(4,0)} node(6)[label={$\mathbf{b_n}$}]{} ;
\addplot[mark=none] coordinates {(-1,-1)} node(-1)[]{} ;
\addplot[mark=none] coordinates {(1,-1)} node(-2)[]{} ;
\draw [->,red,thick] (1) -- (2)node[xshift=0pt,yshift=-10pt] {$e_1$};
\draw [->,red,thick] (5) -- (2)node[xshift=12pt,yshift=5,red] {$e_4$};
\end{axis}        
\end{tikzpicture}
\caption{$\theta^{(n)}_1$}
\end{minipage}
\hfill
\begin{minipage}{0.25\textwidth}
\begin{tikzpicture}[scale = 0.7]
\begin{axis}[dashed, xmin=-1,ymin=-1,xticklabels = {,,}, yticklabels = {},xmax=5,ymax=3,]
\addplot[mark = none] coordinates{(-1,0) (5,0)};
\addplot[mark = none] coordinates{(-1,1) (5,1)};
\addplot[mark = none] coordinates{(-1,2) (5,2)};
\addplot[mark = none] coordinates{(0,-1) (0,3)};
\addplot[mark = none] coordinates{(1,-1) (1,3)};
\addplot[mark = none] coordinates{(2,-1) (2,3)};
\addplot[mark = none] coordinates{(3,-1) (3,3)};
\addplot[mark = none] coordinates{(4,-1) (4,3)};
\addplot[mark = *, solid, thick] coordinates {(0, 0) (2,0) (4,0) (3,1) (2,2) (1,1) (0,0)};
\addplot[mark = *, solid ,thick] coordinates {(1,1) (3,1) (2,0) (1,1)};
\addplot[mark=*] coordinates {(0,0)} node(1)[label={$\mathbf{0}$}]{} ;
\addplot[mark=*] coordinates {(1,1)} node(2)[]{} ;
\addplot[mark=*] coordinates {(2,0)} node(3)[]{} ;
\addplot[mark=*] coordinates {(3,1)} node(4)[]{} ;
\addplot[mark=*] coordinates {(2,2)} node(5)[label={$\mathbf{a_n}$}]{} ;
\addplot[mark=*] coordinates {(4,0)} node(6)[label={$\mathbf{b_n}$}]{} ;
\addplot[mark=none] coordinates {(-1,-1)} node(-1)[]{} ;
\addplot[mark=none] coordinates {(1,-1)} node(-2)[]{} ;
\draw [->,red,thick] (1) -- (3)node[xshift=-10pt,yshift=5pt] {$e_0$};
\addplot[mark=none] coordinates {(3,2)} node(32)[]{} ;
\addplot[mark=none] coordinates {(3,3)} node(33)[]{} ;
\draw [->,red,thick] (5) -- (4)node[xshift=-12pt,yshift=5pt] {$e_5$};
\end{axis}        
\end{tikzpicture}
\caption{$\theta^{(n)}_2$}
\end{minipage}
\hfill
\begin{minipage}{0.25\textwidth}
\begin{tikzpicture}[scale = 0.7]
\begin{axis}[dashed, xmin=-1,ymin=-1,xticklabels = {,,}, yticklabels = {},xmax=5,ymax=3,]
\addplot[mark = none] coordinates{(-1,0) (5,0)};
\addplot[mark = none] coordinates{(-1,1) (5,1)};
\addplot[mark = none] coordinates{(-1,2) (5,2)};
\addplot[mark = none] coordinates{(0,-1) (0,3)};
\addplot[mark = none] coordinates{(1,-1) (1,3)};
\addplot[mark = none] coordinates{(2,-1) (2,3)};
\addplot[mark = none] coordinates{(3,-1) (3,3)};
\addplot[mark = none] coordinates{(4,-1) (4,3)};
\addplot[mark = *, solid, thick] coordinates {(0, 0) (2,0) (4,0) (3,1) (2,2) (1,1) (0,0)};
\addplot[mark = *, solid ,thick] coordinates {(1,1) (3,1) (2,0) (1,1)};
\addplot[mark=*] coordinates {(0,0)} node(1)[label={$\mathbf{0}$}]{} ;
\addplot[mark=*] coordinates {(1,1)} node(2)[]{} ;
\addplot[mark=*] coordinates {(2,0)} node(3)[]{} ;
\addplot[mark=*] coordinates {(3,1)} node(4)[]{} ;
\addplot[mark=*] coordinates {(2,2)} node(5)[label={$\mathbf{a_n}$}]{} ;
\addplot[mark=*] coordinates {(4,0)} node(6)[label={$\mathbf{b_n}$}]{} ;
\addplot[mark=none] coordinates {(-1,-1)} node(-1)[]{} ;
\addplot[mark=none] coordinates {(1,-1)} node(-2)[]{} ;
\draw [->,red,thick] (1) -- (3)node[xshift=-10pt,yshift=5pt] {$e_0$};
\draw [->,red,thick] (5) -- (2)node[xshift=12pt,yshift=5] {$e_4$};
\end{axis}        
\end{tikzpicture}
\caption{$\theta^{(n)}_3$}
\end{minipage}

\begin{minipage}{0.25\textwidth}
\begin{tikzpicture}[scale = 0.7]
\begin{axis}[dashed, xmin=-1,ymin=-1,xticklabels = {,,}, yticklabels = {},xmax=5,ymax=3,]
\addplot[mark = none] coordinates{(-1,0) (5,0)};
\addplot[mark = none] coordinates{(-1,1) (5,1)};
\addplot[mark = none] coordinates{(-1,2) (5,2)};
\addplot[mark = none] coordinates{(0,-1) (0,3)};
\addplot[mark = none] coordinates{(1,-1) (1,3)};
\addplot[mark = none] coordinates{(2,-1) (2,3)};
\addplot[mark = none] coordinates{(3,-1) (3,3)};
\addplot[mark = none] coordinates{(4,-1) (4,3)};
\addplot[mark = *, solid, thick] coordinates {(0, 0) (2,0) (4,0) (3,1) (2,2) (1,1) (0,0)};
\addplot[mark = *, solid ,thick] coordinates {(1,1) (3,1) (2,0) (1,1)};
\addplot[mark=*] coordinates {(0,0)} node(1)[label={$\mathbf{0}$}]{} ;
\addplot[mark=*] coordinates {(1,1)} node(2)[]{} ;
\addplot[mark=*] coordinates {(2,0)} node(3)[]{} ;
\addplot[mark=*] coordinates {(3,1)} node(4)[]{} ;
\addplot[mark=*] coordinates {(2,2)} node(5)[label={$\mathbf{a_n}$}]{} ;
\addplot[mark=*] coordinates {(4,0)} node(6)[label={$\mathbf{b_n}$}]{} ;
\addplot[mark=none] coordinates {(-1,-1)} node(-1)[]{} ;
\addplot[mark=none] coordinates {(1,-1)} node(-2)[]{} ;
\draw [->,red,thick] (1) -- (2)node[xshift=0pt,yshift=-10pt] {$e_1$};
\draw [->,red,thick] (5) -- (4)node[xshift=-12pt,yshift=5pt] {$e_5$};
\end{axis}        
\end{tikzpicture}
\caption{$\theta^{(n)}_4$}
\end{minipage}
\hfill
\begin{minipage}{0.25\textwidth}
\begin{tikzpicture}[scale = 0.7]
\begin{axis}[dashed, xmin=-1,ymin=-1,xticklabels = {,,}, yticklabels = {},xmax=5,ymax=3,]
\addplot[mark = none] coordinates{(-1,0) (5,0)};
\addplot[mark = none] coordinates{(-1,1) (5,1)};
\addplot[mark = none] coordinates{(-1,2) (5,2)};
\addplot[mark = none] coordinates{(0,-1) (0,3)};
\addplot[mark = none] coordinates{(1,-1) (1,3)};
\addplot[mark = none] coordinates{(2,-1) (2,3)};
\addplot[mark = none] coordinates{(3,-1) (3,3)};
\addplot[mark = none] coordinates{(4,-1) (4,3)};
\addplot[mark = *, solid, thick] coordinates {(0, 0) (2,0) (4,0) (3,1) (2,2) (1,1) (0,0)};
\addplot[mark = *, solid ,thick] coordinates {(1,1) (3,1) (2,0) (1,1)};
\addplot[mark=*] coordinates {(0,0)} node(1)[label={$\mathbf{0}$}]{} ;
\addplot[mark=*] coordinates {(1,1)} node(2)[]{} ;
\addplot[mark=*] coordinates {(2,0)} node(3)[]{} ;
\addplot[mark=*] coordinates {(3,1)} node(4)[]{} ;
\addplot[mark=*] coordinates {(2,2)} node(5)[label={$\mathbf{a_n}$}]{} ;
\addplot[mark=*] coordinates {(4,0)} node(6)[label={$\mathbf{b_n}$}]{} ;
\addplot[mark=none] coordinates {(-1,-1)} node(-1)[]{} ;
\addplot[mark=none] coordinates {(1,-1)} node(-2)[]{} ;
\draw [->,red,thick] (1) -- (3)node[xshift=-10pt,yshift=5pt] {$e_0$};
\draw [->,red,thick] (1) -- (3)node[xshift=-10pt,yshift=5pt] {$e_0$};
\end{axis}        
\end{tikzpicture}
\caption{$\theta^{(n)}_5$}
\end{minipage}
\hfill
\begin{minipage}{0.25\textwidth}
\begin{tikzpicture}[scale = 0.7]
\begin{axis}[dashed, xmin=-1,ymin=-1,xticklabels = {,,}, yticklabels = {},xmax=5,ymax=3,]
\addplot[mark = none] coordinates{(-1,0) (5,0)};
\addplot[mark = none] coordinates{(-1,1) (5,1)};
\addplot[mark = none] coordinates{(-1,2) (5,2)};
\addplot[mark = none] coordinates{(0,-1) (0,3)};
\addplot[mark = none] coordinates{(1,-1) (1,3)};
\addplot[mark = none] coordinates{(2,-1) (2,3)};
\addplot[mark = none] coordinates{(3,-1) (3,3)};
\addplot[mark = none] coordinates{(4,-1) (4,3)};
\addplot[mark = *, solid, thick] coordinates {(0, 0) (2,0) (4,0) (3,1) (2,2) (1,1) (0,0)};
\addplot[mark = *, solid ,thick] coordinates {(1,1) (3,1) (2,0) (1,1)};
\addplot[mark=*] coordinates {(0,0)} node(1)[label={$\mathbf{0}$}]{} ;
\addplot[mark=*] coordinates {(1,1)} node(2)[]{} ;
\addplot[mark=*] coordinates {(2,0)} node(3)[]{} ;
\addplot[mark=*] coordinates {(3,1)} node(4)[]{} ;
\addplot[mark=*] coordinates {(2,2)} node(5)[label={$\mathbf{a_n}$}]{} ;
\addplot[mark=*] coordinates {(4,0)} node(6)[label={$\mathbf{b_n}$}]{} ;
\addplot[mark=none] coordinates {(-1,-1)} node(-1)[]{} ;
\addplot[mark=none] coordinates {(1,-1)} node(-2)[]{} ;
\draw [->,red,thick] (1) -- (3)node[xshift=-10pt,yshift=5pt] {$e_0$};
\draw [->,red,thick] (1) -- (2)node[xshift=0pt,yshift=-10pt] {$e_1$};
\end{axis}        
\end{tikzpicture}
\caption{$\theta^{(n)}_6$}
\end{minipage}
\end{figure}

Then $\Theta^{(n)}(z)$ has only six different variables which are corresponding to the graph from Figure 8 to Figure 13. The rest will be one of these six variables by the symmetry of $F^{(n)}$. For example, ${\Theta^{(n)}(z)(\mathbf{a_n},\mathbf{b_n})}_{2}^{5} = {\Theta^{(n)}(z)(\mathbf{0},\mathbf{a_n})}_{4}^{1}$ and 
${\Theta^{(n)}(z)(\mathbf{b_n},\mathbf{0})}_{0}^{2} = {\Theta^{(n)}(z)(\mathbf{0},\mathbf{a_n})}_{4}^{0}$.

Let
 $\theta^{(n)}_1(z) ={\Theta^{(n)}(z)(\mathbf{0},\mathbf{a_n})}_{4}^{1}$,
 $\theta^{(n)}_2(z) ={\Theta^{(n)}(z)(\mathbf{0},\mathbf{a_n})}_{5}^{0}$,
 $\theta^{(n)}_3 (z)={\Theta^{(n)}(z)(\mathbf{0},\mathbf{a_n})}_{4}^{0}$,
 $\theta^{(n)}_4(z) ={\Theta^{(n)}(z)(\mathbf{0},\mathbf{a_n})}_{5}^{1}$,
 $\theta^{(n)}_5(z) ={\Theta^{(n)}(z)(\mathbf{0},\mathbf{0})}_{4}^{4}$, and 
 $\theta^{(n)}_6(z) ={\Theta^{(n)}(z)(\mathbf{0},\mathbf{0})}_{4}^{5}$.

Then we have 
$${\Theta^{(n)}(z)(\mathbf{0},\mathbf{a_n})} = \bordermatrix{~& 0 & 1 & 4              & 5           \cr
                                                       0 & 0 & 0 & \theta^{(n)}_3 & \theta^{(n)}_2       \cr
                                                       1 & 0 & 0 & \theta^{(n)}_1 & \theta^{(n)}_4    \cr
                                                       4 & 0 & 0 & 0 & 0     \cr
                                                       5 & 0 & 0 & 0 & 0        \cr}, $$
and 
$${\Theta^{(n)}(z)(\mathbf{0},\mathbf{0})} = \bordermatrix{~& 0 & 1 & 4              & 5           \cr
                                                       0 & \theta^{(n)}_5 & \theta^{(n)}_6 & 0 & 0 \cr
                                                       1 & \theta^{(n)}_6 & \theta^{(n)}_5 & 0 & 0    \cr
                                                       4 & 0 & 0 & \theta^{(n)}_5 & \theta^{(n)}_6   \cr
                                                       5 & 0 & 0 & \theta^{(n)}_6 & \theta^{(n)}_5        \cr}.$$
Notice that ${g^{(n)}(z)(\mathbf{0},\mathbf{a_n})}_{j}^{i} = [G{\Theta^{(n)}(z)(\mathbf{0},\mathbf{a_n})}]_{ij}$, so we have
\begin{align*}
{g^{(n)}(\mathbf{0},\mathbf{a_n})} &= \tilde{G}{\Theta^{(n)}(z)(\mathbf{0},\mathbf{a_n})}\\
 &= \frac{1}{2}\bordermatrix{~& 0 & 1 & 4              & 5           \cr
                            0 & 0 & 0 &\theta^{(n)}_1-\theta^{(n)}_3 & \theta^{(n)}_4-\theta^{(n)}_2       \cr
                            1 & 0 & 0 & \theta^{(n)}_3-\theta^{(n)}_1 & \theta^{(n)}_2-\theta^{(n)}_4    \cr
                            4 & 0 & 0 & \theta^{(n)}_1+\theta^{(n)}_3 & \theta^{(n)}_2+\theta^{(n)}_4    \cr
                            5 & 0 & 0 & \theta^{(n)}_1+\theta^{(n)}_3 & \theta^{(n)}_2+\theta^{(n)}_4        \cr},                                                 
\end{align*}
and
\begin{align*}
{g^{(n)}(\mathbf{0},\mathbf{0})} &= \tilde{G}{\Theta^{(n)}(z)(\mathbf{0},\mathbf{0})}  \\
 &= \frac{1}{2}\bordermatrix{~& 0 & 1 & 4              & 5           \cr
                            0 & \theta^{(n)}_6-\theta^{(n)}_5 & \theta^{(n)}_5-\theta^{(n)}_6 &\theta^{(n)}_5+\theta^{(n)}_6 & \theta^{(n)}_5+\theta^{(n)}_6       \cr
                            1 & \theta^{(n)}_5-\theta^{(n)}_6 & \theta^{(n)}_6-\theta^{(n)}_5 & \theta^{(n)}_5+\theta^{(n)}_6 & \theta^{(n)}_5+\theta^{(n)}_6    \cr
                            4 & \theta^{(n)}_5+\theta^{(n)}_6 & \theta^{(n)}_5+\theta^{(n)}_6 & \theta^{(n)}_6-\theta^{(n)}_5 & \theta^{(n)}_5-\theta^{(n)}_6    \cr
                            5 & \theta^{(n)}_5+\theta^{(n)}_6 & \theta^{(n)}_5+\theta^{(n)}_6 & \theta^{(n)}_5-\theta^{(n)}_6 & \theta^{(n)}_6-\theta^{(n)}_5        \cr}.                                                 
\end{align*}
Let $u^{(n)}_1 = \frac{1}{2}(\theta^{(n)}_1+\theta^{(n)}_3)$, $u^{(n)}_2 = \frac{1}{2}(\theta^{(n)}_2+\theta^{(n)}_4)$, $u^{(n)}_3 = \frac{1}{2}(\theta^{(n)}_4-\theta^{(n)}_2)$, $u^{(n)}_4 = \frac{1}{2}(\theta^{(n)}_1-\theta^{(n)}_3)$, $u^{(n)}_5 = \frac{1}{2}(\theta^{(n)}_6-\theta^{(n)}_5)$, and $u^{(n)}_6 = \frac{1}{2}(\theta^{(n)}_6+\theta^{(n)}_5)$. We then  get all of the blocks of $g^{(n)}(z)$ with variables $u^{(n)}_i(z)$ for $i =1,...,6$ as stated in the lemma. 

As for $F^{(n)'}$, we obtain $g^{(n)'}$ by using symmetries of $g^{(n)}$.
\end{proof}

\begin{proof} [Proof of Theorem \ref{th: g recursive}] (a) Let $w$ be a path in $ F^{(n+1)} \cup  F^{(n+1)'}$ with starting point at one of the boundary points in 
$\partial F^{(n+1)} \cup \partial F^{(n+1)'} = \{ \mathbf{0}, \mathbf{a_{n+1}}, \mathbf{b_{n+1}},\mathbf{a'_{n+1}}, \mathbf{b'_{n+1}}\}$. 
Let $\zeta ^{(n)}_l$, $l=0, 1, 2, ...$ be the successive hitting times by  $w$ at  $2^{n} (F^{(1)} \cup F^{(1)'})$, i.e., $\zeta _0^{(n)}=0$
and 
$${\zeta_l^{(n)}} = \inf \{ t\ge \zeta_{l-1}^{(n)}+1; w_t=x_t^{k_t}, x_t  \in 2^n(F^{(1)} \cup F^{(1)'})\},$$
 $l \ge 1$. We put ${\zeta_l^{(n)}} =\infty$ if the above set is empty. 
 
 Let  $\tilde{w}_l=\frac{1}{2^n}w_{\zeta _l}$. Here $\zeta_l=\zeta_l ^{(n)}$ and if $w=x^i$, then we define $cw=(cx)^i$, for any constant $c$.
 Then $\tilde{w}$ is a path in $F^{(1)} \cup F^{(1)'}$. Let $\Gamma ^{(n+1)}  (x^i, y^j)$ be the set of all path $w$ in $ F^{(n+1)} \cup  F^{(n+1)'}$ such that $w_0=x^i $ and $ w_{\tau ^{(n+1)}}=y^j$.
 
 By definition, 
 \begin{eqnarray}
 g^{(n+1)}(x, y)^i_j\\
 &=& \sum_{w \in \Gamma ^{(n+1)}(x^i, y^j)}\prod _{t=0}^{\tau ^{(n+1)} -1}\rho ^{(0)}(w_t, w_{t+1})\\
 &=& \sum _{\tilde {w} \in \Gamma ^{(1)} (\frac{1}{2^n}x^i, \frac{1}{2^n}y^j)} \sum_{w} \prod _{t=0}^{\tau^{(n+1)} -1}\rho^{(0)}(w_t, w_{t+1}),
 \end{eqnarray}
 here the second sum of the above equation is summing over all path $w \in \Gamma ^{(n+1)} (x^i, y^j)$ where  $w$ is compatible with $2^n \tilde {w} $. The above equals to 
 \begin{eqnarray}
& = &\sum _{\tilde {w} \in \Gamma ^{(1)} (\frac{1}{2^n}x^i, \frac{1}{2^n}y^j)} \prod_{l=0}^{\tau^{(1)}(\tilde w)-1}\sum_{w \in \Gamma ^{(n)} (2^n \tilde {w_l}, 2^n \tilde {w}_{l+1})} \prod _{t=0}^{\tau^{(n)} -1}\rho^{(0)}(w_t, w_{t+1})\\
 &=& \sum _{\tilde {w} \in \Gamma ^{(1)} (\frac{1}{2^n}x^i, \frac{1}{2^n}y^j)} \prod_{l=0}^{\tau^{(1)}(\tilde w)-1} g^{(n)}(2^n \tilde {w_l}, 2^n \tilde {w}_{l+1})\\
 &=&\sum _{\tilde {w} \in \Gamma ^{(1)} (\frac{1}{2^n}x^i, \frac{1}{2^n}y^j)} \prod_{l=0}^{\tau^{(1)}(\tilde w)-1}
 \rho ^{(n)}( \tilde {w_l}, \tilde {w}_{l+1})\\
 &=& g^{(1)}_{\rho^{(n)}}(\frac{1}{2^n} x, \frac{1}{2^n} y)^i_j.
 \end{eqnarray}
 
 (b) follows from Part (a) and (\ref{eq:rho n}).
 \end{proof}
 
 \begin{proof}[Proof of Theorem \ref{th:Dirichlet g}]

For Part (a), 
by (\ref{eq:Green rho}) with $n=1$,
\begin{equation}
g_{\rho}^{(1)}(\mathbf{x},\mathbf{y})_{j}^{i} =\sum_{t=1}^{\infty} \Psi_\rho (w_0 = \mathbf{x}^i , w_{t} = \mathbf{y}^j ,\tau^{(1)} = t) .
\end{equation}

Let 
\begin{equation}
\bar{g}_{\rho}^{(1)}(\mathbf{x},\mathbf{y})_{j}^{i} =\sum_{t=0}^{\infty} \Psi_\rho (w_0 = \mathbf{x}^i , w_{t} = \mathbf{y}^j ,\bar{\tau}^{(1)} = t) ,
\end{equation}
where 
$\bar{\tau}^{(1)}  = \inf \{ t\geq 0; w_t=x_t^{i_t}, x_t  \in \{\mathbf{0},\mathbf{4},\mathbf{5}\}\}$.

By definitions, we have 
$\bar{g}_{\rho}^{(1)}(\mathbf{x},\mathbf{y})_{j}^{i}=g_{\rho}^{(1)}(\mathbf{x},\mathbf{y}))_{j}^{i}$, if $\mathbf{x}=\mathbf{1}, \mathbf{2}, \mathbf{3}$ and $\bar{g}_{\rho}^{(1)}(\mathbf{x},\mathbf{y})_{j}^{i}=\delta(\mathbf{x}, \mathbf{y})\delta(i, j)$, if $\mathbf{x}=\mathbf{0}, \mathbf{4}, \mathbf{5}$.

Let $\mathbf{x}=\mathbf{1}, \mathbf{2}, \mathbf{3}$.  Let $\Gamma$ be the set of all paths $w$ such that  $w_0 = \mathbf{x}^i , w_{t} = \mathbf{y}^j ,\tau^{(1)} = t$. Let $\Gamma_{\mathbf{u}^k}=\{ w \in \Gamma ; w_1=\mathbf{u}^k\}.$ Then
\begin{eqnarray}
g_{\rho}^{(1)}(\mathbf{x},\mathbf{y})_{j}^{i} =\sum_{t=1}^{\infty} \Psi_\rho (\Gamma) \\
=\sum_{\mathbf{u}^k}\sum_{t=1}^{\infty} \Psi_\rho (\Gamma_{\mathbf{u}^k}) \\
=\sum_{\mathbf{u} } \sum _{k} \rho (\mathbf{x}, \mathbf{u})^i_k \bar{g}^{(1)}_{\rho} (\mathbf{u}, \mathbf{y})^k_j\\
=\sum_{\mathbf{u}=\mathbf{1}, \mathbf{2}, \mathbf{3}}\sum _k \rho(\mathbf{x}, \mathbf{u})^i_k g^{(1)}_\rho(\mathbf{u}, \mathbf{y})^k_j+\rho(\mathbf{x}, \mathbf{y})^i_j.
\end{eqnarray}
This implies that 
\begin{eqnarray}
g^{(1)}_{\rho}(\mathbf{x}, \mathbf{y})=(\tilde {\rho} {g}^{(1)}_{\rho})(\mathbf{x}, \mathbf{y})+\rho(\mathbf{x}, \mathbf{y}),
\end{eqnarray}
for all $\mathbf{x}=\mathbf{1}, \mathbf{2}, \mathbf{3}$, $\mathbf{y}=\mathbf{0}, \mathbf{4}, \mathbf{5}$, here $\tilde {\rho}$ is the $12 \times 12$ matrix obtained from $\rho $  with restriction  to components involving only $\mathbf{1}, \mathbf{2}, \mathbf{3}$ and all $k$.
Therefore, we have 
\begin{eqnarray}
[I_{12}-\tilde {\rho}]g^{(1)}_{\rho}(\mathbf{x}, \mathbf{y})=\rho(\mathbf{x}, \mathbf{y}),
\end{eqnarray}
for all $\mathbf{x}=\mathbf{1}, \mathbf{2}, \mathbf{3}$, $\mathbf{y}=\mathbf{0}, \mathbf{4}, \mathbf{5}$.

For Part (b),  if $I_{12}-\tilde {\rho}$ is invertible, then by Part (a), we have 
for all $\mathbf{y}=\mathbf{0}, \mathbf{4}, \mathbf{5} $, 
\begin{equation}
 \begin{bmatrix}
g_{\rho}^{(1)}(\mathbf{1},\mathbf{y})  \\ g_{\rho}^{(1)}(\mathbf{2},\mathbf{y}) \\ g_{\rho}^{(1)}(\mathbf{3},\mathbf{y})
\end{bmatrix} = \frac{1}{I_{12}-\tilde {\rho}} \begin{bmatrix}
\rho(\mathbf{1},\mathbf{y})  \\ \rho(\mathbf{2},\mathbf{y})  \\ \rho(\mathbf{3},\mathbf{y})
\end{bmatrix},
\end{equation}
where $I_{12}$ is the $12 \times 12$ identity matrix.
\end{proof}

\begin{proof} [Proof of Lemma \ref{le:fast inversion}] (a) We extend $\rho^{(n)}$ by parallel translation  to become a transition function on $F^{(1)}$. We also reorder the matrix by switching $\mathbf{2}^0, \mathbf{2}^1$ with $\mathbf{2}^2, \mathbf{2}^3$ and $\mathbf{3}^2, \mathbf{3}^3$ with $\mathbf{3}^4,\mathbf{3}^5$ for rows and columns. After switching,  by Lemma \ref{le:g reduction sierpinski} and  (\ref{eq:rho n}), we have $[\tilde {\rho}^{(n)}] =$
\begin{scriptsize}
\begin{eqnarray}\label{eq:rho tilde n formula}
\bordermatrix{~& \mathbf{1}^0& \mathbf{1}^1 & \mathbf{1}^4 & \mathbf{1}^5& \mathbf{2}^2&\mathbf{2}^3&\mathbf{2}^0&\mathbf{2}^1 & \mathbf{3}^4&\mathbf{3}^5&\mathbf{3}^2&\mathbf{3}^3       \cr
   \mathbf{1}^0 &u^{(n)}_5&-u^{(n)}_5&u^{(n)}_6&u^{(n)}_6  &u^{(n)}_1 &u^{(n)}_2&0&0 &  0&0&-u^{(n)}_3&-u^{(n)}_4\cr 
   \mathbf{1}^1 &-u^{(n)}_5&u^{(n)}_5&u^{(n)}_6&u^{(n)}_6  &u^{(n)}_1 &u^{(n)}_2&0&0 &  0&0&u^{(n)}_3&u^{(n)}_4 \cr
   \mathbf{1}^4 &u^{(n)}_6&u^{(n)}_6&u^{(n)}_5&-u^{(n)}_5  &u^{(n)}_4 &u^{(n)}_3&0&0 &  0&0&u^{(n)}_2&u^{(n)}_1 \cr
   \mathbf{1}^5 &u^{(n)}_6&u^{(n)}_6&-u^{(n)}_5&u^{(n)}_5  &-u^{(n)}_4&-u^{(n)}_3&0&0&  0&0& u^{(n)}_2&u^{(n)}_1 \cr
   \mathbf{2}^2 &0&0&-u^{(n)}_3&-u^{(n)}_4   &u^{(n)}_5&-u^{(n)}_5&u^{(n)}_6&u^{(n)}_6& u^{(n)}_1&u^{(n)}_2&0&0\cr   
   \mathbf{2}^3 &0&0&u^{(n)}_3&u^{(n)}_4   &-u^{(n)}_5&u^{(n)}_5&u^{(n)}_6&u^{(n)}_6&   u^{(n)}_1&u^{(n)}_2&0&0\cr   
   \mathbf{2}^0 &0&0&u^{(n)}_2 &u^{(n)}_1  &u^{(n)}_6&u^{(n)}_6&u^{(n)}_5&-u^{(n)}_5&   u^{(n)}_4&u^{(n)}_3&0&0\cr
   \mathbf{2}^1 &0&0&u^{(n)}_2 &u^{(n)}_1  &u^{(n)}_6&u^{(n)}_6&-u^{(n)}_5&u^{(n)}_5&  -u^{(n)}_4&-u^{(n)}_3&0&0\cr
   \mathbf{3}^4 &u^{(n)}_1&u^{(n)}_2&0&0   &0&0&-u^{(n)}_3&-u^{(n)}_4  &u^{(n)}_5&-u^{(n)}_5&u^{(n)}_6&u^{(n)}_6 \cr
   \mathbf{3}^5 &u^{(n)}_1&u^{(n)}_2&0&0   &0&0&u^{(n)}_3&u^{(n)}_4   &-u^{(n)}_5&u^{(n)}_5&u^{(n)}_6&u^{(n)}_6 \cr         
   \mathbf{3}^2 &u^{(n)}_4&u^{(n)}_3&0&0   &0&0&u^{(n)}_2&u^{(n)}_1 &u^{(n)}_6&u^{(n)}_6&u^{(n)}_5&-u^{(n)}_5 \cr          
   \mathbf{3}^3 &-u^{(n)}_4&-u^{(n)}_3&0&0 &0&0&u^{(n)}_2&u^{(n)}_1 &u^{(n)}_6&u^{(n)}_6&-u^{(n)}_5&u^{(n)}_5\cr}.
\end{eqnarray}
\end{scriptsize}
Notice we have  switched $\mathbf{2}^0, \mathbf{2}^1$ with $\mathbf{2}^2, \mathbf{2}^3$ and $\mathbf{3}^2, \mathbf{3}^3$ with $\mathbf{3}^4,\mathbf{3}^5$ for rows and columns. This is an important step since  after switching, this matrix has the following block structure, which  is  easier to find its inverse than for a general $12 \times 12 $ matrix,
$$
\left[
\begin{array}{c|c|c}
 A& B &C\\
\hline
C & A& B\\
\hline
B & C& A\\
\end{array}
\right],$$
where 
\begin{eqnarray}\label{eq: An}
A=A^{(n)}=\left[
\begin{array}{cccc}
u^{(n)}_5&-u^{(n)}_5&u^{(n)}_6&u^{(n)}_6\\
-u^{(n)}_5&u^{(n)}_5&u^{(n)}_6&u^{(n)}_6 \\
u^{(n)}_6&u^{(n)}_6&u^{(n)}_5&-u^{(n)}_5 \\
u^{(n)}_6&u^{(n)}_6&-u^{(n)}_5&u^{(n)}_5\\
\end{array}
\right],
\end{eqnarray}
\begin{eqnarray}\label{eq: Bn}
B=B^{(n)}=\left[
\begin{array}{cccc}
u^{(n)}_1 &u^{(n)}_2&0&0 \\
u^{(n)}_1 &u^{(n)}_2&0&0 \\
u^{(n)}_4 &u^{(n)}_3&0&0 \\
-u^{(n)}_4&-u^{(n)}_3&0&0\\
\end{array}
\right],
\end{eqnarray}
\begin{eqnarray}\label{eq: Cn}
C=C^{(n)}=\left[
\begin{array}{cccc}
 0&0&-u^{(n)}_3&-u^{(n)}_4 \\
0&0&u^{(n)}_3&u^{(n)}_4 \\
0&0&u^{(n)}_2&u^{(n)}_1  \\
0&0& u^{(n)}_2&u^{(n)}_1\\
\end{array}
\right].
\end{eqnarray}
Proof (b). By (a),
$$-[I_{12} - \tilde {\rho}^{(n)}]=\left[
\begin{array}{c|c|c}
\bar{A} & B& C \\
\hline
C & \bar{A} & B\\
\hline
B & C& \bar{A}\\
\end{array}
\right].$$
By direct computation, we verify that 
$$
\left[
\begin{array}{c|c|c}
\bar{A} & B & C \\
\hline
C & \bar{A} & B\\
\hline
B & C& \bar{A}\\
\end{array}
\right]
\left[
\begin{array}{c|c|c}
X & Y & Z\\
\hline
Z & X & Y\\
\hline
Y & Z & X\\
\end{array}
\right]=\left[
\begin{array}{c|c|c}
I_4 & 0 & 0\\
\hline
0 & I_4 & 0\\
\hline
0 & 0 & I_4\\
\end{array}
\right].
$$
\end{proof}

\begin{proof}[Proof of Theorem \ref{th:Sierpinski step 1}]

To apply  Lemma \ref{le:fast inversion}, we reorder the matrix $[\tilde {\rho}^{(0)}] $ by switching $\mathbf{2}^0, \mathbf{2}^1$ with $\mathbf{2}^2, \mathbf{2}^3$ and $\mathbf{3}^2, \mathbf{3}^3$ with $\mathbf{3}^4,\mathbf{3}^5$ for rows and columns. After switching, 
by (\ref{eq: An}) - (\ref{eq: Cn}),
we obtain $A^{(0)}, B^{(0)}$, and $C^{(0)}$. 

By Lemma \ref{le:fast inversion}(b), Theorem \ref{th:Dirichlet g}(b), and substitute  the values of $u^{(0)}_1,..., u^{(0)}_6$, we have
$$
{g_{\rho^{(0)}}^{(1)}(\mathbf{1},\mathbf{5})} = \bordermatrix{~& 0 & 1 & 4              & 5           \cr
                                                       0  &0&0& rz&-\frac{r^2z^2}{1+rz}     \cr
                                                       1  &0&0&-\frac{2r^3z^3}{1+rz}-rz&r^2z^2+\frac{r^3z^3}{1+rz} \cr
                                                       4  &0&0&-\frac{2r^3z^3}{1+rz}+rz&r^2z^2+\frac{r^3z^3}{1+rz}    \cr
                                                       5  &0&0& rz&\frac{r^2z^2}{1+rz}       \cr} 
,$$
$$
{g_{\rho^{(0)}}^{(1)}(\mathbf{2},\mathbf{5})} = \bordermatrix{~& 0 & 1 & 4              & 5           \cr
                                                       2  &0&0& -\frac{r^2z^2}{1+rz} &\frac{r^2z^2}{1+rz}    \cr
                                                       3  &0&0&r^2z^2+\frac{r^3z^3}{1+rz}&r^2z^2+\frac{r^3z^3}{1+rz} \cr
                                                       0  &0&0&r^2z^2+\frac{r^3z^3}{1+rz}&r^2z^2+\frac{r^3z^3}{1+rz}    \cr
                                                       1  &0&0& \frac{r^2z^2}{1+rz}&-\frac{r^2z^2}{1+rz}       \cr} 
,$$
$$
{g_{\rho^{(0)}}^{(1)}(\mathbf{3},\mathbf{5})} = \bordermatrix{~& 0 & 1 & 4              & 5           \cr
                                                       4  &0&0&\frac{r^2z^2}{1+rz}& rz     \cr
                                                       5  &0&0&r^2z^2+\frac{r^3z^3}{1+rz}&-\frac{2r^3z^3}{1+rz}+rz \cr
                                                       2  &0&0&r^2z^2+\frac{r^3z^3}{1+rz}&-\frac{2r^3z^3}{1+rz}-rz    \cr
                                                       3  &0&0&-\frac{r^2z^2}{1+rz} &rz      \cr} 
.$$
According to this reordering, 
$$ 
{\rho^{(0)}(\mathbf{1},\mathbf{5})} = \bordermatrix{~& 0 & 1 & 4              & 5           \cr
                                                       0  &0&0& u^{(0)}_4&u^{(0)}_3      \cr
                                                       1  &0&0&- u^{(0)}_4&-u^{(0)}_3   \cr
                                                       4  &0&0& u^{(0)}_1&u^{(0)}_2    \cr
                                                       5  &0&0& u^{(0)}_1&u^{(0)}_2       \cr} ,
$$ 
$$
{\rho^{(0)}(\mathbf{2},\mathbf{5})} = O_{4 \times 4},
$$
$$
{\rho^{(0)}(\mathbf{3},\mathbf{5})} = \bordermatrix{~& 0 & 1 & 4              & 5           \cr
                                                       4  &0&0& u^{(0)}_2&u^{(0)}_1      \cr
                                                       5  &0&0& u^{(0)}_2&u^{(0)}_1   \cr
                                                       2  &0&0& -u^{(0)}_3&-u^{(0)}_4   \cr
                                                       3  &0&0& u^{(0)}_3&u^{(0)}_4       \cr}, 
$$
$$
{\rho^{(0)}(\mathbf{0},\mathbf{2})} = \bordermatrix{~& 2 & 3 & 0              & 1           \cr
                                                       0 & -u^{(0)}_3 & -u^{(0)}_4&0&0       \cr
                                                       1 & u^{(0)}_3 & u^{(0)}_4&0&0   \cr
                                                       4 & u^{(0)}_2 & u^{(0)}_1&0&0    \cr
                                                       5 & u^{(0)}_2 & u^{(0)}_1&0&0       \cr},
$$
$$
{\rho^{(0)}(\mathbf{0},\mathbf{1})} = \bordermatrix{~& 0 & 1 & 4              & 5           \cr
                                                       0  &0&0& u^{(0)}_4&u^{(0)}_3      \cr
                                                       1  &0&0&-u^{(0)}_4&-u^{(0)}_3   \cr
                                                       4  &0&0& u^{(0)}_1&u^{(0)}_2    \cr
                                                       5  &0&0& u^{(0)}_1&u^{(0)}_2       \cr}.
$$
Then
\begin{align}
g_{\rho^{(0)}}^{(1)}(\mathbf{0},\mathbf{a_1}) = g_{\rho^{(0)}}^{(1)}(\mathbf{0},\mathbf{5})&= \rho^{(0)}(\mathbf{0},\mathbf{1})g_{\rho^{(0)}}^{(1)}(\mathbf{1},\mathbf{5})+\rho^{(0)}(\mathbf{0},\mathbf{2})g_{\rho^{(0)}}^{(1)}(\mathbf{2},\mathbf{5})\\
&= \bordermatrix{~& 0 & 1 & 4              & 5           \cr
                                                       0  &0&0& -r^3z^3+r^2z^2-\frac{3r^4z^4}{1+rz} &0    \cr
                                                       1  &0&0& r^3z^3-r^2z^2+\frac{3r^4z^4}{1+rz} &0  \cr
                                                       4  &0&0& r^3z^3+r^2z^2-\frac{r^4z^4}{1+rz} & 2r^3z^3+\frac{2r^4z^4}{1+rz}\cr
                                                       5  &0&0& r^3z^3+r^2z^2-\frac{r^4z^4}{1+rz} & 2r^3z^3+\frac{2r^4z^4}{1+rz}\cr}.
\end{align}
By comparing  the expressions of $g_{\rho^{(0)}}^{(1)}(\mathbf{0},\mathbf{a_1})$  with the one in Lemma \ref{le:g reduction sierpinski},  for $n=1$, we have obtained the formulas for $u^{(1)}_1$ - $u^{(1)}_4$.
Next, for $u^{(1)}_5$ and  $u^{(1)}_6$, we have
$$
{\rho^{(0)}(\mathbf{0},\mathbf{0})} = \bordermatrix{~& 0 & 1 & 4              & 5           \cr
                                                       0  &u^{(0)}_5&-u^{(0)}_5& u^{(0)}_6&u^{(0)}_6      \cr
                                                       1  &-u^{(0)}_5&u^{(0)}_5&u^{(0)}_6&u^{(0)}_6  \cr
                                                       4  &u^{(0)}_6&u^{(0)}_6& u^{(0)}_5&-u^{(0)}_5    \cr
                                                       5  &u^{(0)}_6&u^{(0)}_6& -u^{(0)}_5&u^{(0)}_5       \cr},
$$
$$
{\rho^{(0)}(\mathbf{5},\mathbf{5})} = \bordermatrix{~& 0 & 1 & 4              & 5           \cr
                                                       0  &u^{(0)}_5&-u^{(0)}_5& u^{(0)}_6&u^{(0)}_6      \cr
                                                       1  &-u^{(0)}_5&u^{(0)}_5&u^{(0)}_6&u^{(0)}_6  \cr
                                                       4  &u^{(0)}_6&u^{(0)}_6& u^{(0)}_5&-u^{(0)}_5    \cr
                                                       5  &u^{(0)}_6&u^{(0)}_6& -u^{(0)}_5&u^{(0)}_5       \cr},
$$
$$
{\rho^{(0)}(\mathbf{1},\mathbf{0})} = \bordermatrix{~& 0 & 1 & 4              & 5           \cr
                                                       0  &u^{(0)}_2&u^{(0)}_1&0&0      \cr
                                                       1  &u^{(0)}_2&u^{(0)}_1&0&0   \cr
                                                       4  &-u^{(0)}_3&-u^{(0)}_4&0&0    \cr
                                                       5  &u^{(0)}_3&u^{(0)}_4&0&0       \cr},
$$
$$
{\rho^{(0)}(\mathbf{2},\mathbf{0})} = \bordermatrix{~& 0 & 1 & 4              & 5           \cr
                                                       2  &u^{(0)}_4&u^{(0)}_3&0&0      \cr
                                                       3  &-u^{(0)}_4&-u^{(0)}_3&0&0   \cr
                                                       0  &u^{(0)}_1&u^{(0)}_2&0&0  \cr
                                                       1  &u^{(0)}_1&u^{(0)}_2&0&0       \cr},
$$
$$
{\rho^{(0)}(\mathbf{3},\mathbf{0})} = O_{4 \times 4},
$$
$${\rho^{(0)}(\mathbf{5},\mathbf{1})} = \bordermatrix{~& 0 & 1 & 4              & 5           \cr
                                                       0  &u^{(0)}_2&u^{(0)}_1&0&0      \cr
                                                       1  &u^{(0)}_2&u^{(0)}_1&0&0   \cr
                                                       4  &-u^{(0)}_3&-u^{(0)}_4&0&0    \cr
                                                       5  &u^{(0)}_3&u^{(0)}_4&0&0       \cr},
$$
$${\rho^{(0)}(\mathbf{5},\mathbf{3})} = \bordermatrix{~& 4 & 5 & 2             & 3           \cr
                                                       0  &0&0&u^{(0)}_1&u^{(0)}_2      \cr
                                                       1  &0&0&u^{(0)}_1&u^{(0)}_2   \cr
                                                       4  &0&0&u^{(0)}_4&u^{(0)}_3    \cr
                                                       5  &0&0&-u^{(0)}_4&-u^{(0)}_3       \cr}.
$$                                                 
Plugging in the values of $u^{(0)}_1,...,u^{(0)}_6$ and,  by Theorem \ref{th:Dirichlet g}(b), we have 
$$
{g_{\rho^{(0)}}^{(1)}(\mathbf{1},\mathbf{0})} = \bordermatrix{~& 0 & 1 & 4              & 5           \cr
                                                       0  &\frac{r^2z^2}{1+rz}&rz& 0&0     \cr
                                                       1  &r^2z^2+\frac{r^3z^3}{1+rz}&-\frac{2r^3z^3}{1+rz}+rz&0&0 \cr
                                                       4  &r^2z^2+\frac{r^3z^3}{1+rz}&-\frac{2r^3z^3}{1+rz}-rz&0&0 \cr
                                                       5  &-\frac{r^2z^2}{1+rz}&rz& 0&0     \cr},
$$
$$
{g_{\rho^{(0)}}^{(1)}(\mathbf{2},\mathbf{0})} = \bordermatrix{~& 0 & 1 & 4              & 5           \cr
                                                       2  &rz&-\frac{r^2z^2}{1+rz}&0 &0    \cr
                                                       3  &-\frac{2r^3z^3}{1+rz}-rz&r^2z^2+\frac{r^3z^3}{1+rz}&0&0 \cr
                                                       0  &-\frac{2r^3z^3}{1+rz}+rz&r^2z^2+\frac{r^3z^3}{1+rz}&0&0    \cr
                                                       1  &rz&\frac{r^2z^2}{1+rz}&0 &0       \cr},
$$
$$
{g_{\rho^{(0)}}^{(1)}(\mathbf{3},\mathbf{0})} = \bordermatrix{~& 0 & 1 & 4              & 5           \cr
                                                       4  &-\frac{r^2z^2}{1+rz}&\frac{r^2z^2}{1+rz}&0& 0     \cr
                                                       5  &r^2z^2+\frac{r^3z^3}{1+rz}&r^2z^2+\frac{r^3z^3}{1+rz}&0&0 \cr
                                                       2  &r^2z^2+\frac{r^3z^3}{1+rz}&r^2z^2+\frac{r^3z^3}{1+rz}&0&0   \cr
                                                       3  &\frac{r^2z^2}{1+rz}&-\frac{r^2z^2}{1+rz}&0& 0      \cr}.
$$
Therefore, by (\ref{eq:g1 boundary 1}),
\begin{scriptsize}
\begin{eqnarray}
&&g_{\rho^{(0)}}^{(1)}(\mathbf{0},\mathbf{0}) =  \\
&&\bordermatrix{~& 0 & 1 & 4              & 5           \cr
                                                       0  &r^3z^3+r^2z^2+\frac{3r^4z^4}{1+rz} &-r^3z^3-r^2z^2-\frac{3r^4z^4}{1+rz}& r^3z^3-r^2z^2-\frac{r^4z^4}{1+rz}& r^3z^3-r^2z^2-\frac{r^4z^4}{1+rz}   \cr
                                                       1  &-r^3z^3-r^2z^2-\frac{3r^4z^4}{1+rz}&r^3z^3+r^2z^2+\frac{3r^4z^4}{1+rz}& r^3z^3-r^2z^2-\frac{r^4z^4}{1+rz} &r^3z^3-r^2z^2-\frac{r^4z^4}{1+rz}  \cr
                                                       4  &r^3z^3-r^2z^2-\frac{r^4z^4}{1+rz}&r^3z^3-r^2z^2-\frac{r^4z^4}{1+rz}& r^3z^3+r^2z^2+\frac{3r^4z^4}{1+rz} & -r^3z^3-r^2z^2-\frac{3r^4z^4}{1+rz}\cr
                                                       5  &r^3z^3-r^2z^2-\frac{r^4z^4}{1+rz}&r^3z^3-r^2z^2-\frac{r^4z^4}{1+rz}& -r^3z^3-r^2z^2-\frac{3r^4z^4}{1+rz} & r^3z^3+r^2z^2+\frac{3r^4z^4}{1+rz}\cr}. \nonumber
\end{eqnarray}
\end{scriptsize}
By comparing  the expressions of  $g_{\rho^{(0)}}^{(1)}(\mathbf{0},\mathbf{0})$ with the one in Lemma \ref{le:g reduction sierpinski},  for $n=1$, we 
have obtained  the formulas for $u^{(1)}_5$ and  $u^{(1)}_6$.
\end{proof}

\begin{proof}[Proof of Theorem \ref{th:T formula}] By Lemma \ref{le:fast inversion}(b) and  Theorem \ref{th:Dirichlet g}(b), on $F^{(n+1)}$ we have
\begin{align*}
u^{(n+1)}_3 &= [\rho^{(n)}(\mathbf{0},\mathbf{1})g_{\rho^{(n)}}^{(1)}(\mathbf{1},\mathbf{5})]^0_5+[\rho^{(n)}(\mathbf{0},\mathbf{2})g_{\rho^{(n)}}^{(1)}(\mathbf{2},\mathbf{5})]^0_5\\
&= u^{(n)}_4 g_{\rho^{(n)}}^{(1)}(\mathbf{1},\mathbf{5})^4_5 + u^{(n)}_3 g_{\rho^{(n)}}^{(1)}(\mathbf{1},\mathbf{5})^5_5 - u^{(n)}_4 g_{\rho^{(n)}}^{(1)}(\mathbf{2},\mathbf{5})^3_5 - u^{(n)}_3 g_{\rho^{(n)}}^{(1)}(\mathbf{2},\mathbf{5})^2_5 \\
& = u^{(n)}_4(g_{\rho^{(n)}}^{(1)}(\mathbf{1},\mathbf{5})^4_5 - g_{\rho^{(n)}}^{(1)}(\mathbf{2},\mathbf{5})^3_5) + u^{(n)}_3 (g_{\rho^{(n)}}^{(1)}(\mathbf{1},\mathbf{5})^5_5 - g_{\rho^{(n)}}^{(1)}(\mathbf{2},\mathbf{5})^2_5)
\end{align*}

\begin{align*}
u^{(n+1)}_4 &= [\rho^{(n)}(\mathbf{0},\mathbf{1})g_{\rho^{(n)}}^{(1)}(\mathbf{1},\mathbf{5})]^0_4+[\rho^{(n)}(\mathbf{0},\mathbf{2})g_{\rho^{(n)}}^{(1)}(\mathbf{2},\mathbf{5})]^0_4\\
&= u^{(n)}_4 g_{\rho^{(n)}}^{(1)}(\mathbf{1},\mathbf{5})^4_4 + u^{(n)}_3 g_{\rho^{(n)}}^{(1)}(\mathbf{1},\mathbf{5})^5_4 - u^{(n)}_4 g_{\rho^{(n)}}^{(1)}(\mathbf{2},\mathbf{5})^3_4 - u^{(n)}_3 g_{\rho^{(n)}}^{(1)}(\mathbf{2},\mathbf{5})^2_4 \\
& = u^{(n)}_4(g_{\rho^{(n)}}^{(1)}(\mathbf{1},\mathbf{5})^4_4 - g_{\rho^{(n)}}^{(1)}(\mathbf{2},\mathbf{5})^3_4) + u^{(n)}_3 (g_{\rho^{(n)}}^{(1)}(\mathbf{1},\mathbf{5})^5_4 - g_{\rho^{(n)}}^{(1)}(\mathbf{2},\mathbf{5})^2_4)
\end{align*}

\begin{align*}
u^{(n+1)}_5 &= \rho^{(n)}(\mathbf{0},\mathbf{0})^0_0+[\rho^{(n)}(\mathbf{0},\mathbf{1})g_{\rho^{(n)}}^{(1)}(\mathbf{1},\mathbf{0})]^0_0+[\rho^{(n)}(\mathbf{0},\mathbf{2})g_{\rho^{(n)}}^{(1)}(\mathbf{2},\mathbf{0})]^0_0\\
&= u^{(n)}_5 + u^{(n)}_4 g_{\rho^{(n)}}^{(1)}(\mathbf{1},\mathbf{0})^4_0 + u^{(n)}_3 g_{\rho^{(n)}}^{(1)}(\mathbf{1},\mathbf{0})^5_0 - u^{(n)}_4 g_{\rho^{(n)}}^{(1)}(\mathbf{2},\mathbf{0})^3_0 - u^{(n)}_3 g_{\rho^{(n)}}^{(1)}(\mathbf{2},\mathbf{0})^2_0 \\
& = u^{(n)}_5 + u^{(n)}_4(g_{\rho^{(n)}}^{(1)}(\mathbf{1},\mathbf{0})^4_0 - g_{\rho^{(n)}}^{(1)}(\mathbf{2},\mathbf{0})^3_0) + u^{(n)}_3 (g_{\rho^{(n)}}^{(1)}(\mathbf{1},\mathbf{0})^5_0 - g_{\rho^{(n)}}^{(1)}(\mathbf{2},\mathbf{0})^2_0)
\end{align*}

Similarly, on $F^{(n+1)'}$ we have

\begin{align*}
u^{(n+1)}_1 &= [\rho^{(n)}(\mathbf{0},\mathbf{1'})g_{\rho^{(n)}}^{(1)}(\mathbf{1'},\mathbf{5'})]^0_1+[\rho^{(n)}(\mathbf{0},\mathbf{2'})g_{\rho^{(n)}}^{(1)}(\mathbf{2'},\mathbf{5'})]^0_1\\
&= u^{(n)}_1 g_{\rho^{(n)}}^{(1)}(\mathbf{1'},\mathbf{5'})^1_1 + u^{(n)}_2 g_{\rho^{(n)}}^{(1)}(\mathbf{1'},\mathbf{5'})^0_1 + u^{(n)}_1 g_{\rho^{(n)}}^{(1)}(\mathbf{2'},\mathbf{5'})^2_1 + u^{(n)}_2 g_{\rho^{(n)}}^{(1)}(\mathbf{2'},\mathbf{5'})^3_1 \\
& = u^{(n)}_1(g_{\rho^{(n)}}^{(1)}(\mathbf{1'},\mathbf{5'})^0_1 + g_{\rho^{(n)}}^{(1)}(\mathbf{2'},\mathbf{5'})^2_1) + u^{(n)}_2 (g_{\rho^{(n)}}^{(1)}(\mathbf{1'},\mathbf{5'})^0_1 + g_{\rho^{(n)}}^{(1)}(\mathbf{2'},\mathbf{5'})^3_1)
\end{align*}

\begin{align*}
u^{(n+1)}_2 &= [\rho^{(n)}(\mathbf{0},\mathbf{1'})g_{\rho^{(n)}}^{(1)}(\mathbf{1'},\mathbf{5'})]^0_0+[\rho^{(n)}(\mathbf{0},\mathbf{2'})g_{\rho^{(n)}}^{(1)}(\mathbf{2'},\mathbf{5'})]^0_0\\
&= u^{(n)}_1 g_{\rho^{(n)}}^{(1)}(\mathbf{1'},\mathbf{5'})^1_0 + u^{(n)}_2 g_{\rho^{(n)}}^{(1)}(\mathbf{1'},\mathbf{5'})^0_0 + u^{(n)}_1 g_{\rho^{(n)}}^{(1)}(\mathbf{2'},\mathbf{5'})^2_0 + u^{(n)}_2 g_{\rho^{(n)}}^{(1)}(\mathbf{2'},\mathbf{5'})^3_0 \\
& = u^{(n)}_1(g_{\rho^{(n)}}^{(1)}(\mathbf{1'},\mathbf{5'})^0_0 + g_{\rho^{(n)}}^{(1)}(\mathbf{2'},\mathbf{5'})^2_0) + u^{(n)}_2 (g_{\rho^{(n)}}^{(1)}(\mathbf{1'},\mathbf{5'})^0_0 + g_{\rho^{(n)}}^{(1)}(\mathbf{2'},\mathbf{5'})^3_0)
\end{align*}

\begin{align*}
u^{(n+1)}_6 &= \rho^{(n)}(\mathbf{0},\mathbf{0})^0_4 + [\rho^{(n)}(\mathbf{0},\mathbf{1'})g_{\rho^{(n)}}^{(1)}(\mathbf{1'},\mathbf{0})]^0_4+[\rho^{(n)}(\mathbf{0},\mathbf{2'})g_{\rho^{(n)}}^{(1)}(\mathbf{2'},\mathbf{0})]^0_4\\
&= u^{(n)}_6 + u^{(n)}_1 g_{\rho^{(n)}}^{(1)}(\mathbf{1'},\mathbf{0})^1_4 + u^{(n)}_2 g_{\rho^{(n)}}^{(1)}(\mathbf{1'},\mathbf{0})^0_4 + u^{(n)}_1 g_{\rho^{(n)}}^{(1)}(\mathbf{2'},\mathbf{0})^2_4 + u^{(n)}_2 g_{\rho^{(n)}}^{(1)}(\mathbf{2'},\mathbf{0})^3_4 \\
& = u^{(n)}_6 + u^{(n)}_1(g_{\rho^{(n)}}^{(1)}(\mathbf{1'},\mathbf{0})^0_4 + g_{\rho^{(n)}}^{(1)}(\mathbf{2'},\mathbf{0})^2_4) + u^{(n)}_2 (g_{\rho^{(n)}}^{(1)}(\mathbf{1'},\mathbf{0})^0_4 + g_{\rho^{(n)}}^{(1)}(\mathbf{2'},\mathbf{0})^3_4)
\end{align*}
\end{proof}

\section{Proofs of Main Results}\label{sec:proofs of main theorems}

\subsection{Recurrence}

In this subsection,  we shall apply Theorem \ref{th:T formula} and Theorem \ref{th:Dirichlet g}, with the aid of Monte Carlo integration to obtain the recurrence of quantum random walks on Sierpinski gasket.

\begin{theorem} \label{th:recurrent}
With initial state $|0^0\rangle$, the quantum random walk on Sierpinski gasket is recurrent. 
\end{theorem}

\begin{proof} [Proof of Theorem \ref{th:recurrent}]

By Theorem \ref{th:T formula}, Theorem \ref{th:Dirichlet g} and Monte Carlo integration, we have 
$$\lim_{n \rightarrow \infty} P^{(n)}(\mathbf{0}, \mathbf{a_n'})^0_1 = \lim_{n \rightarrow \infty}\frac{1}{2\pi} \int_0^{2\pi} |u^{(n)}_1(e^{i\theta})|^2 d\theta = 0,$$
$$\lim_{n \rightarrow \infty} P^{(n)}(\mathbf{0}, \mathbf{a_n'})^0_0 = \lim_{n \rightarrow \infty}\frac{1}{2\pi} \int_0^{2\pi} |u^{(n)}_2(e^{i\theta})|^2 d\theta = 0,$$
$$\lim_{n \rightarrow \infty} P^{(n)}(\mathbf{0}, \mathbf{a_n})^0_4 =\lim_{n \rightarrow \infty} \frac{1}{2\pi} \int_0^{2\pi} |u^{(n)}_3(e^{i\theta})|^2 d\theta = 0,$$
$$\lim_{n \rightarrow \infty} P^{(n)}(\mathbf{0}, \mathbf{a_n})^0_5 = \lim_{n \rightarrow \infty}\frac{1}{2\pi} \int_0^{2\pi} |u^{(n)}_4(e^{i\theta})|^2 d\theta = 0,$$
$$\lim_{n \rightarrow \infty} P^{(n)}(\mathbf{0}, \mathbf{0})^0_0 =\lim_{n \rightarrow \infty} \frac{1}{2\pi} \int_0^{2\pi} |u^{(n)}_5(e^{i\theta})|^2 d\theta = \frac{1}{4},$$
$$\lim_{n \rightarrow \infty} P^{(n)}(\mathbf{0}, \mathbf{0})^0_4 = \lim_{n \rightarrow \infty}\frac{1}{2\pi} \int_0^{2\pi} |u^{(n)}_6(e^{i\theta})|^2 d\theta = \frac{1}{4}.$$

The numerical values of the limiting integrals are shown on Fig. \ref{fig:quantum hitting probability Sierpinski}. 
These results imply  that  the quantum random walk on Sierpinski gasket is recurrent.
\end{proof}

\subsection{Expected hitting time}

In this subsection, we consider the expected value of the first exit time.

Let $\partial (F^{(n)}\cup F^{(n)'}) := \{\mathbf{a_n}, \mathbf{b_n},\mathbf{a'_n}, \mathbf{b'_n}\}$. 
 Let 
 \begin{eqnarray}
 T^{(n)} = \inf \{ t\geq 1; w_t \in \partial (F^{(n)}\cup F^{(n)'}) \} 
 \end{eqnarray}
  be the first-passage time taken to exit $F^{(n)}\cup F^{(n)'}$ at the four vertices. 

The corresponding amplitude Green function for quantum random walk is defined by
\begin{equation}
g_1^{(n)}(z)(\mathbf{x},\mathbf{y})_{j}^{i} =\sum_{t=1}^{\infty}z^t \Psi (w_0 = \mathbf{x}^i , w_{t} = \mathbf{y}^j ,T^{(n)} = t) ,
\end{equation}
here and below we use subscript $1$ for quantities with the boundary conditions $\partial (F^{(n)}\cup F^{(n)'}) := \{\mathbf{a_n}, \mathbf{b_n},\mathbf{a'_n}, \mathbf{b'_n}\}$.
The probability that a quantum random walk starts with $\mathbf{0}^i$ and exits from $\mathbf{y}^j$ is given by
\begin{equation}
P_1^{(n)}(\mathbf{0},\mathbf{y})_{j}^{i} = \sum_{t=1}^{\infty}|\Psi (w_0 = \mathbf{x}^i , w_{t} = \mathbf{y}^j ,T^{(n)} = t)|^2 .
\end{equation}
In this section, we shall use the result of amplitude Green functions $u^{(n)}_i$ in the previous sections to find the expectation of first-passage time $E(T^{(n)})$ given $w_0 = \mathbf{0}^i$. We have 
\begin{scriptsize}
\begin{align*}
E(T^{(n)}) &= \sum_{\mathbf{y} \in \partial (F^{(n)}\cup F^{(n)'})}\sum_{\mathbf{e}_j \in out(\mathbf{y}) } \sum_{t=1}^{\infty}t P_1^{(n)}(\mathbf{0},\mathbf{y},T^{(n)} = t)^i_j \\
		   &= \sum_{\mathbf{y} \in \partial (F^{(n)}\cup F^{(n)'})}\sum_{\mathbf{e}_j \in out(\mathbf{y}) } \sum_{t=1}^{\infty} t|\Psi (w_0 = \mathbf{0}^i , w_{t} = \mathbf{y}^j ,T^{(n)} = t)|^2 \\
		   &= \sum_{\mathbf{y} \in \partial (F^{(n)}\cup F^{(n)'})}\sum_{\mathbf{e}_j \in out(\mathbf{y}) } \frac{1}{2\pi}\int_0^{2\pi} \sum_{t=1}^{\infty} e^{i\theta t} t\Psi (w_0 = \mathbf{0}^i , w_{t} = \mathbf{y}^j ,T^{(n)} = t) \sum_{t=1}^{\infty} e^{-i\theta t} \Psi (w_0 = \mathbf{0}^i , w_{t} = \mathbf{y}^j ,T^{(n)} = t) d\theta \\
		   & = \sum_{\mathbf{y} \in \partial (F^{(n)}\cup F^{(n)'})}\sum_{\mathbf{e}_j \in out(\mathbf{y}) } \frac{1}{2\pi}\int_0^{2\pi} (\partial_s g_1^{(n)}(e^{s+i\theta})(\mathbf{0},\mathbf{y})_{j}^{i}|_{s=0}) g_1^{(n)}(e^{-i\theta})(\mathbf{0},\mathbf{y})_{j}^{i} d\theta.
\end{align*}
\end{scriptsize}
So we have the following formula for expectation of hitting time
\begin{eqnarray}\label{eq:expected exit time formula}
&&E(T^{(n)})  \\
&=& \sum_{y \in \partial (F^{(n)}\cup F^{(n)'})}\sum_{\mathbf{e}_j \in out(\mathbf{y}) } \frac{1}{2\pi}\int_0^{2\pi}
(\partial_s g_1^{(n)}(e^{s+i\theta})(\mathbf{0},\mathbf{y})_{j}^{i}|_{s=0}) g_1^{(n)}(e^{-i\theta})(\mathbf{0},\mathbf{y})_{j}^{i} d\theta. \nonumber 
\end{eqnarray}

Given $u^{(n)}_i$ for $i = 1,...,6$, we will first compute the matrix $g_1^{(n+1)}(z)(\mathbf{0},\mathbf{5})$. Then the Green function from $\mathbf{0}$ to other vertices can be obtained by symmetry. Consider the relabeled figure again as Fig \ref{fig:F(1)  F(1)'}, we have
 \begin{equation}
  \begin{bmatrix}
g_1^{(n+1)}(z)(\mathbf{0},\mathbf{5})  \\ g_1^{(n+1)}(z)(\mathbf{1},\mathbf{5})  \\ g_1^{(n+1)}(z)(\mathbf{2},\mathbf{5}) \\g_1^{(n+1)}(z)(\mathbf{3},\mathbf{5}) \\g_1^{(n+1)}(z)(\mathbf{1'},\mathbf{5}) \\g_1^{(n+1)}(z)(\mathbf{2'},\mathbf{5}) \\g_1^{(n+1)}(z)(\mathbf{3'},\mathbf{5}) 
\end{bmatrix} = \frac{1}{I_{28}-[\rho^{(n)}(z)]|_{0,1,2,3,1',2',3'}} 
 \begin{bmatrix}
g_1^{(n)}(z)(\mathbf{0},\mathbf{5})  \\ g_1^{(n)}(z)(\mathbf{1},\mathbf{5})  \\ g_1^{(n)}(z)(\mathbf{2},\mathbf{5}) \\g_1^{(n)}(z)(\mathbf{3},\mathbf{5}) \\g_1^{(n)}(z)(\mathbf{1'},\mathbf{5}) \\g_1^{(n)}(z)(\mathbf{2'},\mathbf{5}) \\g_1^{(n)}(z)(\mathbf{3'},\mathbf{5}) 
\end{bmatrix},
 \end{equation}
where $[\rho^{(n)}(z)]|_{0,1,2,3,1',2',3'}$ is the transition matrix in terms of $u^{(n)}_i$, for $i = 1,...,6$.\\
For $n = 0$, by solving the above equation, we obtain
\begin{equation}\label{eq:g1 05}
g_1^{(1)}(z)(\mathbf{0},\mathbf{5}) = \bordermatrix{~& 0 & 1 & 4              & 5           \cr
                                                       0  &0&0& \frac{z^2(z^2+z-2)}{2(z^3-z^2-4)} & \frac{z^5}{2(z^3-z^2-4)}   \cr
                                                       1  &0&0& \frac{z^2(z^2-z+2)}{2(z^3-z^2-4)} & \frac{z^5}{2(z^3-z^2-4)}   \cr
                                                       4  &0&0& \frac{-z^2(z+2)}{2(z^3-z^2-4)} & \frac{-z^3(z+2)}{2(z^3-z^2-4)}   \cr
                                                       5  &0&0& \frac{-z^2(z+2)}{2(z^3-z^2-4)} & \frac{-z^3(z+2)}{2(z^3-z^2-4)}   \cr}.
                                                       \end{equation}

If we start $w_0 = \mathbf{0}^0$, then each exit Green amplitude function can be obtained from the above matrix (\ref{eq:g1 05}). For example, we have 
\begin{eqnarray}
g_1^{(1)}(z)(\mathbf{0},\mathbf{4})_3 ^0 
= g_1^{(1)}(z)(\mathbf{0},\mathbf{5})_4^1 
= \frac{z^2(z^2-z+2)}{2(z^3-z^2-4)},
\end{eqnarray}
 and 
 \begin{eqnarray*}
 g_1^{(1)}(z)(\mathbf{0},\mathbf{5'})^0_1 
 = g_1^{(1)}(z)(\mathbf{0},\mathbf{4})^4 _4 
 = \frac{-z^2(z+2)}{2(z^3-z^2-4)}.
 \end{eqnarray*}
 
 By Parseval's identity, 
\begin{eqnarray} \label{eq:P 1 n}
P_1^{(n)}(\mathbf{x},\mathbf{y})_{j}^{i} 
= \frac{1}{2\pi}\int_0^{2\pi} |g_1^{(n)}(e^{i\theta})(\mathbf{x},\mathbf{y})_{j}^{i}|^2 d\theta.
\end{eqnarray}
So we have 
\begin{equation}
 P_1^{(1)}(z)(\mathbf{0},\mathbf{5}) = \bordermatrix{~& 0 & 1 & 4              & 5           \cr
                                                       0  &0&0& \frac{17}{132} & \frac{19}{1056}   \cr
                                                       1  &0&0& \frac{25}{264} & \frac{19}{1056}   \cr
                                                       4  &0&0& \frac{91}{1056} & \frac{91}{1056}   \cr
                                                       5  &0&0& \frac{91}{1056} & \frac{91}{1056}   \cr}.
                                                       \end{equation}
Using (\ref{eq:expected exit time formula}), we have $E(T^{(1)}) = \frac{2173}{1152}$.
                                            
The results of probability distribution of $\mathbf{x}^i_{T^{(1)}}$ is summarized in Figure \ref{fig:probability distribution of x i T 1 }  and the expected values of $T^{(1)}$  are shown in Figure      \ref{fig:expectations of x i T 1 }.                                          
\begin{figure}
\begin{minipage}{0.45\textwidth}
\begin{tikzpicture}[scale = 0.7]
\begin{axis}[dashed, xmin=-3,ymin=-3,xticklabels = {,,}, yticklabels = {},xmax=5,ymax=3,]
\addplot[mark = none] coordinates{(-5,0) (5,0)};
\addplot[mark = none] coordinates{(-5,1) (5,1)};
\addplot[mark = none] coordinates{(-5,2) (5,2)};
\addplot[mark = none] coordinates{(-5,-1) (5,-1)};
\addplot[mark = none] coordinates{(-5,-2) (5,-2)};
\addplot[mark = none] coordinates{(0,-3) (0,3)};
\addplot[mark = none] coordinates{(1,-3) (1,3)};
\addplot[mark = none] coordinates{(2,-3) (2,3)};
\addplot[mark = none] coordinates{(3,-3) (3,3)};
\addplot[mark = none] coordinates{(4,-3) (4,3)};
\addplot[mark = none] coordinates{(-1,-3) (-1,3)};
\addplot[mark = none] coordinates{(-2,-3) (-2,3)};
\addplot[mark = none] coordinates{(-3,-3) (-3,3)};
\addplot[mark = none] coordinates{(-4,-3) (-4,3)};
\addplot[mark = *, solid, thick] coordinates {(0, 0) (2,0) (4,0) (3,1) (2,2) (1,1) (0,0)};
\addplot[mark = *, solid ,thick] coordinates {(1,1) (3,1) (2,0) (1,1)};
\addplot[mark = *, solid, thick] coordinates {(0, 0) (-1,-1) (-2,-2) (0,-2) (2,-2) (1,-1) (0,0)};
\addplot[mark = *, solid ,thick] coordinates {(-1,-1) (0,-2) (1,-1) (-1,-1)};
\addplot[mark=*] coordinates {(0,0)} node(1)[]{} ;
\addplot[mark=] coordinates {(-0.5,0)}[xshift=0pt,yshift=-10pt] node[label={$\mathbf{0}$}]{} ;
\addplot[mark=*] coordinates {(1,1)} node(2)[]{} ;
\addplot[mark=*] coordinates {(2,0)} node(3)[xshift=0pt,yshift=-20pt][]{} ;
\addplot[mark=*] coordinates {(3,1)} node(4)[]{} ;
\addplot[mark=*] coordinates {(2,2)} node(5)[label={$\mathbf{a_1}$}]{} ;
\addplot[mark=*] coordinates {(4,0)} node(6)[]{} ;
\addplot[mark=] coordinates {(4.5,0)} node[xshift=0pt,yshift=-10pt][label={$\mathbf{b_1}$}]{} ;
\addplot[mark=*] coordinates {(-1,-1)} node(-2)[]{} ;
\addplot[mark=*] coordinates {(1,-1)} node(-3)[]{} ;
\addplot[mark=*] coordinates {(0,-2)} node(-4)[]{} ;
\addplot[mark=*] coordinates {(-2,-2)} node(-5)[]{} ;
\addplot[mark=] coordinates {(-2.5,-2)} node[xshift=2pt,yshift=-12pt][label={$\mathbf{a'_1}$}]{} ;
\addplot[mark=*] coordinates {(2,-2)} node(-6)[]{} ;
\addplot[mark=] coordinates {(2.5,-2)} node[xshift=0pt,yshift=-12pt][label={$\mathbf{b'_1}$}]{} ;
\addplot[mark=none] coordinates {(-1,-1)} node(-1)[]{} ;
\addplot[mark=none] coordinates {(1,-1)} node(-2)[]{} ;
\draw [->,thick ,solid] (1) -- (5)node[xshift=-17pt,yshift=-5pt] {$\frac{17}{132}$};
\draw [->,thick ,solid] (6) -- (5)node[xshift=17pt,yshift=-5pt]{$\frac{19}{1056}$};
\draw [->,thick ,solid] (1) -- (6)node[xshift=-12pt,yshift=-9pt] {$\frac{25}{264}$};
\draw [->,thick ,solid] (5) -- (6)node[xshift=0pt,yshift=17pt]{$\frac{19}{1056}$};

\draw [->,thick ,solid] (1) -- (-5)node[xshift=2pt,yshift=18pt] {$\frac{91}{1056}$};
\draw [->,thick ,solid] (-6) -- (-5)node[xshift=15pt,yshift=-8pt]{$\frac{91}{1056}$};
\draw [->,thick ,solid] (1) -- (-6)node[xshift=-2pt,yshift=18pt] {$\frac{91}{1056}$};
\draw [->,thick ,solid] (-5) -- (-6)node[xshift=-15pt,yshift=-8pt]{$\frac{91}{1056}$};

\end{axis}        
\end{tikzpicture}
\caption{probability distribution of $\mathbf{x}^i_{T^{(1)}}$ on $F^{(1)} \cup F^{(1)'}$ with $w_0 = \mathbf{0}^0$}\label{fig:probability distribution of x i T 1 }
\end{minipage}
\hfill
\begin{minipage}{0.45\textwidth}
\begin{tikzpicture}[scale = 0.7]
\begin{axis}[dashed, xmin=-3,ymin=-3,xticklabels = {,,}, yticklabels = {},xmax=5,ymax=3,]
\addplot[mark = none] coordinates{(-5,0) (5,0)};
\addplot[mark = none] coordinates{(-5,1) (5,1)};
\addplot[mark = none] coordinates{(-5,2) (5,2)};
\addplot[mark = none] coordinates{(-5,-1) (5,-1)};
\addplot[mark = none] coordinates{(-5,-2) (5,-2)};
\addplot[mark = none] coordinates{(0,-3) (0,3)};
\addplot[mark = none] coordinates{(1,-3) (1,3)};
\addplot[mark = none] coordinates{(2,-3) (2,3)};
\addplot[mark = none] coordinates{(3,-3) (3,3)};
\addplot[mark = none] coordinates{(4,-3) (4,3)};
\addplot[mark = none] coordinates{(-1,-3) (-1,3)};
\addplot[mark = none] coordinates{(-2,-3) (-2,3)};
\addplot[mark = none] coordinates{(-3,-3) (-3,3)};
\addplot[mark = none] coordinates{(-4,-3) (-4,3)};
\addplot[mark = *, solid, thick] coordinates {(0, 0) (2,0) (4,0) (3,1) (2,2) (1,1) (0,0)};
\addplot[mark = *, solid ,thick] coordinates {(1,1) (3,1) (2,0) (1,1)};
\addplot[mark = *, solid, thick] coordinates {(0, 0) (-1,-1) (-2,-2) (0,-2) (2,-2) (1,-1) (0,0)};
\addplot[mark = *, solid ,thick] coordinates {(-1,-1) (0,-2) (1,-1) (-1,-1)};
\addplot[mark=*] coordinates {(0,0)} node(1)[]{} ;
\addplot[mark=] coordinates {(-0.5,0)}[xshift=0pt,yshift=-10pt] node[label={$\mathbf{0}$}]{} ;
\addplot[mark=*] coordinates {(1,1)} node(2)[]{} ;
\addplot[mark=*] coordinates {(2,0)} node(3)[xshift=0pt,yshift=-20pt][]{} ;
\addplot[mark=*] coordinates {(3,1)} node(4)[]{} ;
\addplot[mark=*] coordinates {(2,2)} node(5)[label={$\mathbf{a_1}$}]{} ;
\addplot[mark=*] coordinates {(4,0)} node(6)[]{} ;
\addplot[mark=] coordinates {(4.5,0)} node[xshift=0pt,yshift=-10pt][label={$\mathbf{b_1}$}]{} ;
\addplot[mark=*] coordinates {(-1,-1)} node(-2)[]{} ;
\addplot[mark=*] coordinates {(1,-1)} node(-3)[]{} ;
\addplot[mark=*] coordinates {(0,-2)} node(-4)[]{} ;
\addplot[mark=*] coordinates {(-2,-2)} node(-5)[]{} ;
\addplot[mark=] coordinates {(-2.5,-2)} node[xshift=2pt,yshift=-12pt][label={$\mathbf{a'_1}$}]{} ;
\addplot[mark=*] coordinates {(2,-2)} node(-6)[]{} ;
\addplot[mark=] coordinates {(2.5,-2)} node[xshift=0pt,yshift=-12pt][label={$\mathbf{b'_1}$}]{} ;
\addplot[mark=none] coordinates {(-1,-1)} node(-1)[]{} ;
\addplot[mark=none] coordinates {(1,-1)} node(-2)[]{} ;
\draw [->,thick ,solid] (1) -- (5)node[xshift=-17pt,yshift=-5pt] {$\frac{7093}{17424}$};
\draw [->,thick ,solid] (6) -- (5)node[xshift=17pt,yshift=-5pt]{$\frac{27073}{278784}$};
\draw [->,thick ,solid] (1) -- (6)node[xshift=-12pt,yshift=-9pt] {$\frac{18049}{69696}$};
\draw [->,thick ,solid] (5) -- (6)node[xshift=0pt,yshift=17pt]{$\frac{27073}{278784}$};

\draw [->,thick ,solid] (1) -- (-5)node[xshift=2pt,yshift=18pt] {$\frac{59497}{278784}$};
\draw [->,thick ,solid] (-6) -- (-5)node[xshift=15pt,yshift=-8pt]{$\frac{83521}{278784}$};
\draw [->,thick ,solid] (1) -- (-6)node[xshift=-2pt,yshift=18pt] {$\frac{59497}{278784}$};
\draw [->,thick ,solid] (-5) -- (-6)node[xshift=-15pt,yshift=-8pt]{$\frac{83521}{278784}$};

\end{axis}        
\end{tikzpicture}
\caption { $E(T^{(1)})$ on $F^{(1)} \cup F^{(1)'}$ with $w_0 = \mathbf{0}^0$}\label{fig:expectations of x i T 1 }
\end{minipage}
\end{figure}
We note  that unlike $\tau^{(1)}$, if we sum up the exit probabilities for $T^{(1)}$, we have $P_1(T^{(1)} < \infty) =  \frac{319}{528} < 1$. So we obtain 
$$     E(T^{(1)}| T^{(1)} < \infty) = \frac{2173}{696} = 3.121. $$ 
This is an exact value.

 \subsection{Results of main exponents}
 
In this section, we apply recursive formulas as the proof of Theroem 4 and using Monte Carlo integration, we obtain the numerical values of exponents listed on the Table \ref{tab:localization table}, the results are shown in Figures \ref{fig:quantum d} -\ref{fig:quantum P(0, an)44}.

\begin{figure}[htbp]
\centering
\includegraphics[trim={0 8cm 0 7cm},clip, width=0.5\textwidth]{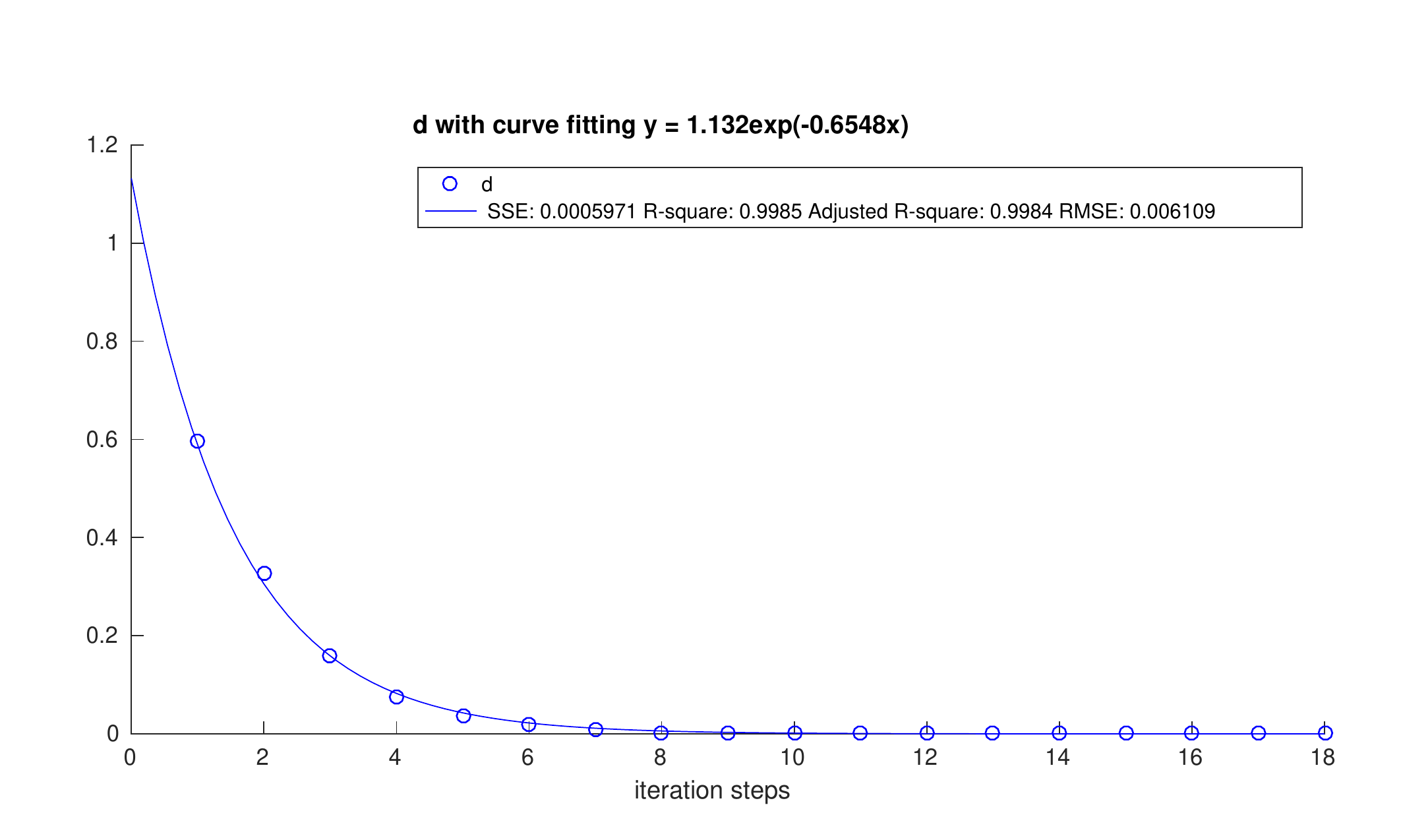}
\caption{Quantum $\delta_w$} \label{fig:quantum d}
\includegraphics[trim={0 8cm 0 7cm},clip, width=0.5\textwidth]{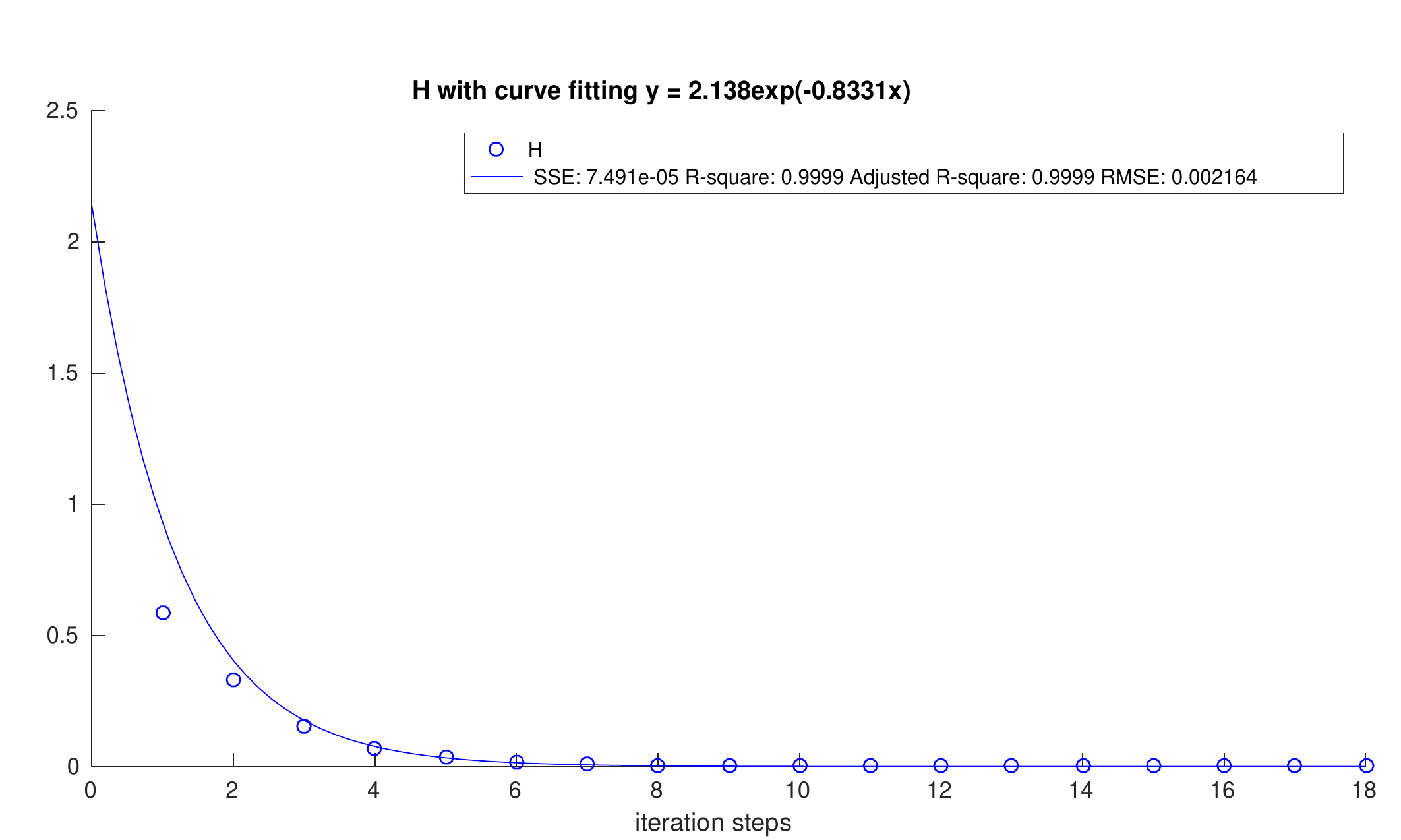}
\caption{Quantum $\eta_w$}   \label{fig:quantum H}
\includegraphics[trim={0 8cm 0 7cm},clip, width=0.5\textwidth]{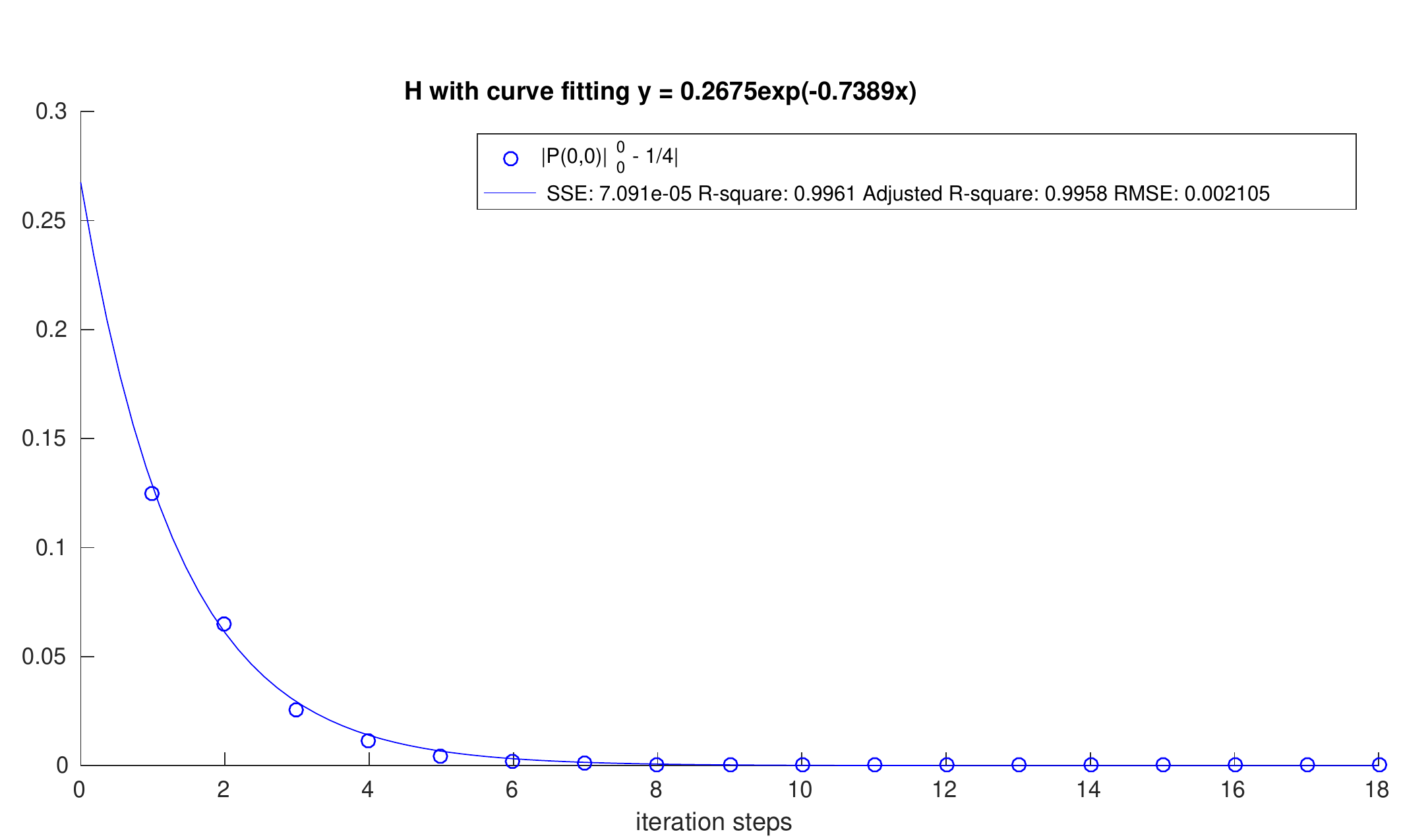}
\caption{Quantum $\beta_w(0, 0)^0_0$} \label{fig:quantum P(0, 0)00}
\includegraphics[trim={0 8cm 0 7cm},clip, width=0.5\textwidth]{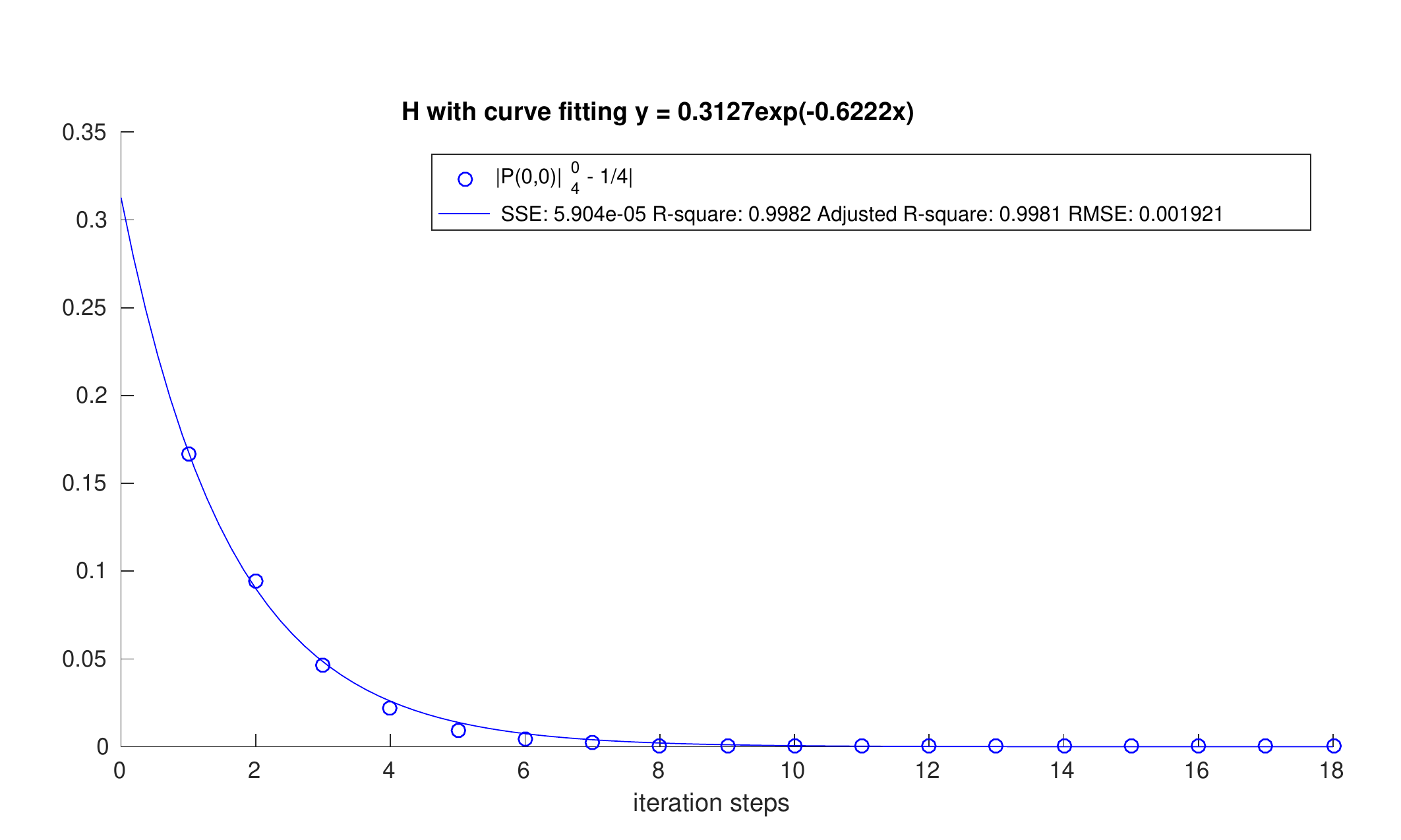}
\caption{Quantum $\beta_w(0, 0)^0_4$}   \label{fig:quantum P(0, 0)04}
\end{figure}

\begin{figure}[htbp]
\centering
\includegraphics[trim={0 8cm 0 7cm},clip, width=0.5\textwidth]{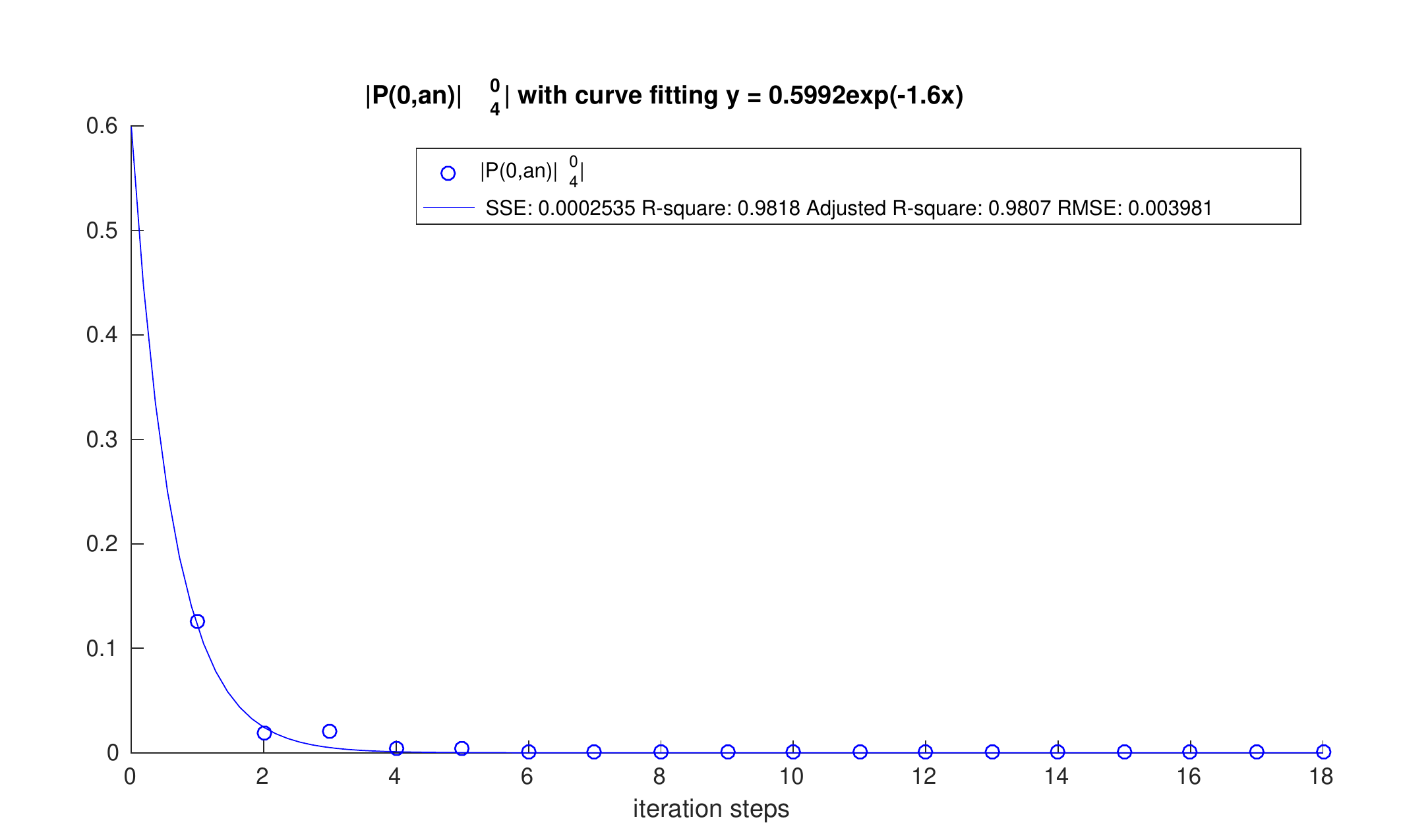}
\caption{Quantum $\gamma_w(0, a_n)^0_4$}\label{fig:quantum P(0, an)04}
\includegraphics[trim={0 8cm 0 7cm},clip, width=0.5\textwidth]{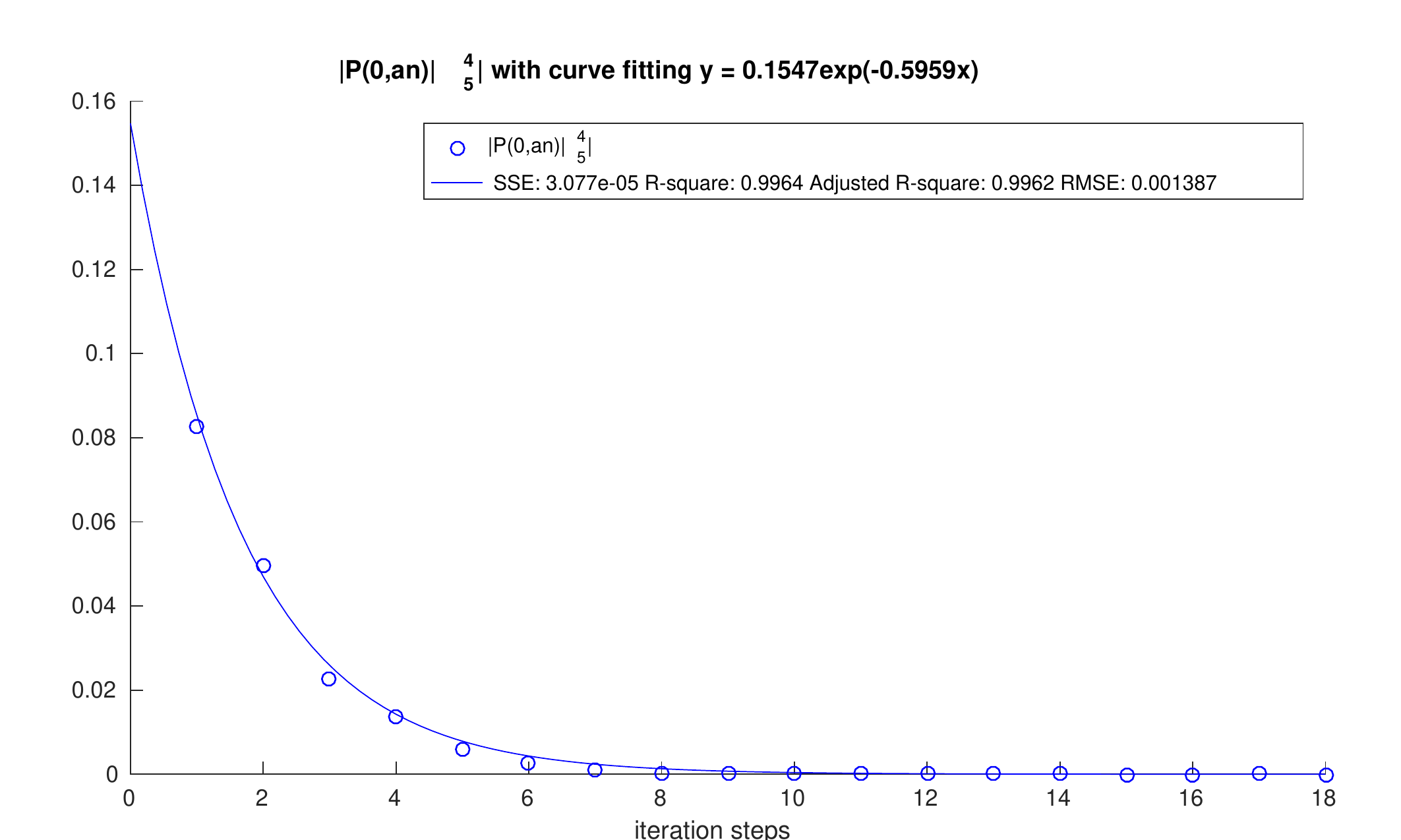}
\caption{Quantum $\gamma_w(0, a_n)^4_5$}\label{fig:quantum P(0, an)45}
\includegraphics[trim={0 8cm 0 7cm},clip, width=0.5\textwidth]{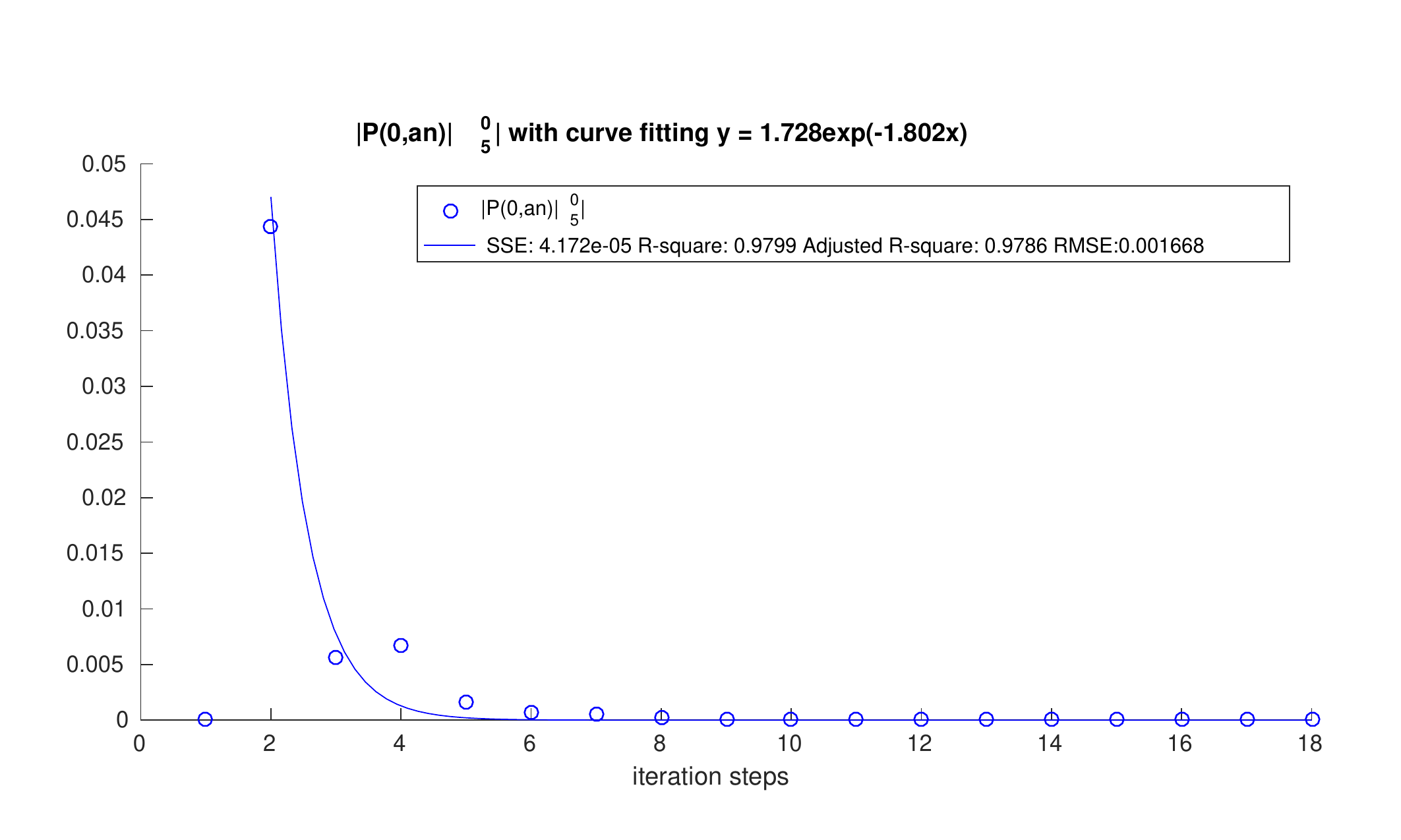}
\caption{Quantum $\gamma_w(0, a_n)^0_5$}\label{fig:quantum P(0, an)05}
\includegraphics[trim={0 8cm 0 7cm},clip, width=0.5\textwidth]{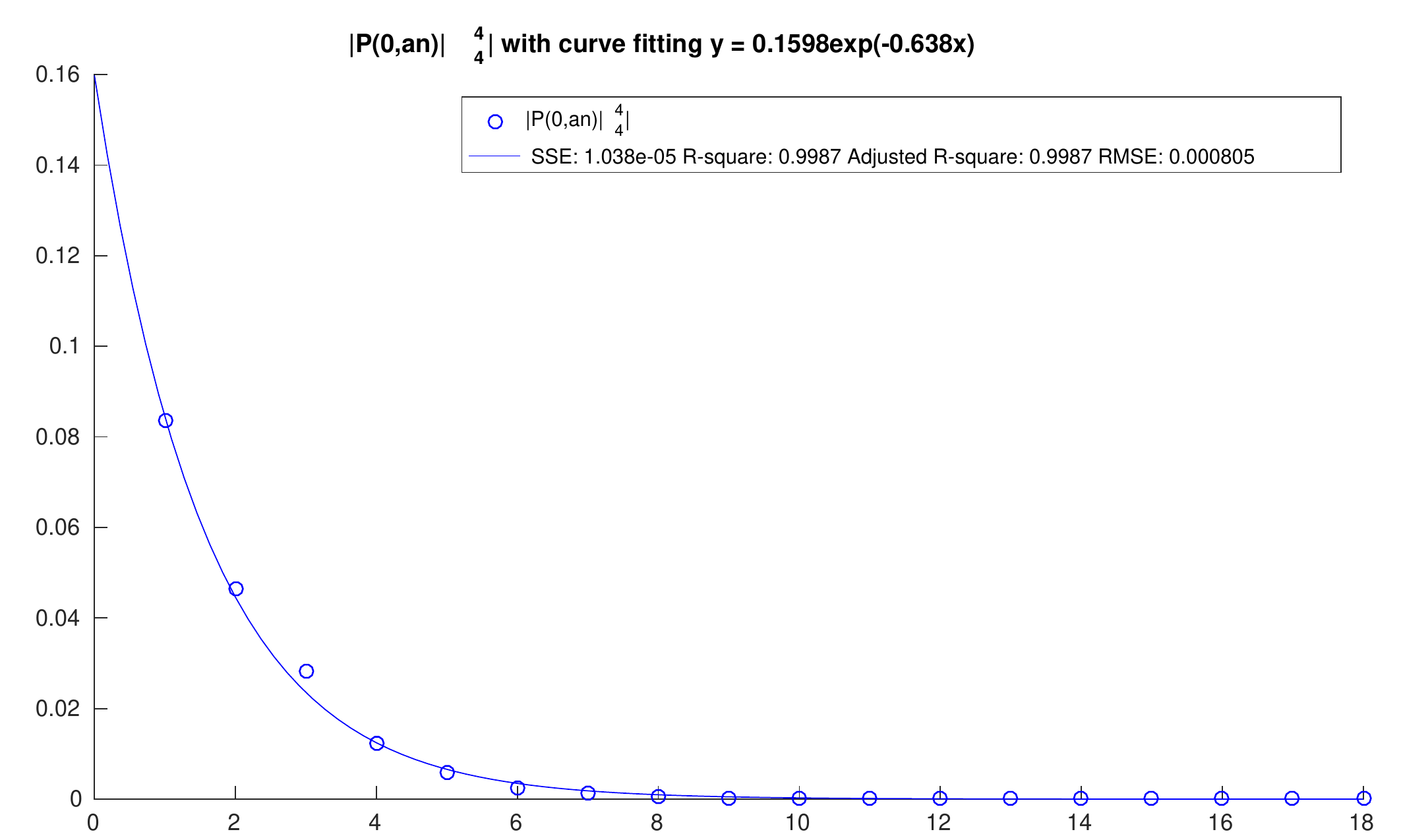}
\caption{Quantum $\gamma_w(0, a_n)^4_4$}\label{fig:quantum P(0, an)44}
\end{figure}

\section{Conclusion}
For quantum random walks on Sierpinski gasket, we have obtained the recursive formulas for amplitude Green functions Theorem \ref{th:T formula} and solved quantum version of Poisson equations Theorem \ref{th:Dirichlet g}. We see that the quantum random walks are more likely to return to the origin  and less tangling than the classical random walks on the same space. These results suggest that quantum random walks may be useful for spatial search on Sierpinski like structures. Our methods rely on  the self-similar structure of the underlying space and they should apply to any other underlying space with self-similar structure as well.

\newpage

\appendix
\section{Classical random walks on Sierpinski gasket}

In this section, for the convenience  of a reader and for completeness, we include here the analysis of the classical random walks on the Sierpinski gaskets. 
Our analysis yields known result on $d_{w, c}$, see e.g. \cite{2005 diffusion},  as well as new result on the detail probability distributions and expected hitting times with respected to  each of the departing and arriving directions.

The $n^{th}$ order Sierpinski gasket $F^{(n)}$ is a degree-4 regular graph except overall corners with Hausdorff dimension of $d_f = \ln3/\ln2$, which is bigger than dimension of the line and smaller than that of the plane. The graph of $F^{(2)}$ is shown in Figure 2.1. For the $n^{th}$ order Sierpinski gasket $F^{(n)}$ embedded in the two-dimensional plane, we can extend $F^{(n)}$ as following: if we think of $\mathbf{0},\mathbf{a_n},\mathbf{b_n}$ as boundary points of $F^{(n)}$, then we call reflection of $F^{(n)}$ denoted by $F^{(n)'}$ with boundary points $\mathbf{0},\mathbf{a'_n},\mathbf{b'_n}$, here $\mathbf{a'_n}= -\mathbf{a_n}=(-2^n,-2^n)$ and $\mathbf{b'_n} = (2^n,-2^n)$, see Figure 2.2.

\begin{figure}[htb] \label{fig F^n}
\begin{center}   
\begin{minipage}{0.45\textwidth}
\begin{tikzpicture}[scale = 0.7]
\begin{axis}[dashed, xmin=0,ymin=0,xtick = {0,1,...,9},xmax=8.4,ymax=4.5]
\addplot[mark = none] coordinates{(1,0) (1,5)};
\addplot[mark = none] coordinates{(2,0) (2,5)};
\addplot[mark = none] coordinates{(3,0) (3,5)};
\addplot[mark = none] coordinates{(4,0) (4,5)};
\addplot[mark = none] coordinates{(5,0) (5,5)};
\addplot[mark = none] coordinates{(6,0) (6,5)};
\addplot[mark = none] coordinates{(7,0) (7,5)};
\addplot[mark = none] coordinates{(8,0) (8,5)};
\addplot[mark = none] coordinates{(0,1) (9,1)};
\addplot[mark = none] coordinates{(0,2) (9,2)};
\addplot[mark = none] coordinates{(0,3) (9,3)};
\addplot[mark = none] coordinates{(0,4) (9,4)};
\addplot[mark = *, solid, thick] coordinates {(0, 0) (2,0) (4,0) (6,0) (8,0) (7,1) (6,2) (5,3) (4,4) (3,3) (2,2) (1,1) (0,0)};
\addplot[mark = *, solid ,thick] coordinates {(1,1) (3,1) (2,0) (1,1)};
\addplot[mark = *, solid, thick] coordinates {(2,2) (4,2) (6,2) (5,1) (4,0) (3,1) (2,2)};
\addplot[mark = *, solid ,thick] coordinates {(3,3) (5,3) (4,2) (3,3)};
\addplot[mark = *, solid, thick] coordinates {(5,1) (7,1) (6,0) (5,1)};
\addplot[mark=*] coordinates {(2,0)} node[label={$b_0$}]{} ;
\addplot[mark=*] coordinates {(1,1)} node[label={$a_0$}]{} ;
\addplot[mark=*] coordinates {(4,0)} node[label={$b_1$}]{} ;
\addplot[mark=*] coordinates {(2,2)} node[label={$a_1$}]{} ;
\addplot[mark=*] coordinates {(8,0)} node[label={$b_2$}]{} ;
\addplot[mark=*] coordinates {(4,4)} node[label={$a_2$}]{} ;

\end{axis}        
\end{tikzpicture}
\caption{$2^{nd}$ order Sierpinski Gasket $F^{(2)}$}
\end{minipage}
\hfill
\begin{minipage}{0.45\textwidth}
\begin{tikzpicture}[scale = 0.7]
\begin{axis}[dashed, xmin=-5,ymin=-5,xticklabels = {,,}, yticklabels = {},xmax=9,ymax=5,]
\addplot[mark = none] coordinates{(-5,0) (9,0)};
\addplot[mark = none] coordinates{(-5,1) (9,1)};
\addplot[mark = none] coordinates{(-5,2) (9,2)};
\addplot[mark = none] coordinates{(-5,3) (9,3)};
\addplot[mark = none] coordinates{(-5,4) (9,4)};
\addplot[mark = none] coordinates{(-5,-1) (9,-1)};
\addplot[mark = none] coordinates{(-5,-2) (9,-2)};
\addplot[mark = none] coordinates{(-5,-3) (9,-3)};
\addplot[mark = none] coordinates{(-5,-4) (9,-4)};
\addplot[mark = none] coordinates{(-5,-1) (5,-1)};
\addplot[mark = none] coordinates{(-5,-2) (5,-2)};
\addplot[mark = none] coordinates{(0,-5) (0,9)};
\addplot[mark = none] coordinates{(1,-5) (1,9)};
\addplot[mark = none] coordinates{(2,-5) (2,9)};
\addplot[mark = none] coordinates{(3,-5) (3,9)};
\addplot[mark = none] coordinates{(4,-5) (4,9)};
\addplot[mark = none] coordinates{(5,-5) (5,9)};
\addplot[mark = none] coordinates{(6,-5) (6,9)};
\addplot[mark = none] coordinates{(7,-5) (7,9)};
\addplot[mark = none] coordinates{(8,-5) (8,9)};
\addplot[mark = none] coordinates{(-1,-5) (-1,9)};
\addplot[mark = none] coordinates{(-2,-5) (-2,9)};
\addplot[mark = none] coordinates{(-3,-5) (-3,9)};
\addplot[mark = none] coordinates{(-4,-5) (-4,9)};
\addplot[mark = *, solid, thick] coordinates {(0, 0) (2,0) (4,0) (6,0) (8,0) (7,1) (6,2) (5,3) (4,4) (3,3) (2,2) (1,1) (0,0)};
\addplot[mark = *, solid ,thick] coordinates {(1,1) (3,1) (2,0) (1,1)};
\addplot[mark = *, solid, thick] coordinates {(2,2) (4,2) (6,2) (5,1) (4,0) (3,1) (2,2)};
\addplot[mark = *, solid ,thick] coordinates {(3,3) (5,3) (4,2) (3,3)};
\addplot[mark = *, solid, thick] coordinates {(5,1) (7,1) (6,0) (5,1)};
\addplot[mark=*] coordinates {(0,0)} node[label={$0$}]{} ;
\addplot[mark=*] coordinates {(2,0)} node[label={$b_0$}]{} ;
\addplot[mark=*] coordinates {(1,1)} node[label={$a_0$}]{} ;
\addplot[mark=*] coordinates {(4,0)} node[label={$b_1$}]{} ;
\addplot[mark=*] coordinates {(2,2)} node[label={$a_1$}]{} ;
\addplot[mark=*] coordinates {(8,0)} node[label={$b_2$}]{} ;
\addplot[mark=*] coordinates {(4,4)} node[label={$a_2$}]{} ;
\addplot[mark = *, solid, thick] coordinates {(0, 0) (-1,-1) (-2,-2) (-3,-3) (-4,-4) (-2,-4) (0,-4) (2,-4) (4,-4) (3,-3) (2,-2) (1,-1) (0,0)};
\addplot[mark = *, solid ,thick] coordinates {(-1,-1) (1,-1) (0,-2) (-1,-1)};
\addplot[mark = *, solid, thick] coordinates {(-2,-2) (0,-2) (2,-2) (1,-3) (0,-4) (-1,-3) (-2,-2)};
\addplot[mark = *, solid ,thick] coordinates {(-3,-3) (-1,-3) (-2,-4) (-3,-3)};
\addplot[mark = *, solid, thick] coordinates {(1,-3) (3,-3) (2,-4) (1,-3)};
\addplot[mark=*] coordinates {(1,-1)} node[xshift=0pt,yshift=-4pt][label={$b'_0$}]{} ;
\addplot[mark=*] coordinates {(-1,-1)} node[xshift=0pt,yshift=-4pt][label={$a'_0$}]{} ;
\addplot[mark=*] coordinates {(2,-2)} node[xshift=0pt,yshift=-4pt][label={$b'_1$}]{} ;
\addplot[mark=*] coordinates {(-2,-2)} node[xshift=0pt,yshift=-4pt][label={$a'_1$}]{} ;
\addplot[mark=*] coordinates {(4,-4)} node[xshift=0pt,yshift=-4pt][label={$b'_2$}]{} ;
\addplot[mark=*] coordinates {(-4,-4)} node[xshift=0pt,yshift=-4pt][label={$a'_2$}]{} ;
\end{axis}        
\end{tikzpicture}
\caption{$F^{(2)} \cup F^{(2)'}$}
\end{minipage}
\end{center}
\end{figure}

Let us introduce the random walk on generations of Sierpinski Gasket. The walker can reside on any site of the lattice, at each time it only moves to one of its neighbors with equal probability $1/4$.
A path $w$ is defined by $w = (w_0, \dots, w_m)$, where $w_t = \mathbf{x}_t \in F^{(n)} \cup F^{(n)'} $, and $\Vert \mathbf{x}_{t+1} - \mathbf{x}_t \Vert =2$, where $\Vert \mathbf{x} -\mathbf{y}\Vert = |x_1 - y_1|+|x_2-y_2|.$ Each step $w_t$ is updated from one of four neighbors of $w_{t-1}$ with equal probability $1/4$.
The length of $w$ is defined by $|w| = m$.

\begin{definition}  
Let $\mathring{F}^{(n)} = \{ x \in F^{(n)}; x \neq \mathbf{0}, \mathbf{a_n}, \mathbf{b_n} \}$, and 
${\tau^{(n)}} = \inf \{ t\geq 1; w_t \not\in \mathring{F}^{(n)}\cup \mathring{F}^{(n)'}  \}$ be the first exiting time of $\mathring{F}^{(n)}\cup \mathring{F}^{(n)'}$.
Let $$W_1^{(n)} = \{w = (w_0 , \dots, w_{\tau^{(n)}}), w_0 = \mathbf{0} , w_{\tau^{(n)}} = \mathbf{a_n} \}$$ the set of path with initial $\mathbf{0}$ first exit of $\mathring{F}^{(n)}\cup \mathring{F}^{(n)'}$ at $\mathbf{a_n}$. And similarly, let
$$W_0^{(n)} = \{w = (w_0 , \dots, w_{\tau^{(n)}}), w_0 = \mathbf{0} , w_{\tau^{(n)}} = \mathbf{0} \}$$
\end{definition}

\begin{definition}
Let $$S_0(w) =\# \{i; w_{i+1} = w_i, i = 0,1,...,\tau^{(1)}-1\}$$ be the number of self-loops. And
$$S_1(w) =\# \{i; \Vert w_{i+1}- w_i\Vert = 2, i = 0,1,...,\tau^{(1)}-1\}.$$
\end{definition}

Let $$\Phi_0^{(n)}(\alpha,\beta) = \sum_{w\in W_0^{(n)}} \alpha^{S_0(w)}\beta^{S_1(w)}, \ \text{and}
\ \Phi_1^{(n)}(\alpha,\beta) = \sum_{w\in W_1^{(n)}} \alpha^{S_0(w)}\beta^{S_1(w)}. $$ 
By assigning $\alpha = 0,$ and $\beta = 1/4$ we have the hitting probability 
$$P^{(n)}(w_{\tau^{(n)}} = \mathbf{0}|w_0 = \mathbf{0}) = \Phi_0^{(n)}(0,1/4),$$
$$P^{(n)}(w_{\tau^{(n)}} = \mathbf{a_n}|w_0 = \mathbf{0}) = \Phi_1^{(n)}(0,1/4).$$

In general we can obtain a map $T:\mathbb{R}^2 \rightarrow \mathbb{R}^2$ such that
$$(\Phi_0^{(n+1)},\Phi_1^{(n+1)}) = T(\Phi_0^{(n)},\Phi_1^{(n)}).$$
Using the self-similar property of the Sierpinski gasket, we have
\begin{lemma}
Let $$\Phi_0^{(n)}(\alpha,\beta) = \sum_{w\in W_0^{(n)}} \alpha^{S_0(w)}\beta^{S_1(w)}, $$ and
$$\Phi_1^{(n)}(\alpha,\beta) = \sum_{w\in W_1^{(n)}} \alpha^{S_0(w)}\beta^{S_1(w)}. $$ Then 
$\Phi_0^{(0)}(\alpha,\beta) = \alpha,\Phi_1^{(0)}(\alpha,\beta)= \beta,$
$$\Phi_0^{(n+1)}(\alpha,\beta) = \Phi_0^{(1)}(\Phi_0^{(n)}(\alpha,\beta),\Phi_1^{(n)}(\alpha,\beta)), $$ and
$$\Phi_1^{(n+1)}(\alpha,\beta) = \Phi_1^{(1)}(\Phi_0^{(n)}(\alpha,\beta),\Phi_1^{(n)}(\alpha,\beta)). $$
\end{lemma}

\begin{theorem}\label{th:classical recursive}
The recursive relation between $(\Phi_0^{(n+1)},\Phi_1^{(n+1)})$ and $(\Phi_0^{(n)},\Phi_1^{(n)})$ is
$$
\Phi_0^{(n+1)}=\Phi_0^{(1)}(\Phi_0^{(n)},\Phi_1^{(n)})=\Phi_0^{(n)}+\frac{4(\Phi_1^{(n)})^2(1-\Phi_0^{(n)})}{(1-\Phi_0^{(n)}+\Phi_1^{(n)})(1-\Phi_0^{(n)}-2\Phi_1^{(n)})}
$$
and
$$\Phi_1^{(n+1)} = \Phi_1^{(1)}(\Phi_0^{(n)},\Phi_1^{(n)})= \frac{(\Phi_1^{(n)})^2(1-\Phi_0^{(n)}+2\Phi_1^{(n)})}{(1-\Phi_0^{(n)}+\Phi_1^{(n)})(1-\Phi_0^{(n)}-2\Phi_1^{(n)})}.$$
\end{theorem}

Therefore, we obtain the map $T:\mathbb{R}^2 \rightarrow \mathbb{R}^2$, $T(\alpha,\beta) = (T_0(\alpha,\beta), T_1(\alpha,\beta))$:
$$(\Phi_0^{(n+1)},\Phi_1^{(n+1)}) = T(\Phi_0^{(n)},\Phi_1^{(n)}).$$
where 
$$
\Phi_0^{(n+1)}=T_0(\Phi_0^{(n)},\Phi_1^{(n)})=\Phi_0^{(n)}+\frac{4(\Phi_1^{(n)})^2(1-\Phi_0^{(n)})}{(1-\Phi_0^{(n)}+\Phi_1^{(n)})(1-\Phi_0^{(n)}-2\Phi_1^{(n)})}
$$
and
$$\Phi_1^{(n+1)} = T_1(\Phi_0^{(n)},\Phi_1^{(n)})= \frac{(\Phi_1^{(n)})^2(1-\Phi_0^{(n)}+2\Phi_1^{(n)})}{(1-\Phi_0^{(n)}+\Phi_1^{(n)})(1-\Phi_0^{(n)}-2\Phi_1^{(n)})}.$$

Let $\alpha=0$ and $\beta = \frac{1}{4}$ we have
$$ \Phi_1^{(1)}(0,\frac{1}{4}) = 0.15,\ \  \Phi_0^{(1)}(0,\frac{1}{4}) = 0.4
$$
$$ \Phi_1^{(2)}(0,\frac{1}{4}) = 0.09,\ \  \Phi_0^{(2)}(0,\frac{1}{4}) = 0.64
$$
$$ \Phi_1^{(3)}(0,\frac{1}{4}) = 0.054,\ \  \Phi_0^{(3)}(0,\frac{1}{4}) = 0.784
$$
$$ \vdots$$

\begin{corollary}
For initial probability distribution 
$\Phi_0^{(0)} = 0,\Phi_1^{(0)}=\frac{1}{4}$, we have
$$\Phi_0^{(n)} + 4\Phi_1^{(n)} = 1,$$ moreover 
$$\Phi_1^{(n+1)} = 0.6 \Phi_1^{(n)} \text{and } \Phi_0^{(n+1)} = 0.4 + 0.6\Phi_0^{(n)}.$$
\end{corollary}
So $\lim_{n \rightarrow \infty} \Phi_1^{(n)} = 0$ and $\lim_{n \rightarrow \infty} \Phi_0^{(n)} = 1$. The random walk is recurrent.

\subsection{Classical diffusion and recurrence exponents $d_{w, c}$ and $r_{w, c}$} 
Let $T^{(n)} = \inf \{ t\geq 0; w_t \in \partial (F^{(n)}\cup F^{(n)'}) \}, \partial (F^{(n)}\cup F^{(n)'}) := \{\mathbf{a_n}, \mathbf{b_n},\mathbf{a'_n}, \mathbf{b'_n}\}$ be the first-passage time taken to exit $F^{(n)}\cup F^{(n)'}$ at the four vertices.
\begin{theorem} \label{th:classical exponents}
For classical random walks on the Sierpinski gaskets, we have 

(a) $E(T^{(n+1)}|w_0 = \mathbf{0}) = 5E(T^{(n)}|w_0 = \mathbf{0})$.

(b) $d_{w, c} = \frac{\ln5}{\ln2}$.
\end{theorem}

Let $\tau^{(n)} = \inf \{ t\geq 1; w_t \in \partial (F^{(n)})\cup \partial (F^{(n)'}) \}, \partial (F^{(n)})\cup \partial (F^{(n)'}) := \{\mathbf{0}, \mathbf{a_n}, \mathbf{b_n},\mathbf{a'_n}, \mathbf{b'_n}\}$ be the exit time taken to exit $F^{(n)}\cup  F^{(n)'}$ at the four vertices or back to the origion.
\begin{theorem} \label{th:recurrence exponent classical case}
For classical random walks on the Sierpinski gaskets, we have 

(a) $E(\tau^{(n+1)}|w_0 = \mathbf{0}) = 3E(\tau^{(n)}|w_0 = \mathbf{0})$.

(b) $r_{w, c}=\frac{\ln3}{\ln2}$.

\end{theorem}

\subsection{Uniform coin operator, classical case}
For classical random walks, we can also compute the expected hitting times and exit probability  distributions using the same methods as we used for quantum random walks. When the coin operator is uniform, then $U$ is no longer a unitary operator and in this case, the quantum random walk becomes the classical random walk. In doing so, we not only obtain known results but also   new results such as the detail probability distributions as for  each of the departing and arriving directions.

Let
$$G = r
  \left[
 \begin{array}{ c c c c }
1 & 1 & 1 & 1 \\
1 & 1 & 1 & 1 \\
1 & 1 & 1 & 1 \\
1 & 1 & 1 & 1 
\end{array} \right],$$ where $\displaystyle r =1/4$. In this case,  the formulas for the amplitude function for the quantum case becomes the formulas for the transition probability for the classical case.  

In this case,  we only have three free variables for each iteration instead of   six free variables in the quantum case, where
\begin{equation}
{g^{(n)}(\mathbf{0},\mathbf{a_n})} = \bordermatrix{~& 0 & 1 & 4              & 5           \cr
                                                       0 & 0 & 0 & u^{(n)}_1 & u^{(n)}_2       \cr
                                                       1 & 0 & 0 & u^{(n)}_1 & u^{(n)}_2    \cr
                                                       4 & 0 & 0 & u^{(n)}_1& u^{(n)}_2    \cr
                                                       5 & 0 & 0 & u^{(n)}_1 & u^{(n)}_2        \cr},
\end{equation}
\begin{equation}
{g^{(n)}(\mathbf{0},\mathbf{0})} = \bordermatrix{~& 0 & 1 & 4              & 5           \cr
                                                       0 & 0 & 0 & u^{(n)}_3 & u^{(n)}_3       \cr
                                                       1 & 0 & 0 & u^{(n)}_3 & u^{(n)}_3    \cr
                                                       4 & 0 & 0 & u^{(n)}_3& u^{(n)}_3   \cr
                                                       5 & 0 & 0 & u^{(n)}_3 & u^{(n)}_3        \cr}.
\end{equation}
When n = 1, we have
$${g^{(1)}(\mathbf{0},\mathbf{a_1})} = \bordermatrix{~& 0 & 1 & 4              & 5           \cr
                                                       0 & 0 & 0 & -\frac{z^2}{2(z^2+2z-8)} & -\frac{z^3}{4(z^2+2z-8)}     \cr
                                                       1 & 0 & 0 & -\frac{z^2}{2(z^2+2z-8)} & -\frac{z^3}{4(z^2+2z-8)}   \cr
                                                       4 & 0 & 0 & -\frac{z^2}{2(z^2+2z-8)} & -\frac{z^3}{4(z^2+2z-8)}  \cr
                                                       5 & 0 & 0 & -\frac{z^2}{2(z^2+2z-8)} & -\frac{z^3}{4(z^2+2z-8)}     \cr},$$
$$
{g^{(1)}(\mathbf{0},\mathbf{0})} = \bordermatrix{~& 0 & 1 & 4              & 5           \cr
                                                       0 & -\frac{z^2}{2(z^2+2z-8)} & -\frac{z^2}{2(z^2+2z-8)} & -\frac{z^2}{2(z^2+2z-8)} & -\frac{z^2}{2(z^2+2z-8)}     \cr
                                                       1 & -\frac{z^2}{2(z^2+2z-8)} & -\frac{z^2}{2(z^2+2z-8)} &  -\frac{z^2}{2(z^2+2z-8)} & -\frac{z^2}{2(z^2+2z-8)} \cr
                                                       4 & -\frac{z^2}{2(z^2+2z-8)} & -\frac{z^2}{2(z^2+2z-8)} &  -\frac{z^2}{2(z^2+2z-8)} & -\frac{z^2}{2(z^2+2z-8)} \cr
                                                       5 & -\frac{z^2}{2(z^2+2z-8)} & -\frac{z^2}{2(z^2+2z-8)} &  -\frac{z^2}{2(z^2+2z-8)} & -\frac{z^2}{2(z^2+2z-8)}     \cr}.
$$
So we have 
\begin{eqnarray*}
u^{(1)}_1&=&-\frac{z^2}{2(z^2+2z-8)},\\
u^{(0)}_2&=&-\frac{z^3}{4(z^2+2z-8)},\\
u^{(0)}_3&=&-\frac{z^2}{2(z^2+2z-8)}.
\end{eqnarray*}
Letting  $z=1$, we have thus obtained the transition probabilities
$$
{P_c^{(1)}(\mathbf{0},\mathbf{a_1})} = \bordermatrix{~& 0 & 1 & 4              & 5           \cr
                                                       0 & 0 & 0 & 0.1 & 0.05      \cr
                                                       1 & 0 & 0 & 0.1 & 0.05   \cr
                                                       4 & 0 & 0 & 0.1 & 0.05  \cr
                                                       5 & 0 & 0 & 0.1 & 0.05     \cr}
,$$
$$
{P_c^{(1)}(\mathbf{0},\mathbf{0})} = \bordermatrix{~& 0 & 1 & 4              & 5           \cr
                                                       0 & 0.1 & 0.1 & 0.1 & 0.1     \cr
                                                       1 & 0.1 & 0.1 & 0.1 & 0.1 \cr
                                                       4 & 0.1 & 0.1 & 0.1 & 0.1 \cr
                                                       5 & 0.1 & 0.1 & 0.1 & 0.1     \cr}
,$$
$${P_c^{(2)}(\mathbf{0},\mathbf{a_1})} = \bordermatrix{~& 0 & 1 & 4              & 5           \cr
                                                       0 & 0 & 0 & 0.05 & 0.04     \cr
                                                       1 & 0 & 0 & 0.05 & 0.04    \cr
                                                       4 & 0 & 0 & 0.05 & 0.04   \cr
                                                       5 & 0 & 0 & 0.05 & 0.04     \cr}
,$$
 $$
{P_c^{(2)}(\mathbf{0},\mathbf{0})} = \bordermatrix{~& 0 & 1 & 4              & 5           \cr
                                                       0 & 0.16 & 0.16 & 0.16 & 0.16     \cr
                                                       1 & 0.16 & 0.16 & 0.16 & 0.16  \cr
                                                       4 & 0.16 & 0.16 & 0.16 & 0.16  \cr
                                                       5 & 0.16 & 0.16 & 0.16 & 0.16      \cr}
,$$
$${P_c^{(3)}(\mathbf{0},\mathbf{a_1})} = \bordermatrix{~& 0 & 1 & 4              & 5           \cr
                                                       0 & 0 & 0 & 0.028 & 0.026     \cr
                                                       1 & 0 & 0 & 0.028 & 0.026    \cr
                                                       4 & 0 & 0 & 0.028 & 0.026   \cr
                                                       5 & 0 & 0 & 0.028 & 0.026     \cr}
,$$
$$
{P_c^{(3)}(\mathbf{0},\mathbf{0})} = \bordermatrix{~& 0 & 1 & 4              & 5           \cr
                                                       0 & 0.196 & 0.196 & 0.196 & 0.196     \cr
                                                       1 & 0.196 & 0.196 & 0.196 & 0.196  \cr
                                                       4 & 0.196 & 0.196 & 0.196 & 0.196  \cr
                                                       5 & 0.196 & 0.196 & 0.196 & 0.196      \cr}
,$$
here subscript $c$ is denoted for classical case. 
Summing each row of each of the above matrices, we obtain the probabilities  that agree with  the known results for the classical case \cite{2005 diffusion}. Here we have obtained more detail information about the transition probabilities as for  each of the departing and arriving directions. 

In general, we have a map $T:\mathbb{R}^3 \rightarrow \mathbb{R}^3$, 
$$(u^{(n+1)}_1,u^{(n+1)}_2,u^{(n+1)}_3) = T(u^{(n)}_1,u^{(n)}_2,u^{(n)}_3),$$
where 
\begin{eqnarray*}\label{eq:classical hitting prob n}
&&u^{(n+1)}_1 = \\
&& -(u^{(n)}_1 + u^{(n)}_2)(2(u^{(n)}_2)^2 + 2u^{(n)}_1u^{(n)}_2 + u^{(n)}_1 - 4u^{(n)}_1u^{(n)}_3)\\
&&  (2(u^{(n)}_1)^2 + 4u^{(n)}_1u^{(n)}_2 - 4u^{(n)}_1u^{(n)}_3 + u^{(n)}_1 + 2(u^{(n)}_1)^2 - 4u^{(n)}_2u^{(n)}_3 + u^{(n)}_2 - 16(u^{(n)}_3)^2 + 8u^{(n)}_3 - 1)^{(-1)} ,\\
&&u^{(n+1)}_2 = \\
&&-(u^{(n)}_1 + u^{(n)}_2)(2(u^{(n)}_1)^2 + 2u^{(n)}_1u^{(n)}_2 + u^{(n)}_2 - 4u^{(n)}_2u^{(n)}_3)\\
&& (2(u^{(n)}_1)^2 + 4u^{(n)}_1u^{(n)}_2 - 4u^{(n)}_1u^{(n)}_3 + u^{(n)}_1 + 2(u^{(n)}_1)^2 - 4u^{(n)}_2u^{(n)}_3 + u^{(n)}_2 - 16(u^{(n)}_3)^2 + 8u^{(n)}_3 - 1)^{(-1)} ,\\
&&u^{(n+1)}_3 = \\
&& [6(u^{(n)}_1)^2u^{(n)}_3 -(u^{(n)}_1)^2 + 12u^{(n)}_1u^{(n)}_2u^{(n)}_3 - 2u^{(n)}_1u^{(n)}_2 - 4u^{(n)}_1(u^{(n)}_3)^2 + u^{(n)}_1u^{(n)}_3 \\
 &&+6(u^{(n)}_2)^2u^{(n)}_3 - (u^{(n)}_2)^2 - 4u^{(n)}_2(u^{(n)}_3)^2
+u^{(n)}_2u^{(n)}_3 - 16(u^{(n)}_3)^3 + 8(u^{(n)}_3)^2 - u^{(n)}_3]\\
&&[(2(u^{(n)}_1)^2 + 4u^{(n)}_1u^{(n)}_2 - 4u^{(n)}_1u^{(n)}_3 + u^{(n)}_1 
 +2(u^{(n)}_1)^2 - 4u^{(n)}_2u^{(n)}_3 + u^{(n)}_2 - 16(u^{(n)}_3)^2 + 8u^{(n)}_3 - 1)]^{(-1)} .
\end{eqnarray*}
Note that for classical case, the above $u$'s at $z=1$ are the probabilities (instead of amplitude functions for quantum case). By plugging $z=1$, we have thus obtained

\begin{theorem}\label{th:classical hitting probabilities theorem}
For classical random walks on the Sierpinski gaskets, if the initial probability distribution 
$u^{(0)}_1|_{z=1} = \frac{1}{4},u^{(0)}_2|_{z=1} = 0, u^{(0)}_3|_{z=1} = 0$, then we have
$$u^{(n)}_1|_{z=1} +u^{(n)}_2|_{z=1} +u^{(n)}_3|_{z=1}  = \frac{1}{4}.$$ Moreover 
$$u^{(n+1)}_1|_{z=1}  = 0.4u^{(n)}_1|_{z=1} +0.2u^{(n)}_2|_{z=1} ,u^{(n+1)}_2|_{z=1}  = 0.2u^{(n)}_1|_{z=1} +0.4u^{(n)}_2|_{z=1},$$ 
$$u^{(n+1)}_3|_{z=1}  = 0.1+0.6u^{(n)}_3|_{z=1} .$$
\end{theorem}
\begin{corollary}
$\lim_{n \rightarrow \infty} u^{(n)}_1|_{z=1} = u^{(n)}_2|_{z=1} = 0$ and $\lim_{n \rightarrow 
\infty} u^{(n)}_3|_{z=1} = \frac{1}{4}$.
\end{corollary}
The limiting behavior of classical hitting probability and quantum hitting probability are depicted in Fig. \ref{fig:classical hitting probability Sierpinski} and Fig. \ref{fig:quantum hitting probability Sierpinski}.

\begin{figure}[htbp]
\centering
\includegraphics[trim={0 8cm 0 7cm},clip, width=0.6\textwidth]{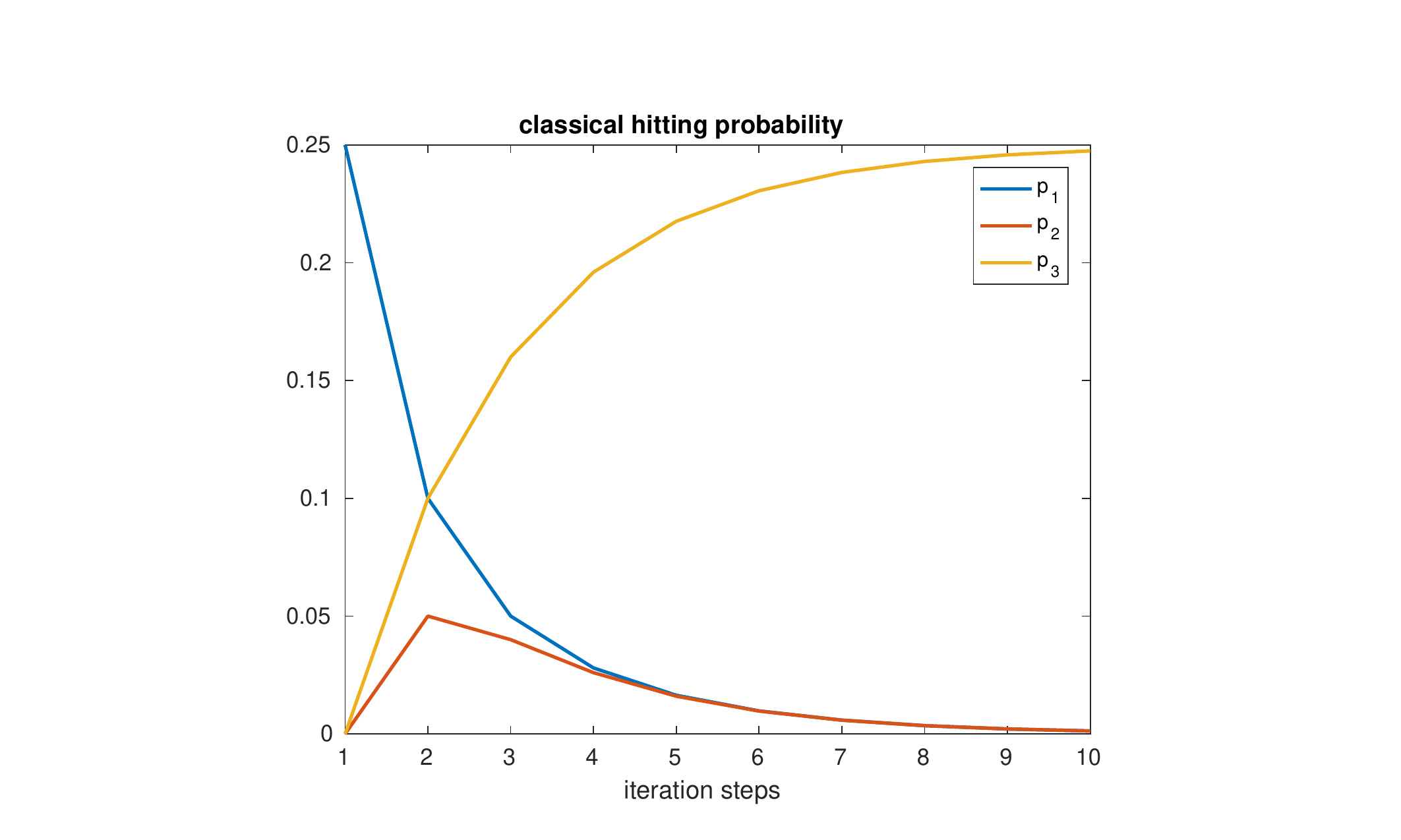}
\caption{Classical hitting probability on Sierpinski gasket}\label{fig:classical hitting probability Sierpinski}
\includegraphics[trim={0 8cm 0 7cm},clip, width=0.6\textwidth]{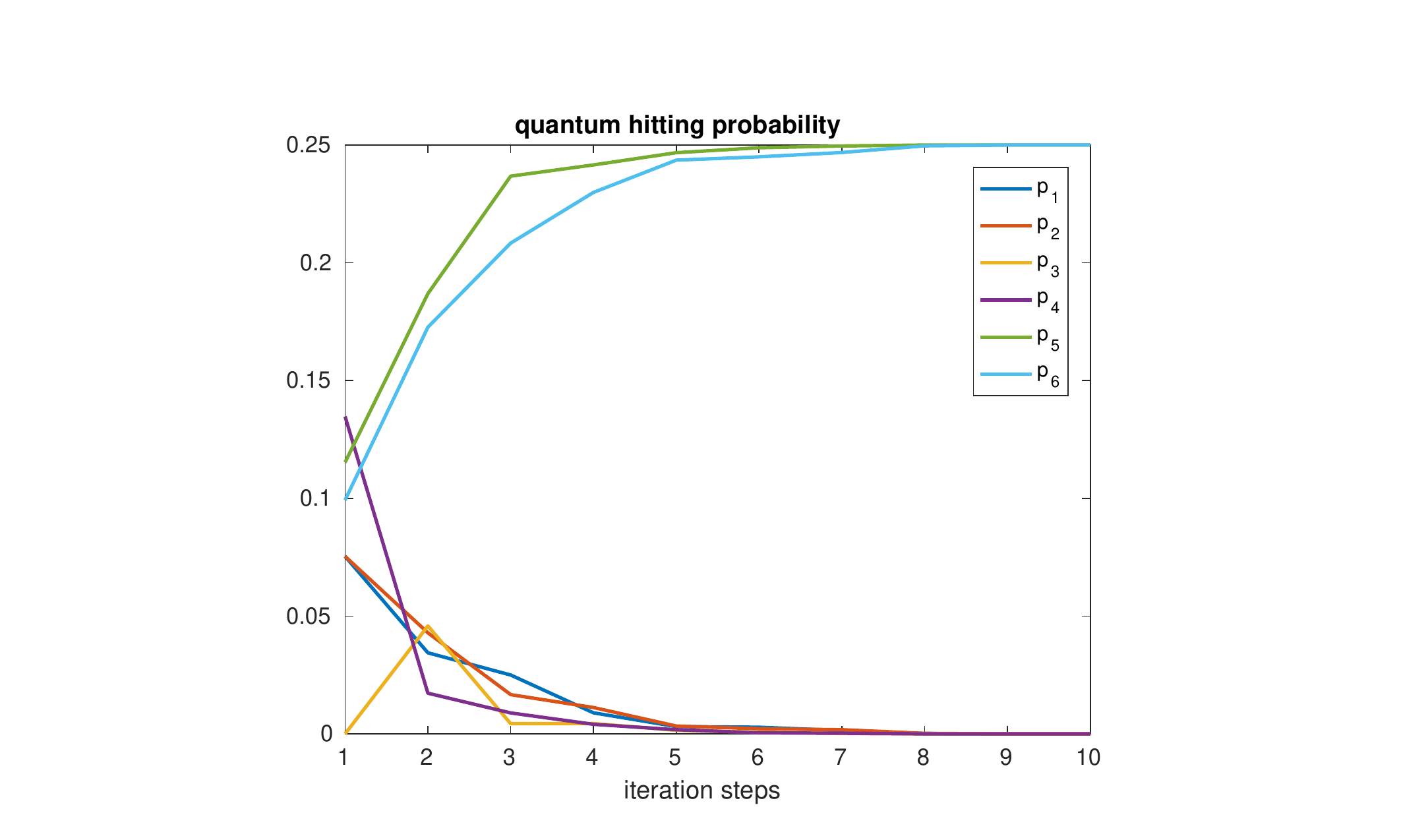}
\caption{Quantum hitting probability on Sierpinski gasket}\label{fig:quantum hitting probability Sierpinski}
\end{figure}

We note that the hitting distributions for classical random walks on Sierpinski gaskets have been  known for the total probabilities of all directions \cite{2005 diffusion}. The above theorem and corollary give  the detail hitting probability with respect to each of the arriving and departing directions.  For examples, $n=1$ and $n=2$, the detail of hitting probabilities distribution are depicted  in Fig. \ref{Green function and hitting probability on F(1)  F^(1)'}.
\begin{figure}[htb]
\begin{minipage}{0.45\textwidth}
\begin{tikzpicture}[scale = 0.7]
\begin{axis}[dashed, xmin=-3,ymin=-3,xticklabels = {,,}, yticklabels = {},xmax=5,ymax=3,]
\addplot[mark = none] coordinates{(-5,0) (5,0)};
\addplot[mark = none] coordinates{(-5,1) (5,1)};
\addplot[mark = none] coordinates{(-5,2) (5,2)};
\addplot[mark = none] coordinates{(-5,-1) (5,-1)};
\addplot[mark = none] coordinates{(-5,-2) (5,-2)};
\addplot[mark = none] coordinates{(0,-3) (0,3)};
\addplot[mark = none] coordinates{(1,-3) (1,3)};
\addplot[mark = none] coordinates{(2,-3) (2,3)};
\addplot[mark = none] coordinates{(3,-3) (3,3)};
\addplot[mark = none] coordinates{(4,-3) (4,3)};
\addplot[mark = none] coordinates{(-1,-3) (-1,3)};
\addplot[mark = none] coordinates{(-2,-3) (-2,3)};
\addplot[mark = none] coordinates{(-3,-3) (-3,3)};
\addplot[mark = none] coordinates{(-4,-3) (-4,3)};
\addplot[mark = *, solid, thick] coordinates {(0, 0) (2,0) (4,0) (3,1) (2,2) (1,1) (0,0)};
\addplot[mark = *, solid ,thick] coordinates {(1,1) (3,1) (2,0) (1,1)};
\addplot[mark = *, solid, thick] coordinates {(0, 0) (-1,-1) (-2,-2) (0,-2) (2,-2) (1,-1) (0,0)};
\addplot[mark = *, solid ,thick] coordinates {(-1,-1) (0,-2) (1,-1) (-1,-1)};
\addplot[mark=*] coordinates {(0,0)} node(1)[]{} ;
\addplot[mark=] coordinates {(-0.5,0)}[xshift=0pt,yshift=-10pt] node[label={$\mathbf{0}$}]{} ;
\addplot[mark=*] coordinates {(1,1)} node(2)[]{} ;
\addplot[mark=*] coordinates {(2,0)} node(3)[xshift=0pt,yshift=-20pt][]{} ;
\addplot[mark=*] coordinates {(3,1)} node(4)[]{} ;
\addplot[mark=*] coordinates {(2,2)} node(5)[label={$\mathbf{a_1}$}]{} ;
\addplot[mark=*] coordinates {(4,0)} node(6)[]{} ;
\addplot[mark=] coordinates {(4.5,0)} node[xshift=0pt,yshift=-10pt][label={$\mathbf{b_1}$}]{} ;
\addplot[mark=*] coordinates {(-1,-1)} node(-2)[]{} ;
\addplot[mark=*] coordinates {(1,-1)} node(-3)[]{} ;
\addplot[mark=*] coordinates {(0,-2)} node(-4)[]{} ;
\addplot[mark=*] coordinates {(-2,-2)} node(-5)[]{} ;
\addplot[mark=] coordinates {(-2.5,-2)} node[xshift=2pt,yshift=-12pt][label={$\mathbf{a'_1}$}]{} ;
\addplot[mark=*] coordinates {(2,-2)} node(-6)[]{} ;
\addplot[mark=] coordinates {(2.5,-2)} node[xshift=0pt,yshift=-12pt][label={$\mathbf{b'_1}$}]{} ;
\addplot[mark=none] coordinates {(-1,-1)} node(-1)[]{} ;
\addplot[mark=none] coordinates {(1,-1)} node(-2)[]{} ;
\draw [->,thick ,solid] (1) -- (5)node[xshift=-13pt,yshift=-5pt] {$u^{(1)}_1$};
\draw [->,thick ,solid] (6) -- (5)node[xshift=17pt,yshift=-5pt]{$u^{(1)}_2$};
\draw [->,thick ,solid] (1) -- (6)node[xshift=-12pt,yshift=-9pt] {$u^{(1)}_1$};
\draw [->,thick ,solid] (5) -- (6)node[xshift=1pt,yshift=15pt]{$u^{(1)}_2$};
\draw [->,thick ,solid] (5) -- (1)node[xshift=2pt,yshift=16pt] {$u^{(1)}_3$};
\draw [->,thick ,solid] (6) -- (1)node[xshift=20pt,yshift=8pt]{$u^{(1)}_3$};

\draw [->,thick ,solid] (1) -- (-5)node[xshift=6pt,yshift=18pt] {$u^{(1)}_1$};
\draw [->,thick ,solid] (-6) -- (-5)node[xshift=15pt,yshift=-8pt]{$u^{(1)}_2$};
\draw [->,thick ,solid] (1) -- (-6)node[xshift=-2pt,yshift=18pt] {$u^{(1)}_1$};
\draw [->,thick ,solid] (-5) -- (-6)node[xshift=-15pt,yshift=-8pt]{$u^{(1)}_2$};
\draw [->,thick ,solid] (-5) -- (1)node[xshift=-2pt,yshift=-13pt] {$u^{(1)}_3$};
\draw [->,thick ,solid] (-6) -- (1)node[xshift=18pt,yshift=-8pt]{$u^{(1)}_3$};
\end{axis}        
\end{tikzpicture}
\end{minipage}
\hfill
\begin{minipage}{0.45\textwidth}
\begin{tikzpicture}[scale = 0.7]
\begin{axis}[dashed, xmin=-3,ymin=-3,xticklabels = {,,}, yticklabels = {},xmax=5,ymax=3,]
\addplot[mark = none] coordinates{(-5,0) (5,0)};
\addplot[mark = none] coordinates{(-5,1) (5,1)};
\addplot[mark = none] coordinates{(-5,2) (5,2)};
\addplot[mark = none] coordinates{(-5,-1) (5,-1)};
\addplot[mark = none] coordinates{(-5,-2) (5,-2)};
\addplot[mark = none] coordinates{(0,-3) (0,3)};
\addplot[mark = none] coordinates{(1,-3) (1,3)};
\addplot[mark = none] coordinates{(2,-3) (2,3)};
\addplot[mark = none] coordinates{(3,-3) (3,3)};
\addplot[mark = none] coordinates{(4,-3) (4,3)};
\addplot[mark = none] coordinates{(-1,-3) (-1,3)};
\addplot[mark = none] coordinates{(-2,-3) (-2,3)};
\addplot[mark = none] coordinates{(-3,-3) (-3,3)};
\addplot[mark = none] coordinates{(-4,-3) (-4,3)};
\addplot[mark = *, solid, thick] coordinates {(0, 0) (2,0) (4,0) (3,1) (2,2) (1,1) (0,0)};
\addplot[mark = *, solid ,thick] coordinates {(1,1) (3,1) (2,0) (1,1)};
\addplot[mark = *, solid, thick] coordinates {(0, 0) (-1,-1) (-2,-2) (0,-2) (2,-2) (1,-1) (0,0)};
\addplot[mark = *, solid ,thick] coordinates {(-1,-1) (0,-2) (1,-1) (-1,-1)};
\addplot[mark=*] coordinates {(0,0)} node(1)[]{} ;
\addplot[mark=] coordinates {(-0.5,0)}[xshift=0pt,yshift=-10pt] node[label={$\mathbf{0}$}]{} ;
\addplot[mark=*] coordinates {(1,1)} node(2)[]{} ;
\addplot[mark=*] coordinates {(2,0)} node(3)[xshift=0pt,yshift=-20pt][]{} ;
\addplot[mark=*] coordinates {(3,1)} node(4)[]{} ;
\addplot[mark=*] coordinates {(2,2)} node(5)[label={$\mathbf{a_1}$}]{} ;
\addplot[mark=*] coordinates {(4,0)} node(6)[]{} ;
\addplot[mark=] coordinates {(4.5,0)} node[xshift=0pt,yshift=-10pt][label={$\mathbf{b_1}$}]{} ;
\addplot[mark=*] coordinates {(-1,-1)} node(-2)[]{} ;
\addplot[mark=*] coordinates {(1,-1)} node(-3)[]{} ;
\addplot[mark=*] coordinates {(0,-2)} node(-4)[]{} ;
\addplot[mark=*] coordinates {(-2,-2)} node(-5)[]{} ;
\addplot[mark=] coordinates {(-2.5,-2)} node[xshift=2pt,yshift=-12pt][label={$\mathbf{a'_1}$}]{} ;
\addplot[mark=*] coordinates {(2,-2)} node(-6)[]{} ;
\addplot[mark=] coordinates {(2.5,-2)} node[xshift=0pt,yshift=-12pt][label={$\mathbf{b'_1}$}]{} ;
\addplot[mark=none] coordinates {(-1,-1)} node(-1)[]{} ;
\addplot[mark=none] coordinates {(1,-1)} node(-2)[]{} ;
\draw [->,thick ,solid] (1) -- (5)node[xshift=-13pt,yshift=-5pt] {$0.1$};
\draw [->,thick ,solid] (6) -- (5)node[xshift=14pt,yshift=-5pt]{$0.05$};
\draw [->,thick ,solid] (1) -- (6)node[xshift=-12pt,yshift=-9pt] {$0.1$};
\draw [->,thick ,solid] (5) -- (6)node[xshift=-3pt,yshift=15pt]{$0.05$};
\draw [->,thick ,solid] (5) -- (1)node[xshift=2pt,yshift=16pt] {$0.1$};
\draw [->,thick ,solid] (6) -- (1)node[xshift=20pt,yshift=8pt]{$0.1$};

\draw [->,thick ,solid] (1) -- (-5)node[xshift=6pt,yshift=18pt] {$0.1$};
\draw [->,thick ,solid] (-6) -- (-5)node[xshift=15pt,yshift=-8pt]{$0.05$};
\draw [->,thick ,solid] (1) -- (-6)node[xshift=-2pt,yshift=18pt] {$0.1$};
\draw [->,thick ,solid] (-5) -- (-6)node[xshift=-15pt,yshift=-8pt]{$0.05$};
\draw [->,thick ,solid] (-5) -- (1)node[xshift=-2pt,yshift=-13pt] {$0.1$};
\draw [->,thick ,solid] (-6) -- (1)node[xshift=18pt,yshift=-8pt]{$0.1$};
\end{axis}        
\end{tikzpicture}
\end{minipage}
\caption{Classical hitting probability on $F^{(1)} \cup F^{(1)'}$}\label{Green function and hitting probability on F(1)  F^(1)'}
\end{figure}
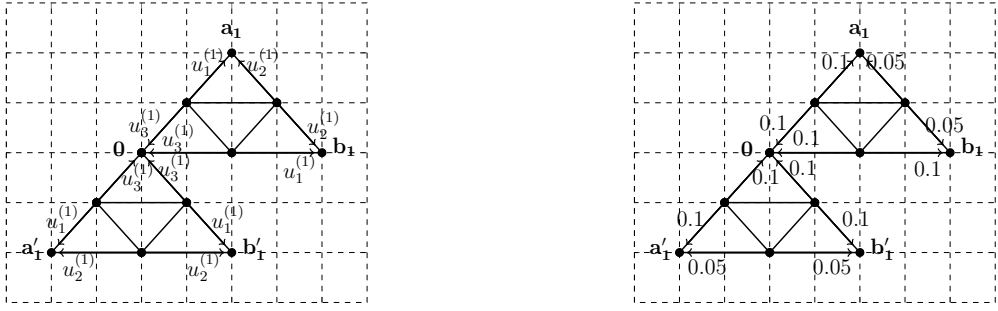

Also for the uniform coin operator, the expectation becomes
\begin{scriptsize}
\begin{align*}
E(T^{(n)}) &= \sum_{\mathbf{y} \in \partial (F^{(n)}\cup F^{(n)'})}\sum_{\mathbf{e}_j \in out(\mathbf{y}) } \sum_{t=1}^{\infty}t P_1^{(n)}(\mathbf{0},\mathbf{y},T = t)^i_j \\
		   &= \sum_{\mathbf{y} \in \partial (F^{(n)}\cup F^{(n)'})}\sum_{\mathbf{e}_j \in out(\mathbf{y}) } \sum_{t=1}^{\infty} t \Psi (w_0 = \mathbf{0}^i , w_{t} = \mathbf{y}^j ,T = t) \\
		   & = \sum_{\mathbf{y} \in \partial (F^{(n)}\cup F^{(n)'})}\sum_{\mathbf{e}_j \in out(\mathbf{y}) }  (\partial_z g_1^{(n)}(z)(\mathbf{0},\mathbf{y})_{j}^{i}|_{z=1})
\end{align*}
\end{scriptsize}

Unlike quantum random walk, if we sum the passage probabilities, then $P(T^{(1)} < \infty) =  1$. 
By taking derivative of $g_1^{(1)}(z)(\mathbf{0},\mathbf{5})$ and plugging $z=1$, we have $E(T^{(1)}) = 5$. Therefore we have obtained the passage probability and expected passage time from each direction, see Fig. \ref{passage probability on F(1)  F^(1)'} and Fig. \ref{classical $E(T^{(1)})$ on F(1)  F^(1)'}.
                                                                  
\begin{figure}
\begin{minipage}{0.45\textwidth}
\begin{tikzpicture}[scale = 0.7]
\begin{axis}[dashed, xmin=-3,ymin=-3,xticklabels = {,,}, yticklabels = {},xmax=5,ymax=3,]
\addplot[mark = none] coordinates{(-5,0) (5,0)};
\addplot[mark = none] coordinates{(-5,1) (5,1)};
\addplot[mark = none] coordinates{(-5,2) (5,2)};
\addplot[mark = none] coordinates{(-5,-1) (5,-1)};
\addplot[mark = none] coordinates{(-5,-2) (5,-2)};
\addplot[mark = none] coordinates{(0,-3) (0,3)};
\addplot[mark = none] coordinates{(1,-3) (1,3)};
\addplot[mark = none] coordinates{(2,-3) (2,3)};
\addplot[mark = none] coordinates{(3,-3) (3,3)};
\addplot[mark = none] coordinates{(4,-3) (4,3)};
\addplot[mark = none] coordinates{(-1,-3) (-1,3)};
\addplot[mark = none] coordinates{(-2,-3) (-2,3)};
\addplot[mark = none] coordinates{(-3,-3) (-3,3)};
\addplot[mark = none] coordinates{(-4,-3) (-4,3)};
\addplot[mark = *, solid, thick] coordinates {(0, 0) (2,0) (4,0) (3,1) (2,2) (1,1) (0,0)};
\addplot[mark = *, solid ,thick] coordinates {(1,1) (3,1) (2,0) (1,1)};
\addplot[mark = *, solid, thick] coordinates {(0, 0) (-1,-1) (-2,-2) (0,-2) (2,-2) (1,-1) (0,0)};
\addplot[mark = *, solid ,thick] coordinates {(-1,-1) (0,-2) (1,-1) (-1,-1)};
\addplot[mark=*] coordinates {(0,0)} node(1)[]{} ;
\addplot[mark=] coordinates {(-0.5,0)}[xshift=0pt,yshift=-10pt] node[label={$\mathbf{0}$}]{} ;
\addplot[mark=*] coordinates {(1,1)} node(2)[]{} ;
\addplot[mark=*] coordinates {(2,0)} node(3)[xshift=0pt,yshift=-20pt][]{} ;
\addplot[mark=*] coordinates {(3,1)} node(4)[]{} ;
\addplot[mark=*] coordinates {(2,2)} node(5)[label={$\mathbf{a_1}$}]{} ;
\addplot[mark=*] coordinates {(4,0)} node(6)[]{} ;
\addplot[mark=] coordinates {(4.5,0)} node[xshift=0pt,yshift=-10pt][label={$\mathbf{b_1}$}]{} ;
\addplot[mark=*] coordinates {(-1,-1)} node(-2)[]{} ;
\addplot[mark=*] coordinates {(1,-1)} node(-3)[]{} ;
\addplot[mark=*] coordinates {(0,-2)} node(-4)[]{} ;
\addplot[mark=*] coordinates {(-2,-2)} node(-5)[]{} ;
\addplot[mark=] coordinates {(-2.5,-2)} node[xshift=2pt,yshift=-12pt][label={$\mathbf{a'_1}$}]{} ;
\addplot[mark=*] coordinates {(2,-2)} node(-6)[]{} ;
\addplot[mark=] coordinates {(2.5,-2)} node[xshift=0pt,yshift=-12pt][label={$\mathbf{b'_1}$}]{} ;
\addplot[mark=none] coordinates {(-1,-1)} node(-1)[]{} ;
\addplot[mark=none] coordinates {(1,-1)} node(-2)[]{} ;
\draw [->,thick ,solid] (1) -- (5)node[xshift=-17pt,yshift=-5pt] {$\frac{1}{6}$};
\draw [->,thick ,solid] (6) -- (5)node[xshift=17pt,yshift=-5pt]{$\frac{1}{12}$};
\draw [->,thick ,solid] (1) -- (6)node[xshift=-12pt,yshift=-9pt] {$\frac{1}{6}$};
\draw [->,thick ,solid] (5) -- (6)node[xshift=0pt,yshift=17pt]{$\frac{1}{12}$};

\draw [->,thick ,solid] (1) -- (-5)node[xshift=2pt,yshift=18pt] {$\frac{1}{6}$};
\draw [->,thick ,solid] (-6) -- (-5)node[xshift=15pt,yshift=-8pt]{$\frac{1}{12}$};
\draw [->,thick ,solid] (1) -- (-6)node[xshift=-2pt,yshift=18pt] {$\frac{1}{6}$};
\draw [->,thick ,solid] (-5) -- (-6)node[xshift=-15pt,yshift=-8pt]{$\frac{1}{12}$};

\end{axis}        
\end{tikzpicture}
\caption{passage hitting probability distribution on $F^{(1)} \cup F^{(1)'}$ with $w_0 = \mathbf{0}^0$}  \label{passage probability on F(1)  F^(1)'}
\end{minipage}
\hfill
\begin{minipage}{0.45\textwidth}
\begin{tikzpicture}[scale = 0.7]
\begin{axis}[dashed, xmin=-3,ymin=-3,xticklabels = {,,}, yticklabels = {},xmax=5,ymax=3,]
\addplot[mark = none] coordinates{(-5,0) (5,0)};
\addplot[mark = none] coordinates{(-5,1) (5,1)};
\addplot[mark = none] coordinates{(-5,2) (5,2)};
\addplot[mark = none] coordinates{(-5,-1) (5,-1)};
\addplot[mark = none] coordinates{(-5,-2) (5,-2)};
\addplot[mark = none] coordinates{(0,-3) (0,3)};
\addplot[mark = none] coordinates{(1,-3) (1,3)};
\addplot[mark = none] coordinates{(2,-3) (2,3)};
\addplot[mark = none] coordinates{(3,-3) (3,3)};
\addplot[mark = none] coordinates{(4,-3) (4,3)};
\addplot[mark = none] coordinates{(-1,-3) (-1,3)};
\addplot[mark = none] coordinates{(-2,-3) (-2,3)};
\addplot[mark = none] coordinates{(-3,-3) (-3,3)};
\addplot[mark = none] coordinates{(-4,-3) (-4,3)};
\addplot[mark = *, solid, thick] coordinates {(0, 0) (2,0) (4,0) (3,1) (2,2) (1,1) (0,0)};
\addplot[mark = *, solid ,thick] coordinates {(1,1) (3,1) (2,0) (1,1)};
\addplot[mark = *, solid, thick] coordinates {(0, 0) (-1,-1) (-2,-2) (0,-2) (2,-2) (1,-1) (0,0)};
\addplot[mark = *, solid ,thick] coordinates {(-1,-1) (0,-2) (1,-1) (-1,-1)};
\addplot[mark=*] coordinates {(0,0)} node(1)[]{} ;
\addplot[mark=] coordinates {(-0.5,0)}[xshift=0pt,yshift=-10pt] node[label={$\mathbf{0}$}]{} ;
\addplot[mark=*] coordinates {(1,1)} node(2)[]{} ;
\addplot[mark=*] coordinates {(2,0)} node(3)[xshift=0pt,yshift=-20pt][]{} ;
\addplot[mark=*] coordinates {(3,1)} node(4)[]{} ;
\addplot[mark=*] coordinates {(2,2)} node(5)[label={$\mathbf{a_1}$}]{} ;
\addplot[mark=*] coordinates {(4,0)} node(6)[]{} ;
\addplot[mark=] coordinates {(4.5,0)} node[xshift=0pt,yshift=-10pt][label={$\mathbf{b_1}$}]{} ;
\addplot[mark=*] coordinates {(-1,-1)} node(-2)[]{} ;
\addplot[mark=*] coordinates {(1,-1)} node(-3)[]{} ;
\addplot[mark=*] coordinates {(0,-2)} node(-4)[]{} ;
\addplot[mark=*] coordinates {(-2,-2)} node(-5)[]{} ;
\addplot[mark=] coordinates {(-2.5,-2)} node[xshift=2pt,yshift=-12pt][label={$\mathbf{a'_1}$}]{} ;
\addplot[mark=*] coordinates {(2,-2)} node(-6)[]{} ;
\addplot[mark=] coordinates {(2.5,-2)} node[xshift=0pt,yshift=-12pt][label={$\mathbf{b'_1}$}]{} ;
\addplot[mark=none] coordinates {(-1,-1)} node(-1)[]{} ;
\addplot[mark=none] coordinates {(1,-1)} node(-2)[]{} ;
\draw [->,thick ,solid] (1) -- (5)node[xshift=-17pt,yshift=-5pt] {$\frac{7}{9}$};
\draw [->,thick ,solid] (6) -- (5)node[xshift=17pt,yshift=-5pt]{$\frac{17}{36}$};
\draw [->,thick ,solid] (1) -- (6)node[xshift=-12pt,yshift=-9pt] {$\frac{7}{9}$};
\draw [->,thick ,solid] (5) -- (6)node[xshift=0pt,yshift=17pt]{$\frac{17}{36}$};

\draw [->,thick ,solid] (1) -- (-5)node[xshift=2pt,yshift=18pt] {$\frac{7}{9}$};
\draw [->,thick ,solid] (-6) -- (-5)node[xshift=15pt,yshift=-8pt]{$\frac{17}{36}$};
\draw [->,thick ,solid] (1) -- (-6)node[xshift=-2pt,yshift=18pt] {$\frac{7}{9}$};
\draw [->,thick ,solid] (-5) -- (-6)node[xshift=-15pt,yshift=-8pt]{$\frac{17}{36}$};

\end{axis}        
\end{tikzpicture}
\caption{$E(T^{(1)})$ on $F^{(1)} \cup F^{(1)'}$ with $w_0 = \mathbf{0}^0$}   \label{classical $E(T^{(1)})$ on F(1)  F^(1)'}
\end{minipage}
\end{figure}

Using the uniform coin operator, we have the expectation of $\tau^{(n)}$, 
\begin{scriptsize}
\begin{align*}
E(\tau^{(n)}) &= \sum_{\mathbf{y} \in \partial (F^{(n)})\cup \partial (F^{(n)'})}\sum_{\mathbf{e}_j \in out(\mathbf{y}) } \sum_{t=1}^{\infty}t P^{(n)}(\mathbf{0},\mathbf{y},\tau = t)^i_j \\
		   &= \sum_{\mathbf{y} \in \partial (F^{(n)})\cup \partial (F^{(n)'})}\sum_{\mathbf{e}_j \in out(\mathbf{y}) } \sum_{t=1}^{\infty} t \Psi (w_0 = \mathbf{0}^i , w_{t} = \mathbf{y}^j ,\tau = t) \\
		   & = \sum_{\mathbf{y} \in \partial (F^{(n)}\cup \partial (F^{(n)'})}\sum_{\mathbf{e}_j \in out(\mathbf{y}) }  (\partial_z g^{(n)}(z)(\mathbf{0},\mathbf{y})_{j}^{i}|_{z=1})
\end{align*}
\end{scriptsize}
Plug in the recursive formulas  just above Theorem \ref{th:classical hitting probabilities theorem},
we obtain  $E_0(\tau^{(0)}) = 1$, $E_0(\tau^{(1)}) = 3$, $E_0(\tau^{(2)}) = 3^2$, $E_0(\tau^{(3)}) = 3^3$, $E_0(\tau^{(4)}) =3^4$,... These results agree with Theorem \ref{th:recurrence exponent classical case}.

\end{document}